\documentclass[9pt,twocolumn,twoside]{IEEEtran}

\IEEEoverridecommandlockouts         
               
\usepackage{amsmath,amssymb,amsthm,latexsym,amsfonts,relsize}

\usepackage{paralist}
\usepackage{graphics} 
\usepackage{epsfig} 
\usepackage{epstopdf}
\usepackage{color}
\usepackage{bbm}
\usepackage{tikz}
\usetikzlibrary{arrows}

\usepackage{subfigure}
\usepackage{pgf}
\usepackage[labelsep=period]{caption}
\usepackage{hyperref}
\usepackage{url}
\usepackage{setspace,booktabs}
\usepackage{caption}
\usepackage{wrapfig}
\usepackage{lipsum}%
\usepackage{comment}

\usepackage{resizegather}

\newtheorem{thm}{Theorem}

\newtheorem{lem}{Lemma}

\newtheorem{assum}{Assumption}
\newtheorem{prop}{Proposition}

\newtheorem{defn}{Definition}

\newtheorem{rem}{Remark}

\DeclareMathOperator{\R}{\mathbb{R}}

\newcommand{\G}{{\mathcal{G}}}
\newcommand{\V}{{\mathcal{V}}}
\newcommand{\EE}{{\mathcal{E}}}

\DeclareMathOperator*{\argmin}{arg\,min}

\newcommand{\Sp}{\hspace{0.05cm}}

\newcommand{\ER}{\Xi_{\mathcal G}}

\graphicspath{{../figs/}}

\begin{document}

\title{Risk of Collision and Detachment in Vehicle Platooning:  Time--Delay--Induced Limitations and  Trade--offs \\ (Extended Version)}
\author{Christoforos Somarakis, Yaser Ghaedsharaf, Nader Motee
\thanks{*This work is supported by NSF CAREER ECCS-1454022, AFOSR FA9550-19-1-0004 and ONR YIP N00014-16-1-2645}
\thanks{C. Somarakis is with the System Sciences Lab, Palo Alto Research Center, Palo Alto, CA, 94304 {\tt \small somarakis@parc.com}. Y. Ghaedsharaf and N. Motee are with the Department of Mechanical Engineering and Mechanics, Lehigh University, Bethlehem, PA, 18015, USA.   \{\tt\small ghaedsharaf,motee\}@lehigh.edu}.}%

\maketitle

\begin{abstract}  We quantify the value-at-risk of  inter-vehicle collision and detachment for a class of  platoons, which are governed by second-order dynamics in presence of communication time-delay and exogenous stochastic noise. Closed-form expressions for the risk measures are obtained as  functions of Laplacian eigen-spectrum as well as their fine explicit  approximations using rational polynomial functions. We quantify several hard limits and fundamental trade-offs among the risk measures, network connectivity, communication time-delay, and statistics of exogenous stochastic noise. Simultaneous presence of stochastic noise and time delay in a platoon imposes some idiosyncratic behavior risk of collision and detachment, for instance, weakening (improving) network connectivity may result in lower (higher) levels of risk. Furthermore, a thorough risk analysis is conducted for networks with specific graph topology. We support our theoretical findings via multiple simulations.
\end{abstract}

\section{Introduction}\label{sect: INTRODUCTION}	
{Networked control systems are often susceptible to external disturbances and onboard hardware limitations (e.g. communication time-delay, limited computational power and battery life).  Notable examples of  such networks include platoon of autonomous vehicles, synchronous power networks with integrated renewable sources, water supply networks, transportation  networks, and inter-dependent financial systems \cite{Bamieh12,doi:10.1137/110851584,Tahbaz13,WRCR:WRCR4853,doi:10.1287/trsc.34.3.239.12300}. Exogenous disturbances usually steer trajectory of a  system  away from the target equilibrium and potentially into undesirable modes of operation. The situation  becomes even more challenging when the unperturbed  network consists of several  interconnected  subsystems, where the troublesome effects of noise propagate and get amplified across the network. Moreover, inherent limitations on the communication layer (e.g., time-delay, receiver and transmitter noise) may also exacerbate the effect of noise and  deteriorate the overall performance of the network.  One of the main engineering challenges is to design robust dynamical networks that damp, if not reject,  undesirable network-wide effects of disturbance and communication time-delay.
 
In this paper, we consider platoons vehicles that exchange information over a time-invariant communication network.  The platooning problem is a simple, yet rich, benchmark to study autonomy in robotic networks using their second- or third-order state-space models \cite{gazis2002traffic,250509,7963567,6515636,6683051}.  We assume vehicle dynamics to be represented by a double control integrator. Also vehicles are capable of receiving, transmitting and processing data to update their own state, according to a second-order consensus protocol. Stemming from real-world applications, vehicles suffer from non-negligible communication time-delays  due to deficiencies of existing hardware modules. It is assumed that all vehicles use identical hardware modules and, as a result, they all experience a uniform (identical) time-delay. In addition, the effects of uncertain surrounding environment on vehicles are modeled and incorporated into our network model via additive independent Gaussian force noises.    

The global objective of platooning is to guarantee the following two group behaviors in steady-state: (i) pair-wise difference between position variables converges to a prescribed distance, and (ii) the platoon of agents attain the same constant velocity. This is illustrated in  Figure  \ref{fig: the platoon}. We identify two types of undesirable events, also referred to as systemic events. It is crucial to identify a systemic event where at least two consecutive vehicles collide and calculate its
probability. We recall that a near-hit-region for two consecutive vehicles in platoon is an unsafe region in the state space of trajectories, as once there, vehicles may collide. This local event may interrupt the platoon and render the state of the entire network to unsafe regions in the state space. Similarly, no two consecutive vehicles may  stand off too remotely, as once there, the communication graph of the platoon may be. Identifying and accounting for such events becomes substantially challenging when platoon is subject to exogenous stochastic noise and communication time-delay.  

\vspace{0.1cm}
\noindent \textit{Related Literature:} Norm induced performance and robustness measures have been widely studied in the context of robust control \cite{Zhou:1996:ROC:225507}. Reference \cite{matt2004} gives an overview and a brief history of risk-sensitive stochastic optimal control and surveys various approaches to controller design as well as their relationships among each other. The control objective, in this context, is to synthesize a controller with satisfactory levels of performance in the presence of disturbances.  These methods face severe shortcomings when they are applied to stochastic dynamical networks. The resulting controllers usually require all-to-all communication and do not respect topology of the underlying  communication layer \cite{1556730}. Moreover, these controller design methods do not scale with network size.

The $\mathcal H_2$--norm has been recently utilized as a measure of performance and coherency for linear consensus  networks (see  \cite{Bamieh12,Siami16TACa,yaserecc16,JOVANOVIC2008528,7963567} and references therein). One of the main advantages of using $\mathcal H_2$--norm is its elegant representation in terms of Laplacian spectrum that makes development of  tractable and scalable network design  algorithms possible \cite{Siami16TACb}. An interesting interpretation of the $\mathcal H_2$ norm, in the context of platooning, is that it quantifies the ability of the entire platoon to withstand the effect of exogenous noise and remain a rigid body \cite{Bamieh12}.
The effects of time-delay in consensus networks have been investigated in various disciplines such as, to name only a few, traffic networks, flocking, distributed optimal control design \cite{gazis2002traffic,Bamieh12,4282756,BEXELIUS196813,RevModPhys.73.1067,SomBarifac15,Yu2010, doi:10.1137/060673813}. These works are mainly concerned with the problem of stability. Performance analysis and design of  time-delay linear consensus networks  using  $\mathcal{H}_2$-norm is recently studied in \cite{7963567,yaserecc16,yasernecsys16,PhysRevE.92.062816,PhysRevE.86.056114,PhysRevE.90.042135}. 

The $\mathcal H_2$--norm is an example of norms of Hardy-Schatten type. These systemic metrics quantify macroscopic features of networks. For instance, in consensus networks, $\mathcal H_2$--norm measures coherency \cite{Bamieh12} and $\mathcal{H}_{\infty}$--norm quantifies global connectivity \cite{Siami16TACb}. However, these measures cannot scrutinize microscopic behaviors of networks. The focus of this paper is to inspect risk of inter-vehicle collision or detachment in the platoon of vehicles. 
We build upon existing notions of risk that are widely used in the context of financial systems \cite{artzner97,artzner99}. In its rudimentary form, risk serves as a surrogate for uncertainty in stochastic models \cite{rockafellar07}. We utilize the notion of value-at-risk to measure the extend and occurrence of a random undesirable event, with a certain confidence level,  over a specified time period.



\noindent {\it Our Contributions:} Building upon our recent works on first order consensus systems \cite{sommiladnader16,somyasnader17,DBLP:journals/corr/abs-1801-06856}, we investigate aspects of fragility in the platooning model from a systemic risk perspective.  In Section \ref{sect: risk}, the value-at-risk measure is quantified to determine safety margins for the following two types of undesirable events: (i) inter-vehicle collision,  and (ii) inter-vehicle detachment. For a single collision or detachment event, we obtain a closed-form expression of risk in terms of Laplacian spectrum in Section \ref{sect: riskplatoon}.  In Section \ref{sect: jointrisk}, tight lower and upper bounds in terms of risk of individual events are provided for risk of  multiple (joint)  events. We outline the computational difficulties in deriving explicit expressions for the risk measures in Section \ref{sect: approx} and provide a tractable method for their approximation using rational functions. In Section  \ref{sect: tradeoff}, we prove that fundamental limits and tradeoffs emerge on the best achievable levels of risk, which are solely due to the presence of exogenous noise and communication time-delay. Furthermore, we show that strengthening (weakening) network connectivity results in higher (lower) levels of risk of inter-vehicle collision and detachment. Finally, in Section \ref{sect: topological}, we apply our approximate formulas to calculate risk of platoons with complete, path, and cyclic communication topologies and show how to identify high-risk vehicles in such networks. The paper concludes with Section \ref{sect: simulation} of simulation examples that outline the theoretical results and Appendix sections with proofs of technical results.

Risk analysis of the platooning problem in this paper differs in almost every building block from our earlier works on first-order linear consensus networks \cite{sommiladnader16,somyasnader17,DBLP:journals/corr/abs-1801-06856}. First, the time-delayed second-order linear consensus networks are inherently more perplexed. Therefore, the stability conditions, that serve as cornerstone of our results, require fundamentally different analytic approach. Second, the nature of systemic events are different, which in turn, result in new risk formulas that do not lend themselves to explicit expressions in terms of Laplacian spectrum.  The present paper is an outgrowth \cite{somyasnader18} in several different aspects.  First, we extend our risk analysis to inter-vehicle detachment events as well as multiple systemic events.  We propose a tractable rational function approximation scheme for evaluation of risk. Furthermore, we examine special communication topologies and present extra simulation examples. The manuscript also contains detailed proofs of all our technical results.

}

\section{Preliminaries}\label{sect: intro}
The $n$-dimensional Euclidean space with elements $\mathbf z=[\Sp z_1,\Sp \dots \Sp,z_n \Sp]^T$ is denoted by $\mathbb R^n$, where $\mathbb R^n_+$ will denote the positive orthant of $\mathbb R^n$. We denote the vector of all ones by $\mathbbm 1_n$. For every $\mathbf z_1,\mathbf z_2 \in \R^n$, we write $\mathbf z_1 \preceq  \mathbf z_2~$ if and only if $ \mathbf z_2-\mathbf z_1 \in \mathbb R_+^n$.  The set of standard Euclidean basis for $\R^n$ is represented by $\{\mathbf e_1,\ldots, \mathbf e_n \}$. The vector or induced matrix $2$-norms are represented by $\| \cdot \|$. The vector of all ones is denoted by $\mathbbm{1}_n= [1,\dots,1]^T$. 

\vspace{0.2cm}
\noindent {\it Algebraic Graph Theory:} 
 A weighted graph is defined by $\mathcal G=(\V, \mathcal{E},w)$, where $\V$ is the set of nodes, $\mathcal{E}$ is the set of links (edges),  and $w: \mathcal{V} \times \V \rightarrow \mathbb{R}_{+}$ is the weight function that assigns a non-negative number to every link.  Two nodes $i,j\in \mathcal V$ are directly connected if and only if $(i,j)\in \mathcal E$. The set of nodes adjacent to $i$ constitutes the neighborhood of node $i$ that is denoted by $\mathcal N_i=\big\{j \in \mathcal V ~\big|~(i,j)\in \mathcal E \big\}$. 
\begin{assum}\label{assum0} Every graph $\mathcal G=(\V, \mathcal{E},w)$ in the paper is connected. In addition for every $i,j \in \mathcal{V}$, the following properties hold:

\noindent (i) \hspace{0.075in}~$w(i,j)>0$ if and only if $(i,j) \in \mathcal E$;

\noindent (ii) \hspace{0.04in}~$w(i,j)=w(j,i)$, i.e., links are undirected;

\noindent (iii)~ $w(j,j)=0$, i.e., links are simple.

\end{assum} 
The Laplacian matrix of $\G$ is a $n \times n$ matrix $L=[l_{ij}]$ with elements  
\begin{equation}\label{eq: laplacian}
\displaystyle l_{ij} := \left\{\begin{array}{ccc}
-k_{ij} & \textrm{if} & i \neq j \\
 &  &  \\
k_{i1}+\ldots+k_{in}& \textrm{if} & i=j
\end{array}\right.,
\end{equation} 
where $k_{ij}:=w(i,j)$. Laplacian matrix of a graph is  symmetric and positive semi-definite. Assumption  \ref{assum0} implies that the smallest Laplacian eigenvalue is zero with algebraic multiplicity one. The spectrum of $L$ can be ordered as \begin{equation}\label{eq: spectrum}0=\lambda_1 < \lambda_2 \leq \ldots \leq \lambda_n.\end{equation} The eigenvector of $L$ corresponding to  $\lambda_k$  is denoted by $\mathbf q_k$. By letting $Q=[\Sp \mathbf q_1~|~\dots~|~\mathbf q_n \Sp]$, it follows that $L=Q \Lambda Q^T $ with $\Lambda=\mathrm{diag}[\Sp 0,\Sp \lambda_2,\Sp \ldots \Sp, \Sp\lambda_n \Sp]$. We normalize the Laplacian eigenvectors such that $Q$ becomes an orthogonal matrix, i.e., $Q^T Q = Q Q^T =I_n$ with  $\mathbf q_1=\frac{1}{\sqrt{n}} \mathbbm{1}_n$. The total effective resistance of $\mathcal G$ is a popular metric of connectivity  \cite{Mieghem:2011:GSC:1983675} that is characterized as  \cite{klein1993} 
\begin{equation}\label{eq: graphr}
\ER =n\sum_{i=2}^n \lambda_i^{-1}.
\end{equation} The smaller the value of $\Xi_{\mathcal G}$, the stronger the connectivity of $\mathcal G$.

\vspace{0.2cm} 
\noindent {\it Probability Theory:}  Let $\mathcal L^2(\mathbb R^q)$ be the set of all $\mathbb R^q$-valued random vectors $\mathbf z=\big[z^{(1)},\dots,z^{(q)}\big]^T$ of a probability space $(\Omega,\mathcal F, \mathbb P)$ with finite second moments. A normal random variable $y \in \mathcal L^2(\mathbb R^q)$ with mean value $\mu\in \mathbb R^q$ and $q\times q$ covariance matrix $\Sigma$ is represented by $y~\sim~\mathcal N(\mu,\Sigma)$. The error function $\mathrm{erf}(x):\mathbb R\rightarrow (-1,1)$ is 
$$\mathrm{erf}(x)=\frac{2}{\sqrt{\pi}}\int_{0}^{x}e^{-t^2}\,dt$$
which is invertible on its range. The complementary error function is $\text{erfc}(x)=1-\text{erf}(x)$. Formulation of stochastic differential equations, we employ standard notation $d\boldsymbol \xi_{t}$ for stochastic differentials\footnote{This formalism will also be retained for deterministic differentials.}. 


\section{Problem Formulation}\label{sect: formulation}
Suppose that a finite number of vehicles  $\mathcal V=\{1,\dots,n\}$ form a platoon along the horizontal axis. Vehicles are labeled in descending order, where the $n$'th vehicle is assumed to be the leader.  The $i$'th vehicle's state is determined by $[x^{(i)},v^{(i)}]^T$, where $x^{(i)}$ is the position and $v^{(i)}$ is the velocity of vehicle $i \in \mathcal V$. The $i$'th vehicle's state evolves in time according to the following stochastic  differential equation 
\begin{equation}\label{eq: model0}\begin{split}
d x_t^{(i)} & = v_t^{(i)}\,dt \\
d v_t^{(i)} & = u_t^{(i)}\,dt+g\,d\xi_t^{(i)}
\end{split}
\end{equation}
where $u_t^{(i)} \in \mathbb R$ is the control input at time $t$. The terms $g \, d\xi_t^{(i)},~i=1,\dots,n$ represent white noise generators affecting  dynamics of the vehicle and models the uncertainty diffused in the system. It is assumed that noise acts on every vehicle additively and independently from the other vehicles' noises. The noise magnitude is represented through diffusion $g\neq 0$, assumed identical for all $i\in \mathcal V$. The control objectives for the platoon are to guarantee the following
two global behaviors: (i) pair-wise difference between position variables of every two consecutive vehicles converges to zero; and (ii) the platoon of vehicles attain the same constant velocity in steady state. It is known that the following feedback control law can achieve these objectives
\begin{eqnarray}
\hspace{-0.3cm}u_t^{(i)} & = & \sum_{j=1}^n ~k_{ij}~\left(v^{(j)}_{t-\tau}-v_{t-\tau}^{(i)}\right)  \nonumber\\
& & \hspace{1cm} +~  \beta ~\sum_{j=1}^n ~k_{ij}~\left(x^{(j)}_{t-\tau}-x_{t-\tau}^{(i)}-(d_j-d_i)\right).\label{eq: feedback}
\end{eqnarray}

\begin{figure}[t]
\center
\includegraphics[scale=1]{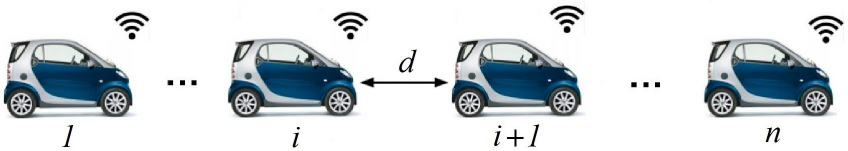}
 \caption{Schematics of platooning. Each vehicle aims to travel in a constant velocity while preserving distance $d$ from its immediate neighbors. Vehicles exchange their state information over a constant communication network, in presence of noise and time-delay.}\label{fig: the platoon}
 \vspace{-0pt}
\end{figure}

Let us denote the communication graph by $\G=(\V,\mathcal E,w)$, where $\V=\{1,\dots,n\}$, $(i,j) \in \EE$ iff $k_{ij}>0$, and $w(i,j)=k_{ij}$ for all $i,j \in \V$. The feedback gains $k_{ij}\geq 0$ are designed so that the resulting communication graph $\G$ with Laplacian matrix \eqref{eq: laplacian} satisfies Assumption \ref{assum0}.  The constant $\tau\geq 0 $ is the communication time-delay. The first term in \eqref{eq: feedback} guarantees the control objective (i). The second term in \eqref{eq: feedback} guarantees the control objective (ii) and stabilizes the relative position of vehicles $i$ and $j$ around the distance $d_j-d_i$. The control parameter $\beta>0$ balances the effect of the relative positions and velocities. Following the scenario described in Figure \ref{fig: the platoon}, we select $d_i=i d$ for some parameter $d>0$ and all $i=1,\dots,n$. The target distance between vehicles $i+1$ and $i$ becomes $d_{i+1}-d_i = d$. Let us define  the vector of positions, velocities, and noise inputs as $\mathbf x_t=[x_t^{(1)},\dots,x_t^{(n)}]^T$, $\mathbf v=[v_t^{(1)},\dots,v_t^{(n)}]^T$,  and $\boldsymbol{\xi}_t= [ \xi^{(1)}_t, \dots, \xi_t^{(n)}]^T$, respectively\footnote{The stochastic process $\{\boldsymbol{ \xi}_t\}_{t\geq 0}$ with  $\boldsymbol{ \xi}_t=\big[\xi_t^{(1)},\dots,\xi_t^{(q)}\big]^T\in \mathcal L^2(\mathbb R^q)$ denotes an $\mathbb R^q$-valued Brownian motion.}. By 
applying the feedback control law \eqref{eq: feedback} to \eqref{eq: model0} and denoting $\mathbf d = d\mathbbm 1_n$, the closed-loop dynamics can be cast as the following initial value problem:
\begin{equation}\label{eq: sys0}\begin{split}
d\mathbf x_t&=\mathbf v_t \Sp d_t\\
d\mathbf v_t&=-L \Sp \mathbf v_{t-\tau} \Sp dt \Sp - \Sp \beta L \Sp \big(\mathbf x_{t-\tau}- \mathbf d) \Sp dt \Sp + \Sp g \Sp d\boldsymbol{\xi}_t
\end{split}
\end{equation} for all $t\geq 0$ and given deterministic initial functions $\boldsymbol \phi^{\mathbf x}(t), \boldsymbol \phi^{\mathbf v}(t) \in \R^n$ over $t \in [-\tau,0]$. Standard results in the theory of stochastic functional differential equations \cite{Mohammed84} guarantee that \eqref{eq: sys0} generates a well-posed stochastic process $\{(\mathbf x_t,\mathbf v_t)\}_{t\geq -\tau}$.

The {\it problem} is to quantify risk of systemic events as a function of communication graph $\G$, time-delay, and statistics of noise. The systemic events are undesirable events that lead to greater (negative) impact on the global behavior of the platoon. We consider two types of systemic events: inter-vehicle collision, where two successive vehicles get too close to each other, and inter-vehicle detachment, where two successive vehicles distance themselves too far from their target distance $d$ and lose connectivity\footnote{This is practically relevant as vehicles are usually equipped with communication modules with limited range.}.

 
 \begin{figure}
\tikzstyle{block} = [draw, fill=blue!20, rectangle, 
    minimum height=3em, minimum width=6em]
\tikzstyle{sum} = [draw, fill=blue!20, circle, node distance=1cm]
\tikzstyle{input} = [coordinate]
\tikzstyle{output} = [coordinate]
\tikzstyle{pinstyle} = [pin edge={to-,thin,black}]
%
\begin{tikzpicture}[auto, node distance=2cm,>=latex']
    \node [input, name=input] {};
    \node [sum, right of=input] (sum){+};
    \node [block, right of=sum,pin={[pinstyle]above: Time-Delay}] (controller) {Controller};
    \node [block, right of=controller, pin={[pinstyle]Stochastic Disturbance},
            node distance=3.5cm] (system) {Vehicle $i$};
    \draw [->] (controller) -- node[name=u] {$u_t^{(i)}$} (system);
    \node [output, right of=system] (output) {};
    \draw (1.0,1.1) node [block] (measurements) {Network};
    \coordinate [below of=u] (tmp);
    \coordinate [above of=system] (tmp2);
    \draw [draw,->] (input) -- node{\tiny Reference} (sum);
    \draw [->] (sum) -- node {} (controller);
    \draw [->] (system) --  node [name=y] {} (output);
    \draw[->] (output) |-(tmp2)-| (measurements);
    \draw [->] (measurements) -- (sum);
    \draw [->] (output) |- (tmp) -| (sum);
        \draw [->] (output)  |- (tmp) -| node[pos=0.99]{$-$} 
        node [near end] {} (sum);  
\end{tikzpicture}

\caption{Block diagram of  vehicles' feedback mechanism. The propagation and process of information is prone to time-delay and exogenous disturbances. Vehicle $i$ receives the state information of the neighboring vehicles (network block), while it communicates its own state information back to them.}\label{fig: controller}
\vspace{-0pt}
\end{figure}
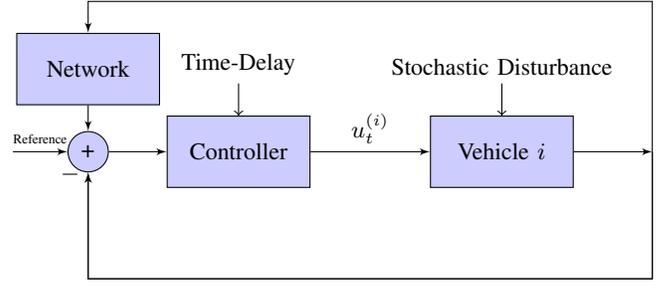

%

\section{Stability and Solution Statistics}\label{sect: prelim} We investigate stability of the unperturbed closed-loop network, i.e.,  when $g=0$ in  \eqref{eq: sys0}. Then, we analyze the statistical properties of the solution of \eqref{eq: sys0} when $g\neq 0$. 
\begin{defn}\label{def: platooning}
The solution of network \eqref{eq: sys0}  converges to a platoon if 
\begin{equation*}\label{eq: convergence}
\lim_{t\rightarrow\infty}\left|v_t^{(i)}-v_{t}^{(j)}\right|=0\hspace{0.1in}~\text{and}~\hspace{0.1in}\lim_{t\rightarrow\infty}\left|x_t^{(i)}-x_{t}^{(j)}-(i-j)d\right|=0
\end{equation*}
for all $i,j\in \mathcal{V}$ and all initial functions  $\boldsymbol \phi^{\mathbf x}(t), \boldsymbol \phi^{\mathbf v}(t)$ over $[-\tau,0]$.  
\end{defn} 
Decomposition of $L=Q\Lambda Q^T$ offers a useful transformation 
 \begin{equation*}
 \mathbf z_t:=Q^T\big(\mathbf x_t-\mathbf d\big)~~\textrm{and}~~ \boldsymbol \upsilon_t :=Q^T \mathbf v_t.
 \end{equation*} In this new coordinates, \eqref{eq: sys0} transforms to 
\begin{equation}\label{eq: sys2}\begin{split}
d\mathbf z_t& \Sp = \Sp \mathbf u_t\, d_t\\
d\boldsymbol \upsilon_t& \Sp = \Sp -\Lambda \Sp \boldsymbol \upsilon_{t-\tau}\Sp dt \Sp - \Sp \beta \Sp \Lambda \Sp \mathbf z_{t-\tau}\,dt \Sp + \Sp g \Sp  Q^T \Sp  d\boldsymbol {\xi}_t.
\end{split} 
\end{equation} 

\subsection{Exponential Stability of the Unperturbed System}

\noindent Consider the set 
\begin{equation}\label{eq: S}\begin{split}
S=\bigg\{&(s_1,s_2)\in \mathbb R^2~\bigg|~ s_1\in \bigg(0,\frac{\pi}{2}\bigg),~ s_2 \in \bigg(0,\frac{a}{\tan(a)}\bigg),~~~~\\~~~~&~~\text{for}~a\in \bigg(0,\frac{\pi}{2}\bigg)~\text{the solution of}~~ a\sin(a)=s_1 \bigg\}.
\end{split}
\end{equation}  
It is noted that so long as $s_1\in (0,\frac{\pi}{2})$ there equation $a\sin(a)=s_1$ is true for a unique $a\in (0,\frac{\pi}{2})$. Figure \ref{fig: stability} illustrates geometry of  $S$. 

\begin{thm}\label{thm: main0} The solution  of network \eqref{eq: sys0} with $g=0$ converges to a platoon if and only if $$(\lambda_i\tau,\beta\tau) \in S~~~~\text{for all}~~~~i=2,\dots,n.$$
\end{thm}

\begin{rem} Theorem \ref{thm: main0} also holds for $\tau = 0$. In the absence of time-delay, the decomposed characteristic polynomial of \eqref{eq: sys2} reads $\Delta_i(s)=s^2+\lambda_i\,s+\lambda_i\beta$ with roots $\frac{\lambda_i}{2}\pm\frac{1}{2}\sqrt{\lambda_i^2-4\lambda_i\, \beta}$ that both lie on the left half plane.  We note however that conclusion of Theorem \ref{thm: main0} seizes to hold when $\beta=0$. In this case, vehicles will align their velocities but not their relative distances.
\end{rem}

\begin{rem}
Stability properties of noise-free system \eqref{eq: sys0} were previously analyzed in \cite{Yu2010}. Theorem \ref{thm: main0} proposes conditions that are more suitable for the analysis to follow. 
\end{rem}

\begin{figure}\center
\includegraphics[trim=20 0 30 10, clip,width=8cm]{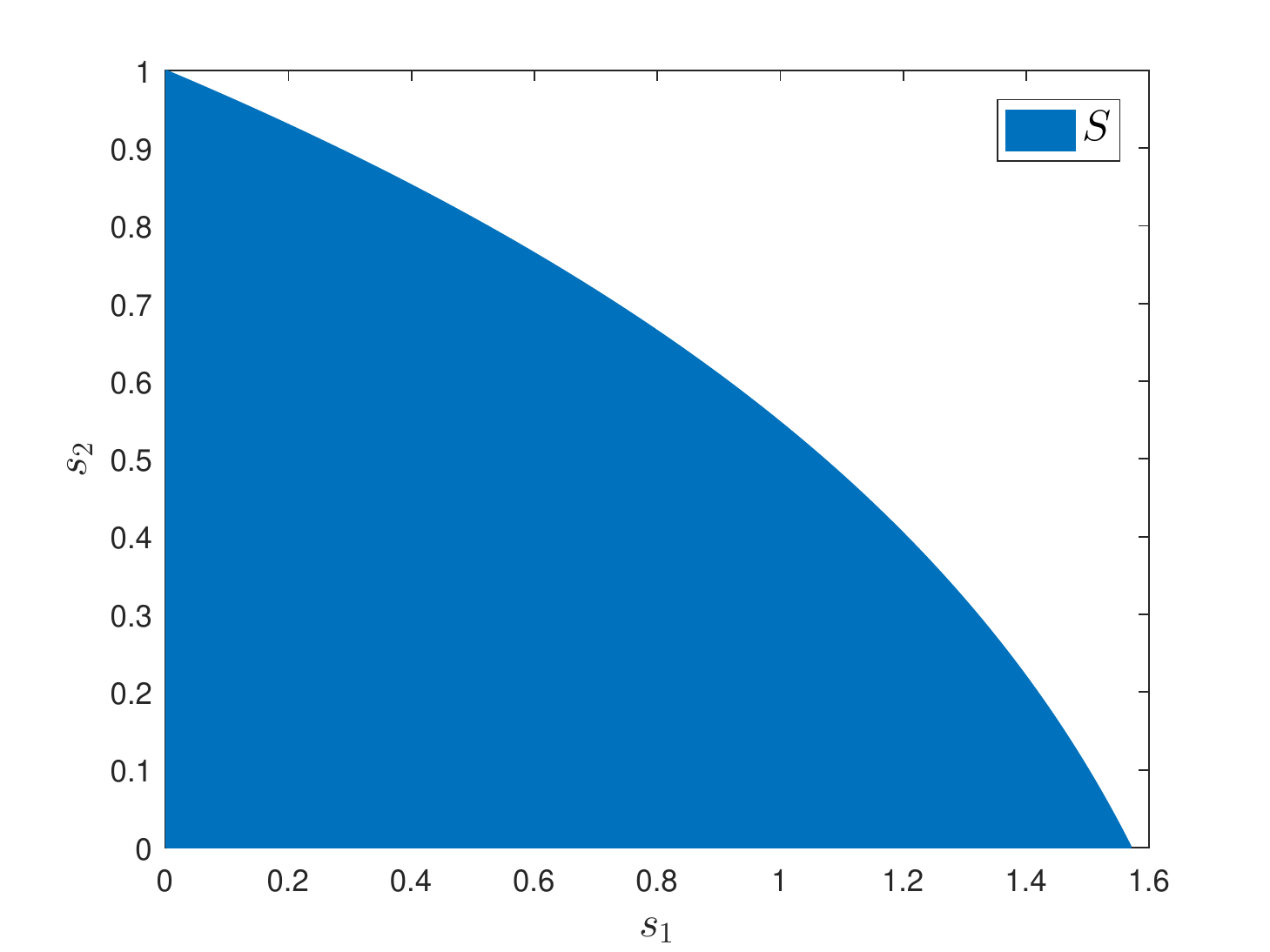}
\caption{The blue area is stability region $S$ as defined in \eqref{eq: S}.}\label{fig: stability}
\end{figure}

\subsection{Steady-State Statistics of Positions and Velocities}
System \eqref{eq: sys2} can be decomposed into two-dimensional subsystems with state variables $\big[z_t^{(i)},\upsilon_t^{(i)}\big]^T$ for $i=1,\dots,n$. We apply the variation of parameters formula for It\^{o} calculus \cite{Mohammed84} to express the solution of the decoupled subsystems 
\begin{equation*}
\begin{bmatrix}
z_t^{(i)}\\
\upsilon_t^{(i)}
\end{bmatrix}=
\Phi_i(t)\begin{bmatrix}
z_0^{(i)}\\
\upsilon_0^{(i)}
\end{bmatrix}+\int_{-\tau}^0\Phi_i(t-s)\mathbf r_s\,ds+g\int_0^t\Phi_i(t-s)B_i \,d\boldsymbol \xi_s.
\end{equation*} 
According to the conditions of Theorem \ref{thm: main0},  the principal solution $\Phi_i(t)$ of the unperturbed system, see \eqref{eq: sys3}, is exponentially decaying (stable) with respect to the consensus equilibrium. The vector $\mathbf r_s=\mathbf r(z_s^{(i)},\upsilon_s^{(i)})=[0 \Sp, \Sp-\lambda_i\beta z_s^{(i)} -\lambda_i \upsilon_s^{(i)}]^T$ depends on the initial functions and
$B_i=\big[
\mathbf 0_{1\times n} \Sp, \Sp
\mathbf q_i^T
\big]^T$.  For each $i \in \{2,\dots,n\}$, the process $\big\{\big[z^{(i)}_t,\upsilon^{(i)}_t\big]^T\big\}_{t\geq -\tau}$ is well-defined and as $t\rightarrow \infty$ it converges, in distribution, to the bi-variate normal $\mathcal N\big(\mathbf 0,\Sigma_{\infty}^{(i)}\big)$ with covariance \begin{equation}\label{eq: cov}\Sigma_{\infty}^{(i)}:=g^2 \int_0^{\infty}\Phi_i(s)\begin{bmatrix}
0 &  0\\
0 & 1 
\end{bmatrix}\Phi_i^T(s)\,ds.\end{equation} 
The steady-state statistics, which are free from the transient effects of initial functions, carry the effects of the persistent network features, i.e., the communication topology, time-delay, and  statistics of the exogenous uncertainties. Explicit calculation of $\Sigma_\infty^{(i)}$ is neither feasible nor useful. We are interested in studying events that are related to the relative distance between vehicles. We are thus only concerned with the marginal statistics of the above system, rather than the full state, i.e., the statistics of $\overline{z}^{(i)}:=\lim_{t\rightarrow \infty}z_t^{(i)}$ for all $i \in \{2,\dots,n\}$. 
 
\begin{lem}\label{thm: main1} Suppose that  conditions of Theorem \ref{thm: main0} hold.  Then, $$\overline{z}^{(i)}\sim \mathcal N\bigg( 0 \Sp ,\Sp \frac{g^2\tau^3}{2\pi} \Sp f\big(\lambda_i\tau,\beta\tau\big) \bigg), $$  for all $i=2,\dots,n$, where  $f:S\rightarrow \mathbb R$ is defined as
\begin{equation}\label{eq: functionf}
f\big(s_1,s_2\big)=\int_{\mathbb R}\frac{d r}{\big(s_1 s_2-r^2 \cos(r)\big)^2+r^2\big(s_1-r \sin(r)\big)^2}.
\end{equation}  
\end{lem} 
The result implies that $\big\{\overline{z}^{(i)}\big\}_{i \in \{ 1,\dots, n-1\}} \in \mathcal L^2(\mathbb R)$ if and only if  $(\lambda_{i+1}\tau,\beta\tau )\in S$. The function $f$ is well-defined in the stability region $S$, however, it cannot be calculated in an explicit form. To address this challenge, we propose in Section \ref{sect: approx} an efficient rational approximation of function $f$. 

\section{Value-At-Risk Measures}\label{sect: risk}  

The standard deviation of a random variable in $\mathcal L^2(\mathbb R)$ is one the common ways to quantify the uncertainty level encapsulated in that random variable. The notion of risk provides a more comprehensive and meticulous way to measure uncertainty in a random variable. Risk measures are defined either in terms of moments of a  random variable or its distribution  \cite{rockafellar07,follmer11}. In this work, we focus on the latter type of risk, known as value-at-risk measures\footnote{For moment-based risk analysis of a class of linear dynamical networks, we refer to \cite{DBLP:journals/corr/abs-1801-06856} for more details.}. These risk measures quantify the manner with which the uncertainty, nested in a random variable, steers its realization close to some undesirable range of values. Let us denote the set of undesirable values, which is also referred to as the set of systemic events,   by $U\subset \mathbb R$. Then, the higher the risk on a random variable, the more likely that random variable will approach $U$. 

\begin{defn}\label{def: event} Let $(\Omega, \mathcal F, \mathbb P)$ be a probability space, $y: \Omega\rightarrow \mathbb R$ and $U\subset \mathbb R$.  The set of systemic events of $y$ is $\big\{\omega \in \Omega ~|~ y(\omega) \in U\big\}\in \mathcal F $.
\end{defn}
We will evaluate the risk of $\big\{\omega \in \Omega ~|~ y(\omega) \in U\big\}$ leveraging the distribution of $y$. The idea is to construct a set structure to measure the distance of $y$ from $U$. Then, the risk of systemic events for $y$ will be defined on the basis of this structure. 

\subsection{Value-at-Risk of Scalar Events} 
Suppose that it is desirable for random variable $y\in \mathcal L^2(\mathbb R)$ to stay away from  the set $U$. Let us consider a collection of super-sets $\{U_{\delta}\}_{\delta\in \mathbb R_+}$ of $U$ with the following properties:
\vspace{0.2cm}
\begin{enumerate}
\item[($\Pi_1$)] $U_{\delta_1}\subset U_{\delta_2}$ when  $\delta_1>\delta_2$ 
\item[($\Pi_2$)] $\lim_{n\rightarrow\infty} U_{\delta_n}=\bigcap_{n=1}^{\infty}U_{\delta_n}=U$ for any sequence $\{\delta_{n}\}_{n=1}^{\infty}$ with property  $\lim_{n\rightarrow\infty}\delta_n=\infty$.
\end{enumerate}
\vspace{0.2cm}
The collection $\{U_{\delta}\}_{\delta\in \mathbb R_+}$ can be further shaped  to cover a suitable vicinity of $U$. This vicinity plays the role of an alarm zone and yields high risk indexes as $y$ approaches $U$.  For some specific  $\delta >0$, the occurrence of $\big\{ y\in U_{\delta}\big\}$ signifies how close $y$ can  get to $U$ in probability. The idea is implemented with the use of quantile functions on the systemic events of $y$. For a given parameter $\varepsilon \in (0,1)$, the risk measure $\mathcal R_{\varepsilon}~:~\mathcal F \rightarrow \mathbb R_+$ is defined by
\begin{equation}\label{eq: risk0}
\mathcal R_{\varepsilon}=\inf\Big\{\delta>0~\Big|~\mathbb P \big\{y \in U_\delta \big\} <  \varepsilon \Big\}.
\end{equation} The number $\varepsilon$ is the cut-off value that characterizes the level of confidence on the systemic events. The smaller its value, the higher the confidence of the index $\mathcal R_{\varepsilon}$.  Let us elaborate and  interpret what typical values of $\mathcal R_{\varepsilon}$ imply. The case  $\mathcal R_{\varepsilon}=0$ signifies that the probability of observing $y$ dangerously close to $U$ is less than $\varepsilon$. We have  $\mathcal R_{\varepsilon}>0$ iff $y\in U_\delta$ for some $\delta>0$ (in fact, $\delta>\mathcal R_{\varepsilon}$) with probability greater than $\varepsilon$. The extreme case $\mathcal R_{\varepsilon}=\infty$ means that the event that $y$ is to be found in $U$ is assigned a probability greater than $\varepsilon$. In addition to several interesting properties (see for instance  \cite{somyasnader17,follmer11,rockafellar07}), the risk index \eqref{eq: risk0} is  non-increasing with $\varepsilon$.
\begin{prop}\label{prop: monotonicityofrisk} Let $\varepsilon_1,\varepsilon_2\in (0,1)$ and consider the set of undesirable values $U$ together with the collection of supersets  $\{U_\delta\}_{\delta\in \mathbb R_+}$ that satisfy properties $\Pi_1$ and $\Pi_2$. Then,
\begin{equation*}
\mathcal R_{\varepsilon_1}<\mathcal R_{\varepsilon_2} ~~~ \text{if and only if} ~~~ \varepsilon_2<\varepsilon_1.
\end{equation*}
\end{prop}

The motivation for risk in terms of quantile functions \eqref{eq: risk0} emanates from the fact that \eqref{eq: sys0} admits stochastic dynamics with random variables in $\mathcal L^2$ with tractable distributions. It is then desirable to monitor the stochastic volatility of desired observables w.r.t. to a specific subset of $\mathbb R$. 

\subsection{Value-at-Risk of Vector of  Events} 

For the case of random vectors $\mathbf y\in \mathcal L^2(\mathbb R^q)$, we first extend the notion of super-sets by considering the product set $U_{\boldsymbol \delta}=U_{\delta_1}\times \cdots \times U_{\delta_{q}}$, where $\boldsymbol \delta = [\delta_1,\dots,\delta_q]^T$. Similar to the scalar case, each sequence $\big\{U_{\delta_i} \big\}_{\delta_i \in \R_+}$ is assumed to satisfy properties $\Pi_1$ and $\Pi_2$. The multi-dimensional extension of Definition \ref{def: event} includes systemic events constructed through combination of set operations. One scenario, for example, is through the union operation 
$$
\left\{ \omega\in \Omega ~\Bigg|~ \bigcup_{i=1}^q \left\{y^{(i)}(\omega)\in U,~i=1,\dots,q \right\}\right\}\in \mathcal F.$$ 
In this case, the associated risk measure becomes \begin{equation}\label{eq: globalvectorrisk}\begin{split}
\boldsymbol{\mathcal R}_{\varepsilon}&=\inf\bigg\{\boldsymbol{\delta}\in \mathbb R^{q}_+ ~\bigg |~  \mathbb P\bigg( \bigcup_{i=1}^{q} \left\{ y^{(i)} \in U_{\delta_i} \right\} \bigg)<\varepsilon  \bigg\}.
\end{split}
\end{equation} A moment of reflection on \eqref{eq: globalvectorrisk} reveals that their calculation requires treating multivariable distributions. Unfortunately, these are rarely expressed in closed form. A computationally efficient surrogate is the vector of scalar risks, defined as \begin{equation}\label{eq: vectorrisk_s}
\mathfrak{R}_{\varepsilon}=\big[\mathcal R_\varepsilon^{(1)},\mathcal R_\varepsilon^{(2)},\dots,\mathcal R_\varepsilon^{(q)}\big]^T,
\end{equation}
where  
\[ \mathcal R_{\varepsilon}^{(i)}=\inf \Big\{ \delta>0~\Big|~ \mathbb P\big\{y^{(i)}\in U_{\delta}\big\}<\varepsilon\Big\}\] 
for $i=1,\dots,q$. The vector $\mathfrak{R}_{\varepsilon}$ is a collection of the scalar risk measures based on the individual distributions of $y^{(i)}$ for $i=1,\dots,q$. In section \ref{sect: jointrisk}, we revisit this part by formulating several vector of systemic events that are of interest to risk analysis of the platooning problem. Furthermore, we investigate their relations with the more computationally tractable vector $\mathfrak{R}_{\varepsilon}$.

\section{Risk Of Single Systemic Events In The Platoon} \label{sect: riskplatoon}
We consider systemic events of inter-vehicle  collision and detachment between two successive vehicles in the platoon. 
Let us represent the relative position of vehicles $i$ and $i+1$ in steady-state by the random variable $$\lim_{t\rightarrow \infty} \left( x_t^{(i+1)}-x_t^{(i)} \right) \Sp = \Sp \overline{x}^{(i+1)}-\overline{x}^{(i)},$$ 
where its statistics can be directly inferred from those of $\overline{\mathbf z}$. 
\begin{thm}\label{cor: main1} Let the conditions of Theorem \ref{thm: main1} hold. Then, the steady-state solution of \eqref{eq: sys0} satisfies 
$$\overline{x}^{(i+1)}-\overline{x}^{(i)} \sim \mathcal N\big(d,\sigma^2_i\big)$$
with 
$$\sigma^2_i \Sp = \Sp g^2 \Sp \frac{\tau^3}{2\pi} \Sp \sum_{j=2}^n \Sp \left([\mathbf e_{i+1}-\mathbf e_i]^T\mathbf q_j\right)^2 \Sp f\big(\lambda_j\tau,\beta\tau\big)$$ for $i=1 ,\dots, n-1$ and $f(s_1,s_2)$ as in \eqref{eq: functionf}. 
\end{thm}
We utilize the result of this theorem to calculate risk of inter-vehicle collision and detachment for two successive vehicles.

\subsection{Inter-Vehicle Collision} 

In absence of noise and under the conditions of Theorem \ref{thm: main0}, the steady-state distance between two successive vehicles satisfy $\overline{x}^{(i+1)}-\overline{x}^{(i)}\equiv d$.
The differential term $g\,d\boldsymbol \xi_t$ forces $x^{(i+1)}_t-x^{(i)}_t$ to fluctuate around $d$. Vehicles $i+1$ and $i$ experience a collision if $x_t^{(i+1)}=x_{t}^{(i)}$ at some $t > 0$.  If $x_t^{(i+1)}-x_{t}^{(i)}<0$, then the collision has already occurred at some time prior to  $t$. When  $x^{(i+1)}_t-x^{(i)}_t$ is positive, but close to zero, this will be an alert for a near collision. Let us define $(0,d/c)$ to be the zone of potential  collisions, where parameter $c\geq 1$ models the fluctuation magnitude and   determines  the upper endpoint of the collision set. For example, the collision set $c=1$ is $(0,d)$, which  allows no tolerance  for any deviation from the target distance $d$ between the two consecutive vehicles. For $c=2$, the collision set $(0,d/2)$ implies that there is no risk involved with the relative distance between two consecutive vehicles being in $[d/2,d]$. In summary, we can use the steady-steady statistics and infer whether vehicles $i$ and $i+1$ have already collided or are dangerously close to each other if $\overline{x}^{(i+1)}-\overline{x}^{(i)}\in (-\infty,0] $ or $\overline{x}^{(i+1)}-\overline{x}^{(i)}\in (0,d/c)$, respectively. The union of the two disjoint sets define the family of parameterized events 
\begin{equation*}\label{eq: collisionevent}
\bigg\{ \overline{x}^{(i+1)}-\overline{x}^{(i)} \in  C_{\delta}  \bigg\} \hspace{0.2in} \text{where} \hspace{0.2in} C_{\delta} =\bigg(-\infty~,~\frac{d}{\delta+c} \bigg)
\end{equation*} for $\delta\in \mathbb R_+$. It is straightforward to verify that the collection $\{C_\delta\}_{\delta\geq 0}$ satisfies properties $\Pi_1$ and $\Pi_2$. Then, the associated risk measure is defined as\footnote{\noindent We read $\mathcal R_{\varepsilon}^{C,i}$ as the risk of collision between vehicles $i+1$ and $i$ with confidence level $\varepsilon$. Similar notation is used for detachment risk.} 
\begin{equation}\label{eq: risk}
\mathcal R_{\varepsilon}^{C,i}=\inf \Big\{\delta>0 ~\Big|~\mathbb P \left\{\overline{x}^{(i+1)}-\overline{x}^{(i)}\in C_\delta \right\}<\varepsilon \Big\}
\end{equation}
for confidence level $\varepsilon\in (0,1)$. In other words, $\mathcal R_{\varepsilon}^{C,i}$ is the safety margin below which the likelihood of a past collision or a new one to be developed  is less than $\varepsilon$. The larger the value of risk, the higher the probability of the two vehicles being vulnerable to a collision. 

\begin{thm}\label{thm: main2} Suppose that the conditions of Theorem \ref{thm: main0} hold. For every $i \in \{1,\dots,n-1\}$, the risk of inter-vehicle  collision is
\begin{equation}\label{eq: riskcolformula}
\mathcal R_{\varepsilon}^{C,i}=\begin{cases}
0~\hspace{0.8in}& \text{if}~~~ \sigma_{i}\leq \frac{d}{\kappa_{\varepsilon}\sqrt{2}}\frac{c-1}{c}~~\text{or}~~\varepsilon\geq \frac{1}{2}\\
\frac{d}{d-\kappa_\varepsilon \sigma_i \sqrt{2}}-c~& \text{if}~~~\frac{d}{\kappa_{\varepsilon}\sqrt{2}}\frac{c-1}{c}<\sigma_i<\frac{d}{\kappa_{\varepsilon}\sqrt{2}}\\
\infty~ & \text{if}~~~\sigma_i\geq \frac{d}{\kappa_{\varepsilon}\sqrt{2}}
\end{cases}
\end{equation}
where $\sigma_i^{2}$ is as in Theorem \ref{cor: main1} and $\kappa_\varepsilon:=\text{{\normalfont erf}}^{-1}(1-2\varepsilon)> 0 $.
\end{thm}

The two extreme values of $\mathcal R_{\varepsilon}^{C,i}$ are self-explanatory. In the unperturbed case, there is no risk of collision in the long run because asymptotic platooning is achieved exponentially fast. In this case, we have $\sigma_i=0$, which implies $\mathcal R_{\varepsilon}^{C,i}=0$. When noise is present, we have $\sigma_i>0$. For $d$ large enough, vehicles lie far away from each other in expectation. Hence vehicle collision is unlikely to occur. Finally, when the standard deviation $\sigma_i$ exceeds an $\varepsilon$-dependent cut-off, the value of risk is $+ \infty$ in the following sense: collision cannot be avoided with probability higher than $1-\varepsilon$ for this range of $\sigma_i$'s. The curve of $\mathcal R_{\varepsilon}^{C,i}=\mathcal R_{\varepsilon}^{C,i}(\sigma_i)$ is graphically illustrated in Fig. \ref{fig: riskgraph}.

\begin{figure}\center
\includegraphics[trim = 20 0 20 10, clip,width=9cm]{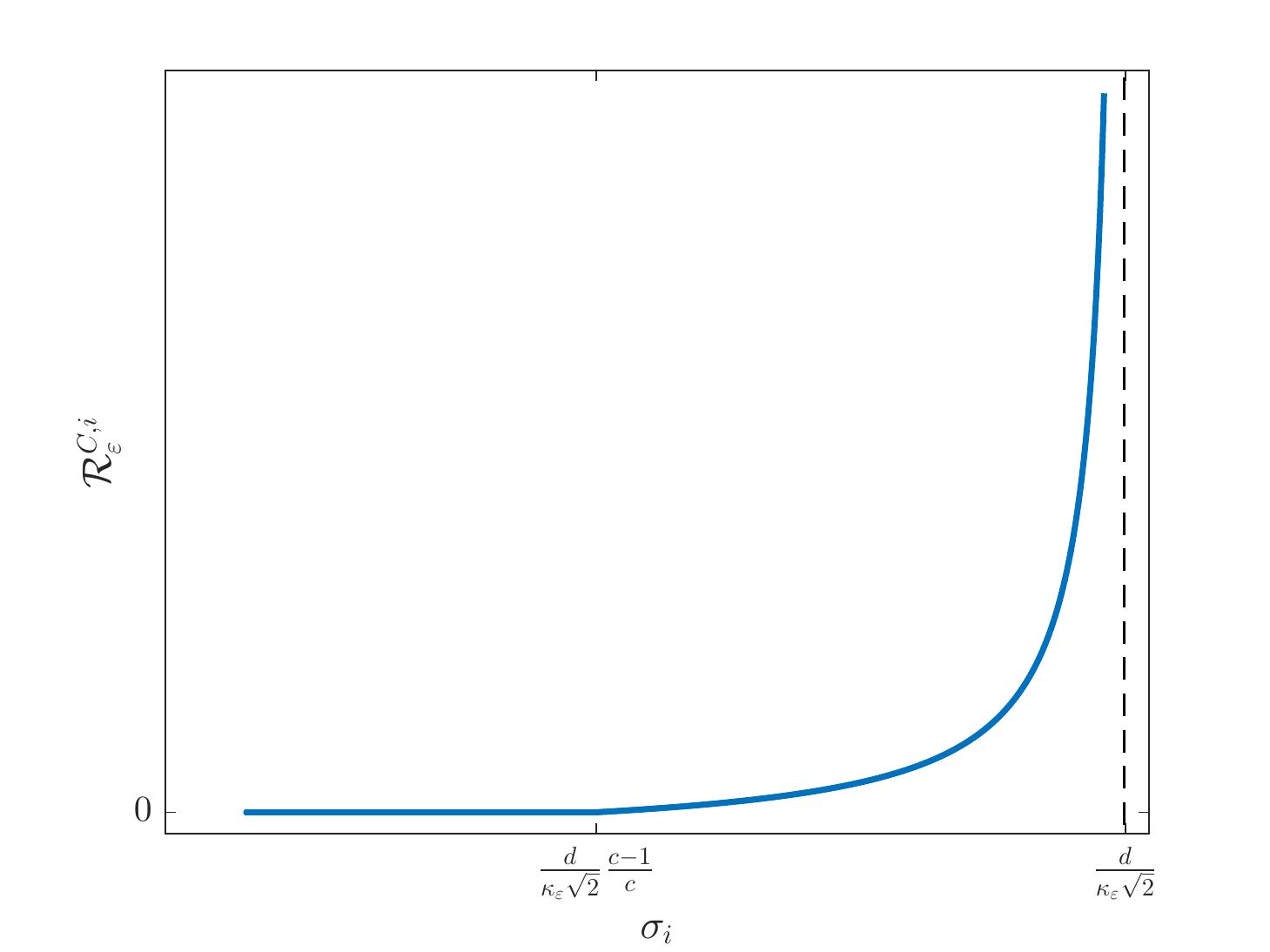}\caption{Graph of the risk of collision as a function of $\sigma_i$. The plot for the risk of detachment is qualitatively equal.}\label{fig: riskgraph}
\vspace{-0pt}
\end{figure}
\subsection{Inter-Vehicle Detachment}
 Cohesive motion in convoys of vehicles implies aligned traveling  of agents within prescribed distance. For platooning ensembles of Fig. \ref{fig: the platoon} any two successive vehicles that steered too distant from each other may trigger undesirable phenomena. A case-study that resonates with the geometry of Fig. \ref{fig: the platoon} is this: In practice, vehicles are equipped with onboard finite-range communication modules. The $i$'th vehicle can establish reliable communication with those vehicles whose positions lie in range $(x_t^{(i)}-r^*,x_t^{(i)}+r^*)$, where $r^*>0$ is the communication range. Thus, a condition like $d\leq r^*$ is necessary to guarantee existence of a  reliable communication topology that guarantees reliable transmission of data in accordance to Assumption \ref{assum0}. Without loss of generality, let us take $$r^*=ad$$ for some parameter $a\geq 1$. In the presence of stochastic noise, it is likely that two successive vehicles move further away from each other beyond $r^*$  and into dangerous communication status.  Detachment between vehicles  $i+1$ and $i$ occurs when their relative position exceeds $r^*$.  It is said that vehicles $i+1$ and $i$ are dangerously close to lose communication and experience detachment  when $x_t^{(i+1)}-x_t^{(i)}$ takes values in $\big[ad-\frac{1}{h},ad\big]$ for some design parameter $h >0$.   Parameter $h$ plays the same role as parameter $c$ in the inter-vehicle collision scenario: the higher the value of $h$, the narrower the length of the alarm zone prior to experiencing detachment. Let us define the family of parameterized events 
\begin{equation*}\label{eq: amputationevent}
\bigg\{ \overline{x}^{(i+1)}-\overline{x}^{(i)} \in  D_{\delta}  \bigg\} \hspace{0.2in} \text{where} \hspace{0.2in} D_{\delta} =\bigg(a d -\frac{1}{\delta+h}~ , ~+\infty\bigg)
\end{equation*} for $\delta\in \mathbb R_+$. 
\noindent Similar to the collision super-sets, the collection $\{D_{\delta}\}_{\delta \in \R_+}$  satisfies properties $\Pi_1$ and $\Pi_2$. The value-at-risk of detachment between vehicles $i+1$ and $i$ is defined as 
\begin{equation}\label{eq: riskamp}
\mathcal R_{\varepsilon}^{D,i}=\inf\left\{\delta>0 ~\Big|~\mathbb P \left\{ \overline{x}^{(i+1)}-\overline{x}^{(i)} \in  D_{\delta} \right\}<\varepsilon \right\}
\end{equation}for fixed confidence level $\varepsilon\in (0,1)$. Following the same steps as in the proof of Theorem \ref{thm: main2},  the closed-form expressions for the  risk of platoon detachment is given by 
\begin{equation}\label{eq: riskampformula}
\mathcal R_{\varepsilon}^{D,i}=\begin{cases}
0~\hspace{0.8in}& \text{if}~~~ \sigma_{i}\leq \frac{(a-1)d-1/h}{\kappa_{\varepsilon}\sqrt{2}} , ~~ \text{or}~~\varepsilon\geq \frac{1}{2}\\
\frac{1}{(a-1)d-\sqrt{2}\kappa_{\varepsilon}\sigma_i}-h~& \text{if}~~~ \frac{(a-1)d-1/h}{\kappa_{\varepsilon}\sqrt{2}}<\sigma_i<\frac{(a-1)d}{\kappa_{\varepsilon}\sqrt{2}}\\
\infty~ & \text{if}~~~\sigma_i\geq \frac{(a-1)d}{\kappa_{\varepsilon}\sqrt{2}}
\end{cases}
\end{equation}
where $\kappa_\varepsilon=\text{erf}^{-1}(1-2\varepsilon),~\varepsilon<\frac{1}{2}$, and $\sigma_i$ is as in Theorem \ref{cor: main1}.  Both collision and detachment scenarios are illustrated in Fig. \ref{fig: scenarios}.

\section{Risk Of Multiple Systemic Events In The Platoon}\label{sect: jointrisk}

We generalize results of the previous section by considering multiple collision and detachment events that involve more than two vehicles and may happen simultaneously. The idea is to define proper super-sets for random vectors $\mathbf y=\big[y^{(1)},\dots,y^{(q)}\big]^T$, where $y^{(i)}=\overline{x}^{(i+1)}-\overline{x}^{(i)}$ for $i=1,\dots,q$. In the following, several scenarios are considered.  

  
\begin{figure}[t]\center
\includegraphics[scale=0.17]{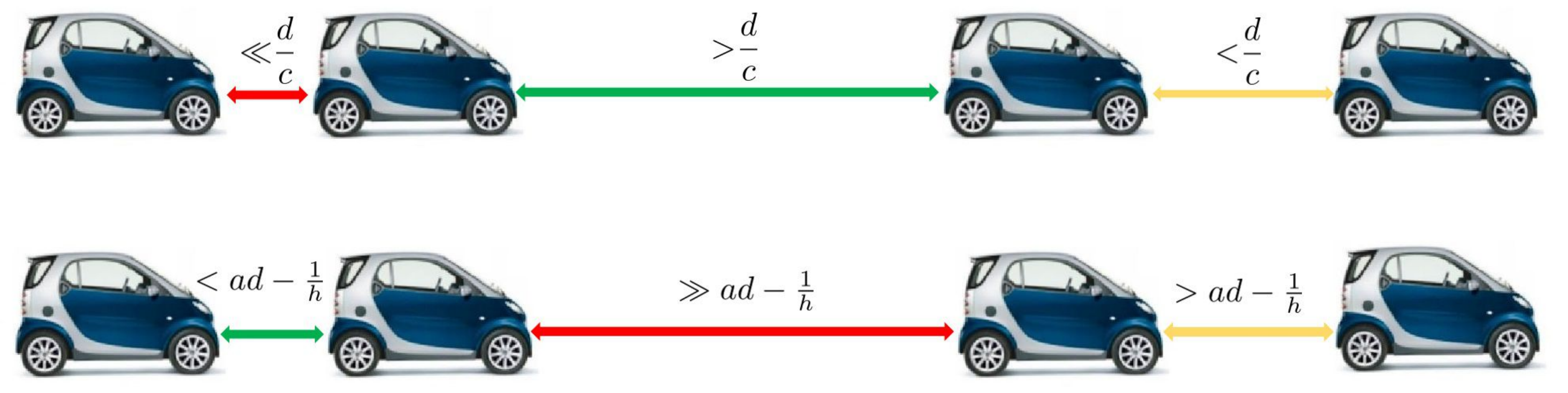}
\caption{Collision and detachment events in platooning. If relative distance between two successive vehicles exceeds $ad-h^{-1}$, then the two vehicles place themselves close in formation of two connected components and into dissolution of flok. If their relative distance is less than $d/c$, then vehicles are dangerously close to experience a collision. Parameters $a,c,d,h$ are  explained in Section \ref{sect: riskplatoon}. The arrows comply with the color code: green implies a safe operating mode, red means an unsafe operating mode, and yellow is an alert for near-collision/detachment events, respectively.}\label{fig: scenarios}
\end{figure}

By constructing the full  conjunction of the individual events discussed in Section \ref{sect: riskplatoon}, we can formulate risk of simultaneous collision and detachment between some pairs of successive vehicles  throughout the platoon via
\begin{equation}\label{eq: jointcuprisk}
\begin{split}
\boldsymbol{\mathcal R}_{\varepsilon}^{C,\cup}&=\inf\left\{\boldsymbol{\delta}\in \mathbb R^{n-1}_+ ~\Bigg |~  \mathbb P\left(  \bigcup_{i=1}^{n-1} \left\{ \overline{x}^{(i+1)}-\overline{x}^{(i)} \in C_{\delta_i} \right\}   \right ) <\varepsilon  \right\}\\
\boldsymbol{\mathcal R}_{\varepsilon}^{D,\cup}&=\inf \left\{\boldsymbol{\delta}\in \mathbb R^{n-1}_+ ~\Bigg |~ \mathbb P\left ( \bigcup_{i=1}^{n-1} \left\{ \overline{x}^{(i+1)}-\overline{x}^{(i)} \in D_{\delta_i} \right\}   \right )<\varepsilon  \right\}.
\end{split}
\end{equation} 
The risk of simultaneous collision and detachment across the platoon is measured by 
\begin{equation}\label{eq: jointcaprisk}\begin{split}
\boldsymbol{\mathcal R}_{\varepsilon}^{C,\cap}&=\inf\left\{\boldsymbol{\delta}\in \mathbb R^{n-1}_+  ~\Bigg |~  \mathbb P\left ( \bigcap_{i=1}^{n-1} \left\{ \overline{x}^{(i+1)}-\overline{x}^{(i)} \in C_{\delta_i} \right\} \right ) <\varepsilon  \right\}\\
\boldsymbol{\mathcal R}_{\varepsilon}^{D,\cap}&=\inf\left\{\boldsymbol{\delta}\in \mathbb R^{n-1}_+  ~\Bigg |~  \mathbb P\left ( \bigcap_{i=1}^{n-1} \left\{ \overline{x}^{(i+1)}-\overline{x}^{(i)} \in D_{\delta_i} \right\} \right ) <\varepsilon  \right\}.
\end{split}
\end{equation} 

The calculation of either \eqref{eq: jointcuprisk} or \eqref{eq: jointcaprisk} in closed-form is mathematically intractable as it requires working with multi-variable normal random variables. Since obtaining an explicit expressions is not feasible, we rely on first-order approximations using the vector of individual surrogates as in \eqref{eq: vectorrisk_s}.  Let us define the vectors of individual risks for collision and detachment as
\begin{equation}\label{eq: vectorrisk}\begin{split}
\mathfrak{R}_{\varepsilon}^{C}&=\left[\mathcal R_\varepsilon^{C,1},\mathcal R_\varepsilon^{C,2},\dots,\mathcal R_\varepsilon^{C,n-1}\right]^T,\\
\mathfrak{R}_{\varepsilon}^{D}&=\left[\mathcal R_\varepsilon^{D,1},\mathcal R_\varepsilon^{D,2},\dots,\mathcal R_\varepsilon^{D,n-1}\right]^T,
\end{split}
\end{equation} \noindent respectively. 

\begin{thm}\label{thm: riskestimates}
Suppose that conditions of Theorem \ref{thm: main0} hold.  The risk of simultaneous collision and detachment of all vehicles satisfies 
\begin{equation*}
\boldsymbol{\mathcal R}^{C,\cap}_{\varepsilon} \subseteq~ \mathbb V_{\mathfrak{R}_\varepsilon^{C}}~~~\text{and}~~~\boldsymbol{\mathcal R}^{D,\cap}_{\varepsilon} \subseteq~ \mathbb V_{\mathfrak{R}_\varepsilon^{D}},
 \end{equation*}
where 
\begin{equation*}\label{epsilon-q}
\mathbb V_{\mathfrak{R}_\varepsilon^{\square}}=\left\{ \left. \left[\begin{array}{c}
\delta_1 \\
\vdots \\
\delta_q
\end{array}\right]  ~\right|\begin{array}{c} 
\vspace{0.06cm}
0 \leq \delta_i \leq \mathcal R_{\varepsilon}^{\square,i} 
\end{array}  \right\} 
\end{equation*}
where  $\square$ is replaced by either $C$ or  $D$, 
and the risk of simultaneous collision and detachment of some of the vehicles satisfy 
\begin{equation*}
\boldsymbol{\mathcal R}^{C,\cup}_{\varepsilon} \subseteq~ \mathbb W_{\mathfrak{R}_\varepsilon^{C}}~~~\text{and}~~~\boldsymbol{\mathcal R}^{D,\cup}_{\varepsilon} \subseteq~ \mathbb W_{\mathfrak{R}_\varepsilon^{D}}
 \end{equation*}
where  
\begin{equation*}
\mathbb W_{\mathfrak{R}^{\square}_\varepsilon}=\left\{ \left. \left[\begin{array}{c}
\delta_1 \\
\vdots \\
\delta_q
\end{array}\right]  ~\right|\begin{array}{c} 
\vspace{0.06cm}
\mathcal R_{\varepsilon}^{\square,i} \leq \delta_i \leq \mathcal R_{\varepsilon_i}^{\square,i}\\
\textrm{for}~\varepsilon_i \in (0,1) ~\textrm{that satisfy:}  \\
\hspace{-0cm}\varepsilon_1+\dots+\varepsilon_q = \varepsilon 
\end{array}  \right\}
\end{equation*}
with  $\square$ being either $C$ or  $D$. 
\end{thm}
As it is discussed in the Appendix, this result relies on the Boole-Fr\'{e}chet inequalities. These inequalities are the best possible estimates on unions and intersections of events for which nothing is known other than the probabilities of the corresponding individual events \cite{halperin65}.  We remark that, unlike the case of global union of events in Theorem \ref{thm: riskestimates}, there is no non-trivial lower limit for the risk of global joint events. Despite their elegance, Boole-Fr\'{e}chet inequalities fail to provide non-trivial lower bounds for the type of joint risk events. This point is elaborated in the Appendix.

One can define more general scenarios by grouping the elements of $\mathbf y$ into classes of interest. Let $\mathfrak{P}$   be a partition of the set $\{1,\dots,n-1\}$ into mutually disjoint subsets $P_1, \dots, P_{|\mathfrak{P}|}$, where  $|\mathfrak{P}|$ denotes the cardinality of $\mathfrak{P}$.  It follows that $\bigcup_{k=1}^{|\mathfrak{P}|} P_k = \{1,\dots,n-1\}$. This labeling classifies the elements of $\mathbf y$ into $|\mathfrak{P}|$ groups. This notation allows to formulate the following class of risk measures 
\begin{gather}
\inf \left\{ \boldsymbol \delta \succeq \mathbf 0~\Bigg|~\mathbb P \left ( \bigcap_{k=1}^{|\mathfrak{P}|} \bigcup_{i\in P_k} \left\{ \overline{x}^{(i+1)}-\overline{x}^{(i)} \in U_{\delta_i}\right\} \right )< \varepsilon \right\}\label{eq:prisk1}
\end{gather}
and
\begin{equation}\label{eq:prisk2}
\inf \left\{ \boldsymbol \delta \succeq \mathbf 0 ~\Bigg |~ \mathbb P \left ( \bigcup_{k=1}^{|\mathfrak{P}|} \bigcap_{i\in P_k}\left\{ \overline{x}^{(i+1)}-\overline{x}^{(i)} \in U_{\delta_i}\right\} \right )< \varepsilon \right\}
\end{equation} 
in which $U_{\delta_i}$ is either $C_{\delta_i}$ or  $D_{\delta_i}$. The risk measure \eqref{eq: jointcaprisk} is a special case of \eqref{eq:prisk1} when $|\mathfrak{P}|=1$, where  \eqref{eq:prisk1} quantifies risk of at least one event in every group of the partition will experience a systemic event. Similarly, \eqref{eq: jointcuprisk} is a special case of    \eqref{eq:prisk2}, where  \eqref{eq:prisk2} measures likelihood of all members of at least one of the groups in the partition experiences a systemic event. The results of  Theorem \ref{thm: riskestimates} and  Proposition \ref{prop: monotonicityofrisk} can be combined to obtain bounds for \eqref{eq:prisk1} and \eqref{eq:prisk2} in terms of vector \eqref{eq: vectorrisk}. The results of this section allow us to design low risk platoons by formulating multi-objective optimization problems based on vectors of individual risk measures. 
%
%

\begin{figure}\center
\includegraphics[trim = 10 0 30 15, clip,width=8.5cm]{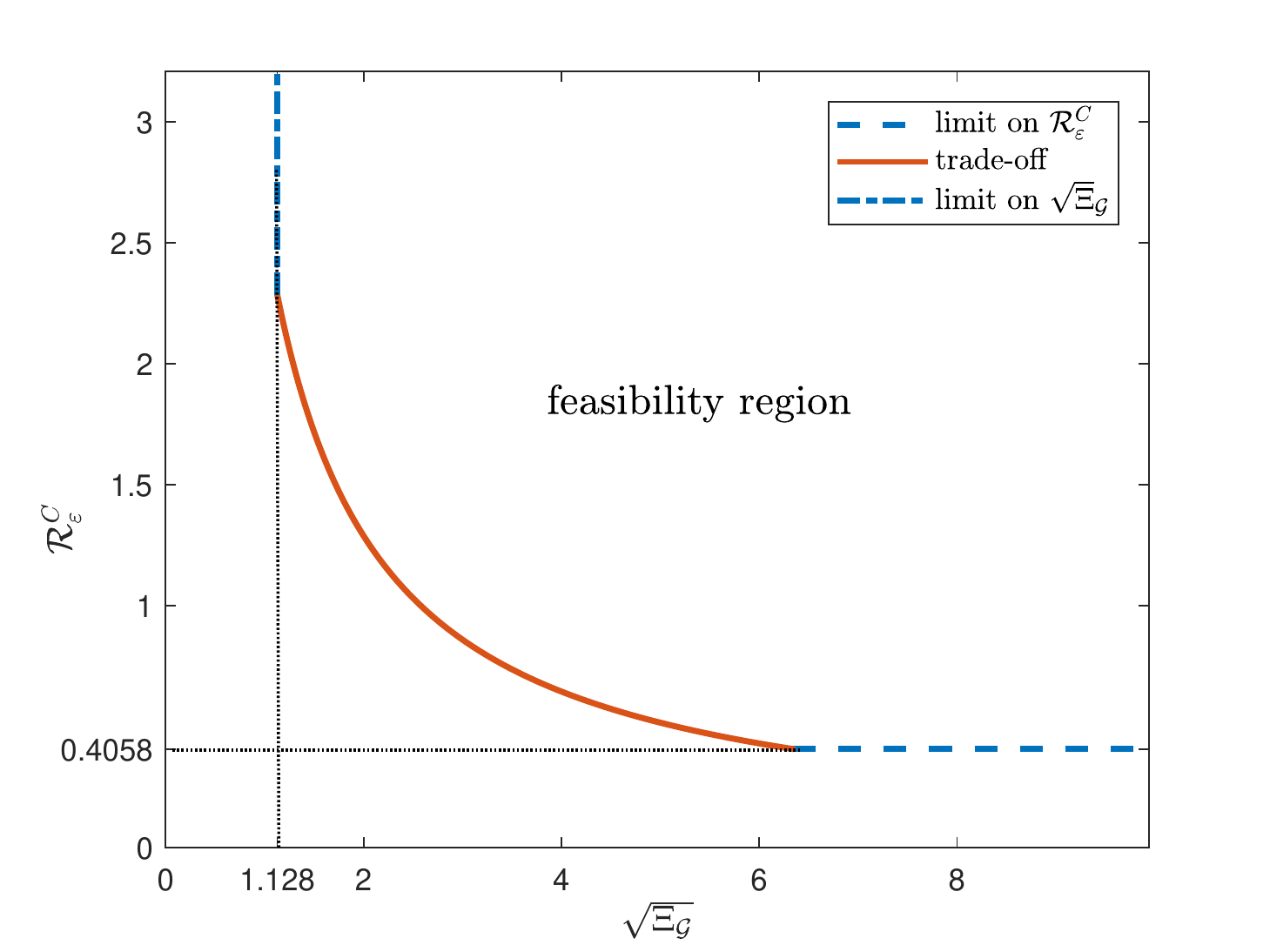}\caption{Inherent limits and trade-offs. The dashed lines illustrate Theorems \ref{cor: limit1} and \ref{prop: effectiveresistancelimit} on the best achievable values of Risk (horizontal) and Effective Graph Resistance (vertical). The solid curve illustrates the trade-off risk of collision and network connectivity as established in Theorem \ref{thm: main3}.}
\end{figure}

\section{Fundamental Limits and Trade-Offs}\label{sect: tradeoff}
An engineer has almost no control over the communication time-delay and exogenous disturbances. In such situations, one can design optimal communication topologies to minimize  the disruptive effect of such imperfections. Our goal is to explore inherent shortcomings in network design when the effects of neither noise nor time-delay can be neglected\footnote{Discussion in this and the following sections is focused on risk of inter-vehicle collision. Results on detachment can be derived in a similar fashion.}.
\begin{lem}\label{prop: limit} The marginal standard deviations $\sigma_i$, as in Theorem \ref{cor: main1}, satisfy the lower bound
\begin{equation*}
\sigma_i\geq \sigma^*:= \sqrt{\frac{1}{\pi}\, \underline{f}} ~\Sp  |g| \Sp \tau^{3/2}
\end{equation*} 
for $i=1,\dots,n-1$, where  $\underline{f}:=\inf_{(s_1, s_2)\in S}f(s_1,s_2)\approx 25.4603$.
\end{lem}
This result reveals that the variance of $\overline{x}^{i+1}-\overline{x}^{i}$ attains a hard lower bound that only depends  on the strength of the diffusion $g$ and the time-delay $\tau$. 

\begin{figure}
\includegraphics[trim = 10 0 30 15, clip,width=8.5cm]{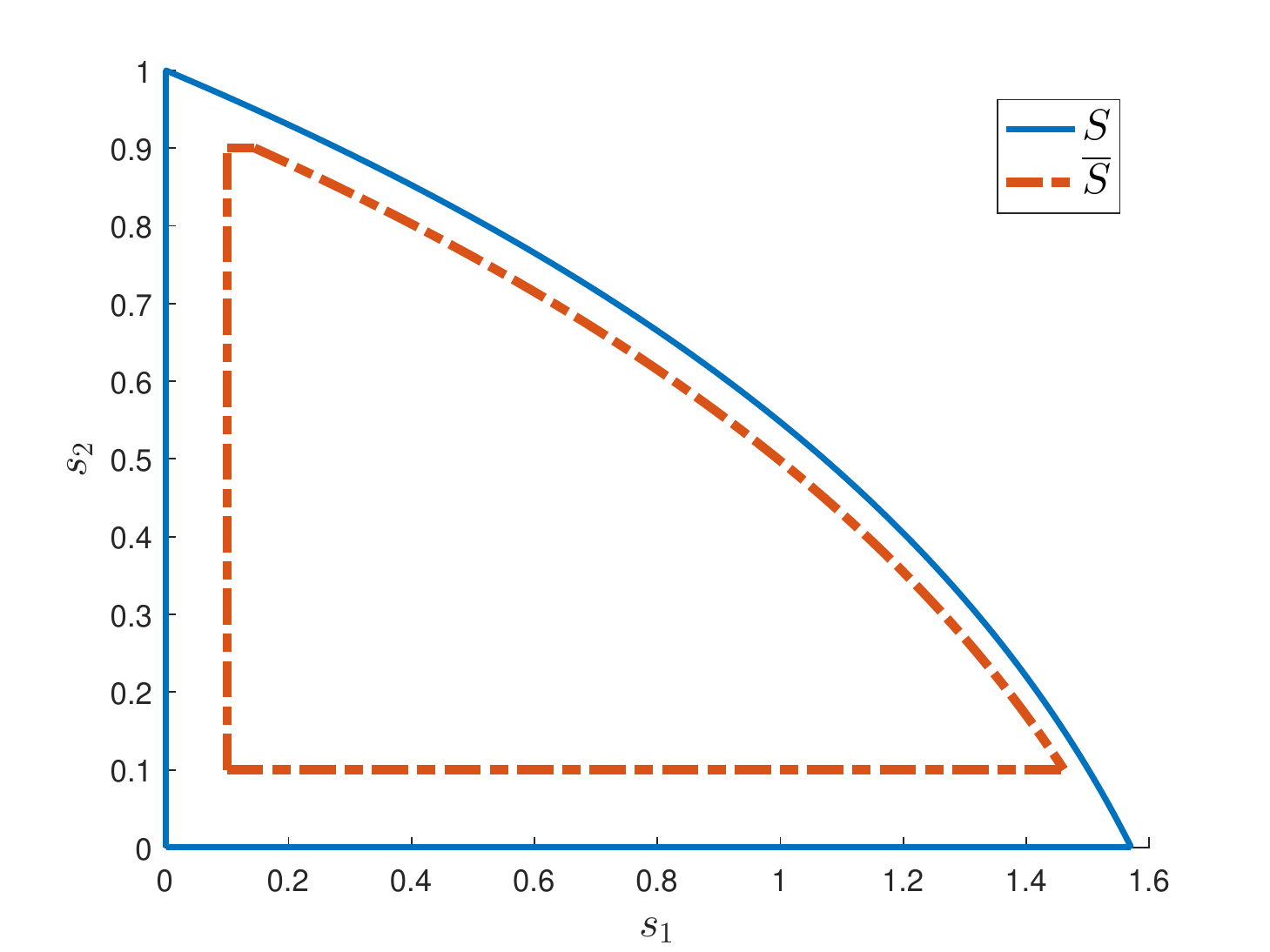}\caption{The compact set $\overline{S}$ is defined in the interior of $S$. Function $f$ is approximated over $\overline{S}$ with relative error in the order of $10^{-4}$.}
\label{fig: stability2}
\end{figure}  

\begin{thm}\label{cor: limit1} There is an inherent fundamental limit on the best achievable values of risk of inter-vehicle collision in the platoon that is given by 
\begin{equation*}\begin{split}
\mathcal R_{\varepsilon}^{C,i}\geq \begin{cases}
0, & \hspace{0.1in} \text{if}~~~\sigma^*\leq \frac{d}{\kappa_{\varepsilon}\sqrt{2}}\frac{c-1}{c}~~\text{or}~~\varepsilon\geq \frac{1}{2}\\
\frac{d}{d-4.02\, \kappa_{\varepsilon} |g| \tau^{3/2}}-c, & \hspace{0.1in} \text{if}~~~\sigma^*\in \big(\frac{d}{\kappa_{\varepsilon}\sqrt{2}}\frac{c-1}{c},\frac{d}{\kappa_{\varepsilon}\sqrt{2}}\big) \\
\infty, & \hspace{0.1in} \text{if}~~~\sigma^*>\frac{d}{\kappa_{\varepsilon}\sqrt{2}}
\end{cases}.
\end{split}
\end{equation*}
\end{thm}

According to this result, the systemic risk measure attains the trivial lower limit zero when $\sigma^*$ is less than or equal to the critical value $\frac{d}{\kappa_{\varepsilon}\sqrt{2}}\frac{c-1}{c}$. If $\sigma^*$ lies in a specific set of values, the risk of collision can be minimized as a function of communication topology. On the other extreme, risk of inter-vehicle collision becomes infinite (i.e., collision becomes inevidable) if $\sigma^*$ exceeds a safety cut-off value. In fact, we can characterize an inevitability condition as follows: the risk of collision between two consecutive vehicles is infinite for all platoons, independent of their communication topology, if\begin{equation*}
 |g|\Sp \Sp \tau^{\frac{3}{2}} \geq \frac{d}{1.12~ \text{\normalfont {erf}}^{-1}(1-2\varepsilon)}
\end{equation*}
for all $\varepsilon \in (0,0.5)$.  The risk can be reduced by minimizing the marginal standard deviations $\sigma_i$ up to a limit, which is characterized by Lemma \ref{prop: limit}, by adjusting the platoon's control parameters (i.e., feedback gains $k_{ij}$ and $\beta$). 
According to Theorem \ref{thm: main2} and Lemma \ref{prop: limit}, the Laplacian spectrum of a platoon with minimal risk must satisfy  
\begin{equation*}
\lambda_1 = 0,~~~~  \lambda_j=\frac{1}{\tau}\,\underline{s}_1~~~~~\text{for}~~ j=2\dots,n, ~~\text{and}~~\beta=\frac{1}{\tau}\,\underline{s}_2
\end{equation*} 
 where $\underline{s}_1$ and $\underline{s}_2$ are determined through 
\begin{equation*}
\big(\underline{s}_1,\underline{s}_2\big)=\argmin_{(s_1,s_2)\in S}f(s_1,s_2)\approx (1.111,~0.220).
\end{equation*} 
Therefore, the optimal communication topology is a complete graph with link weights $k_{ij} = \frac{\underline{s}_1}{n\tau}$ for all $i,j \in \{1,\dots,n\}$. 
\begin{figure}
\includegraphics[scale=0.3]{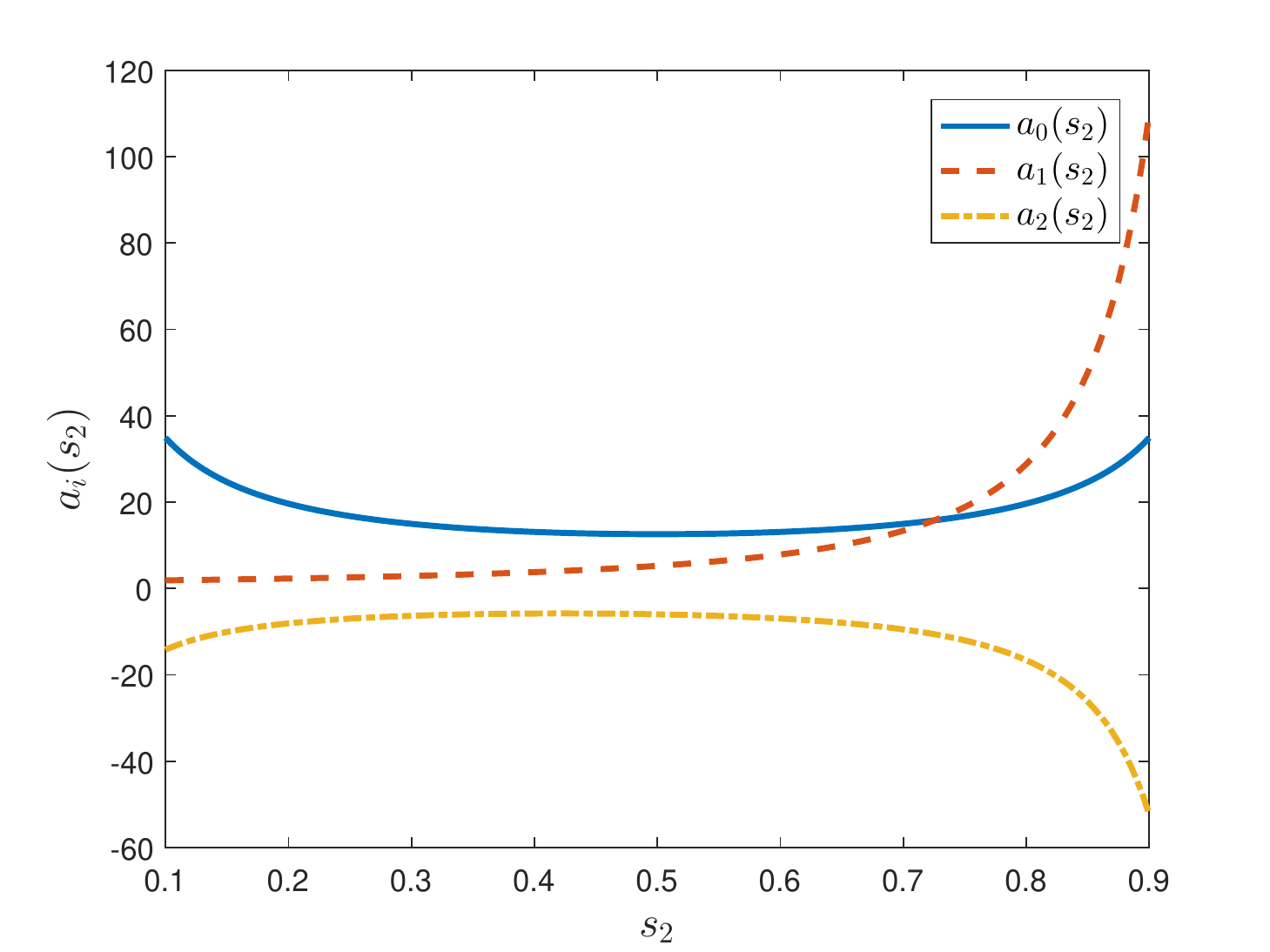}
\includegraphics[scale=0.3]{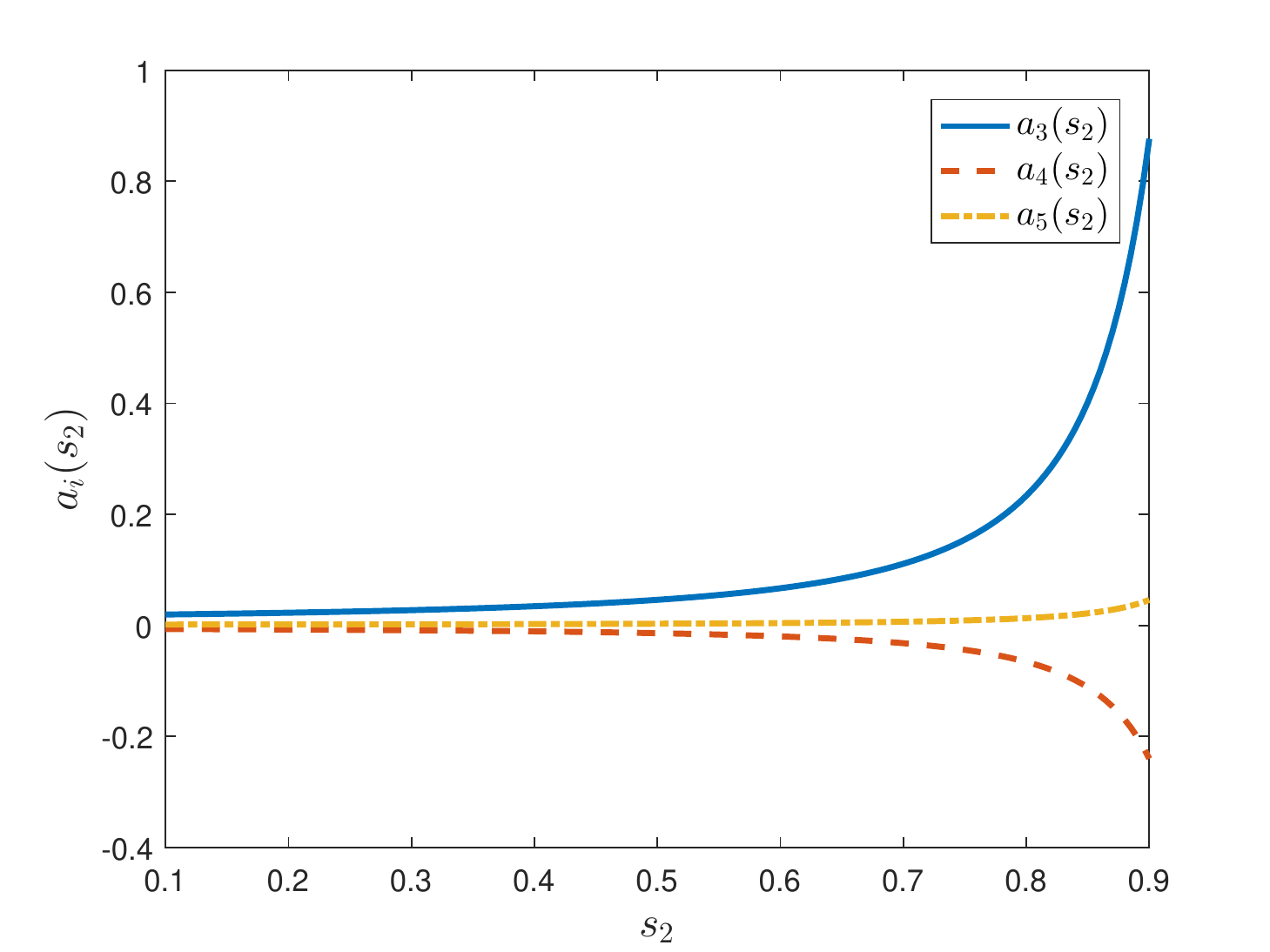}
\includegraphics[scale=0.3]{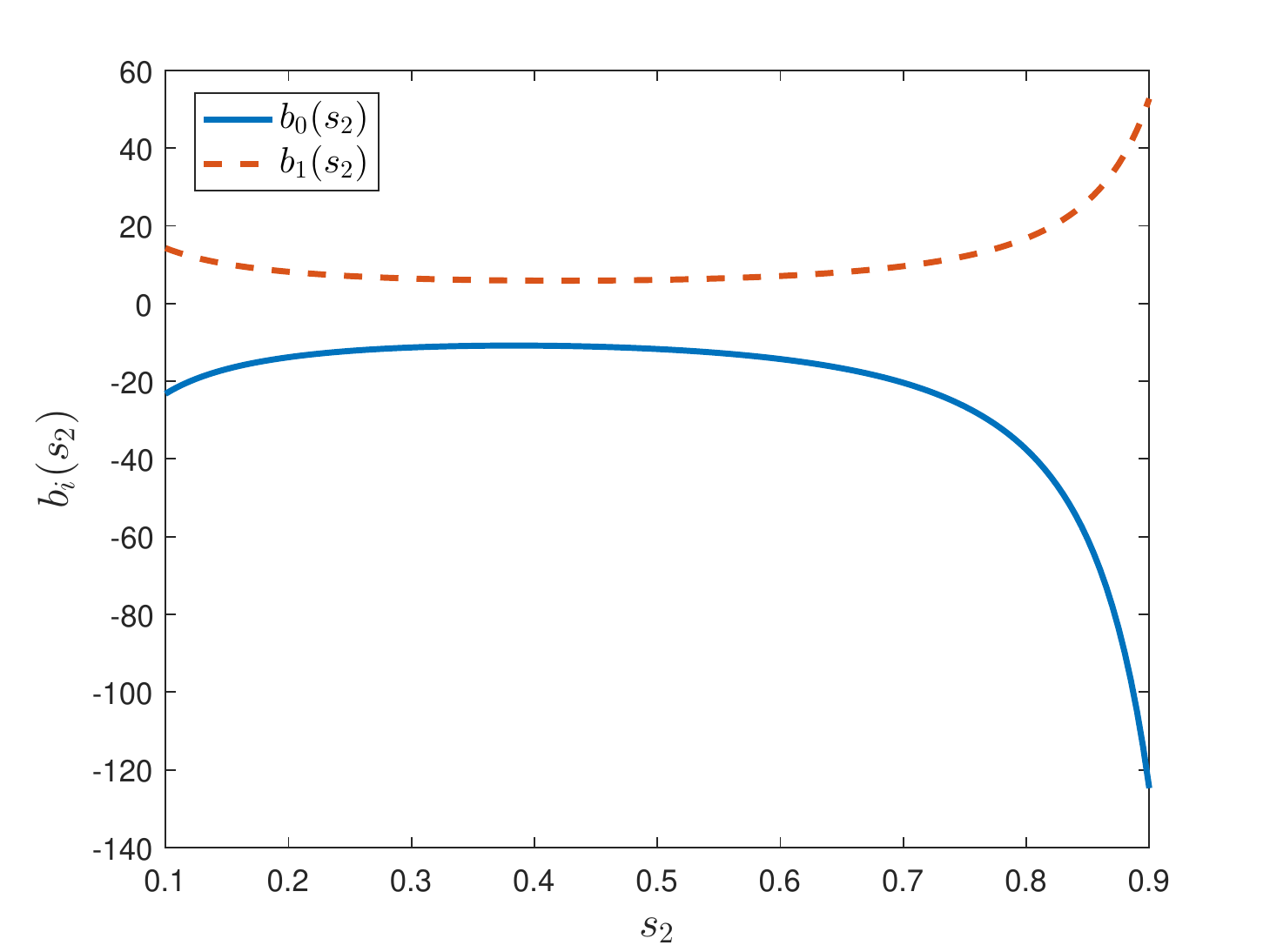}
\includegraphics[scale=0.3]{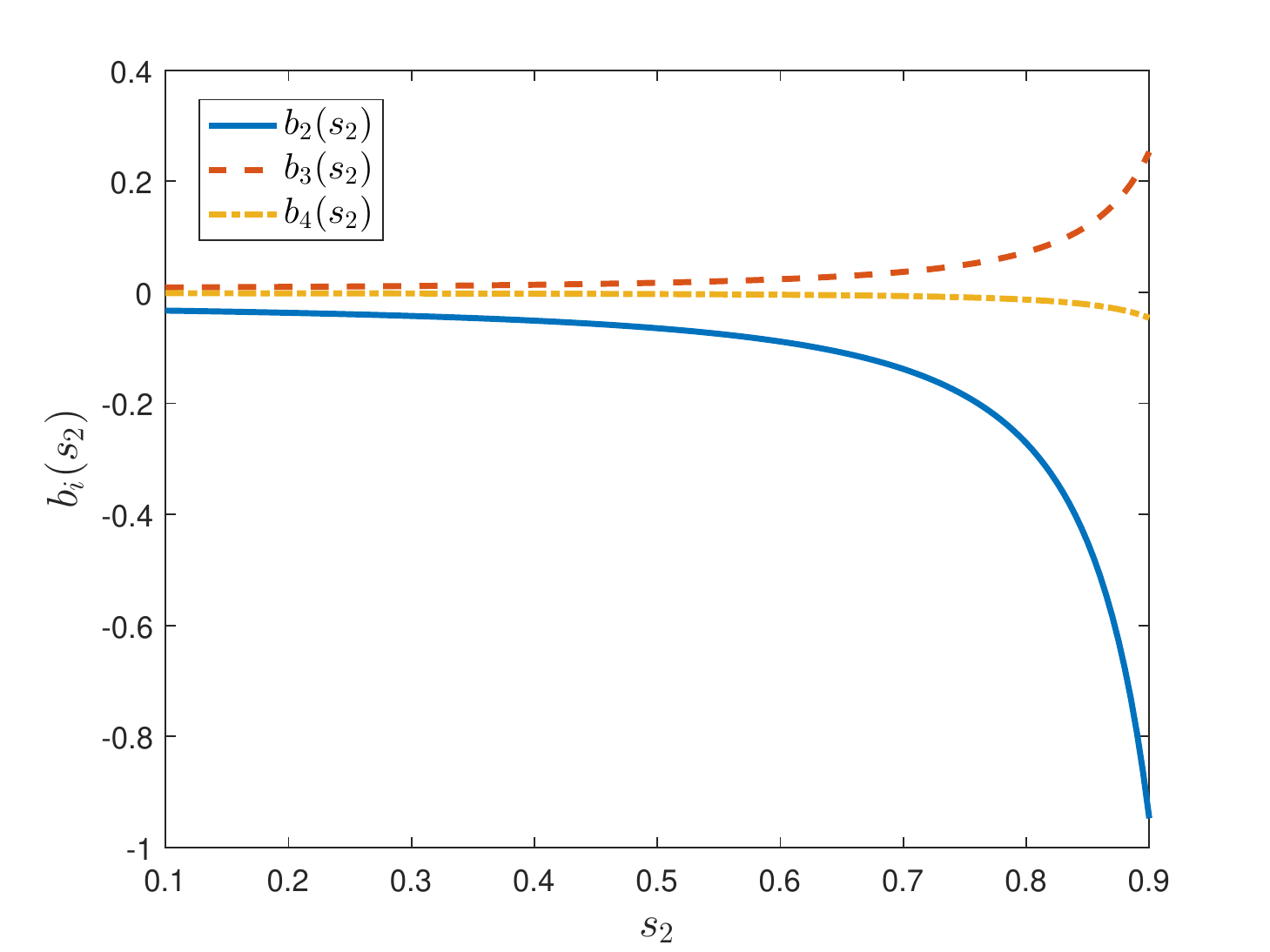}
\caption{Plots of coefficiens $a_i$ of $A_{s_2}(s_1)$, and $b_i$ of $B_{s_2}(s_1)$,  as functions of $s_2$.}\label{fig: coef1}
\end{figure} Optimal network topology as described above is far from ideal in real-world situations. Given an arbitrary (sub-optimal) network, synthesis objective regards developing formal methods to attain safe vehicle interactions. In this setting, safe is interpreted as the control topology that mitigates risk of systemic events. To this end, we recall that network synthesis is conducted via adding new feedback loops (coupling links), sparsifying, or adjusting the exiting feedback gains.  Every such operation impacts the eigen-spectrum of graph laplacian $L$, and this has immediate consequence on systemic risk measures $\mathcal R_{\varepsilon}^{\mathcal C}$. In the remainder of this section we report the existence of a counter-intuitive trade-off between network connectivity and risk of inter-vehicle collision. The measure of network connectivity, already introduced in \S \ref{sect: prelim}, is the effective graph resistance $\Xi_{\mathcal G}$. The next result is an interesting delay-induced fundamental limitation in close spirit to Theorem \ref{cor: limit1}.

\begin{thm}\label{prop: effectiveresistancelimit}
For given control parameter $\beta>0$ and $\tau>0$ such that $\beta\tau \in (0,1)$,  the communication connectivity, which is specified using the total effective resistance \eqref{eq: graphr}, cannot be improved beyond some certain threshold according to inequality  
\begin{equation}\label{eq: xilimit}
\Xi_{\mathcal G} \Sp > \Sp n(n-1)\frac{\tau}{\vartheta(\beta \tau)} 
\end{equation} where  $\vartheta~:~(0,1)\rightarrow \big(0,\frac{\pi}{2}\big)$ is defined by
\begin{equation}\label{fcn-1}
\vartheta(s_2)= \left(g\circ f^{-1}\right)(s_2)
\end{equation}
with $g(x)=x\sin(x)$ and $f(x)=x \cot{x}$.
\end{thm} The smallest lower bound for \eqref{eq: xilimit} is achieved when $\beta=0$, which in that case \eqref{eq: xilimit} becomes 
\begin{equation}\label{eq: xilimit_opt}
\Xi_{\mathcal G}> n(n-1)\frac{2\tau}{\pi}.
\end{equation}
The main advantage of $\Xi_{\mathcal G}$, over other measures of connectivity (such as vertex/edge connectivity \cite{Bollobas_1998}) is that effective resistance applies equally well to weighted symmetric graphs as well as to topological (unweighted) graphs. Moreover it attains an elegant spectral representation that favors the comparison with risk expressions derived in Theorem \ref{thm: main2}.
Going beyond hard limits, we show that fundamental trade-offs emerge between risk and network connectivity. 
These trade-offs explain that for a non-trivial range of feedback gains, $k_{ij}$, improving connectivity, which results in decreasing $\Xi_{\mathcal G}$, leads to higher levels of systemic risk. For a rigorous exposition of our results, we need to employ some new notations. Let us define
$$\underline{f}_m:=\min_{j \in \{2,\dots,n\}}\Sp\inf_{(s_1,s_2)\in S}\bigg\{ f(s_1,s_2)\,\bigg[ \frac{(j-1)}{s_1} +\frac{(n-j)}{\zeta(s_2)}\bigg]^{\frac{2}{m}}\bigg\}$$ 
for every $m\geq 1$, where $f$ is given by  \eqref{eq: functionf} and $\vartheta$ by \eqref{fcn-1}, as well as  
\begin{equation*}
\begin{split}
\underline{E}^C&:=(1-c)^2 +2(1-c)c\frac{\sqrt{2}\kappa_\varepsilon}{d}\max\bigg\{\frac{c-1}{c}\frac{d}{\sqrt{2}\kappa_{\varepsilon}},\sigma^*\bigg\}\\&\hspace{0.2in}+c^2\bigg(\frac{\sqrt{2}\kappa_\varepsilon\,}{d}\bigg)^2 \max\bigg\{\bigg(\frac{c-1}{c}\frac{d}{\sqrt{2}\kappa_{\varepsilon}}\bigg)^2,(\sigma^*)^2\bigg\}
\end{split}
\end{equation*}
in which $\sigma^*$ is the fundamental limit of $\sigma_i$ as in Lemma \ref{prop: limit}. We also  define the sequences $\{\alpha_m^C\}_{m=1}^{\infty}$
with elements that are defined by
\begin{equation*}\label{eq: alphacoeff}\begin{split}
\alpha_m^C&=2^{\frac{m}{2}}(m+1)\left(\frac{|g| \Sp \tau^\frac{3}{2} \Sp \kappa_{\varepsilon}}{d\sqrt{\pi}}\right)^{m} \left( \underline{f}_m \right)^\frac{m}{2}.
\end{split}
\end{equation*}
 
\begin{thm}\label{thm: main3} Suppose that the conditions of Theorem \ref{thm: main2} hold. For $i \in \{1,\dots,n-1\}$, if $\sigma_i<\frac{d}{\kappa_{\varepsilon}\sqrt{2}}$, then $\sum_{m=1}^\infty \alpha_m^C<\infty$ and a fundamental trade-off between the best achievable levels of collision risk and communication connectivity  emerges as follows 
\begin{equation*}
\mathcal R_{\varepsilon}^{C,i}\,\cdot \,\sqrt{\Xi_{\mathcal G}}>\sqrt{n\tau \underline{E}^C \bigg(\frac{2(n-1)}{\pi}+\sum_{m=1}^\infty \alpha_m^C \bigg) }.
\end{equation*} 

\end{thm}

Theorem \ref{thm: main3} asserts that, for a non-trivial range of network parameters, the only way to maintain a safer (low-risk) network is through weakening  the communication connectivity, e.g., by decreasing the feedback gains or sparsifying the communication graph. Equivalently, strengthening the  connectivity, e.g., by increasing the feedback gains or adding new feedback loops or links,  increases the risk of collision (and detachment) between the vehicles. The intrinsic  trade-off between risk and communication connectivity is due to the combined effect of time-delay and noise. 

We conclude this section by highlighting that all the fundamental  limits and trade-offs disappears  when time-delay is absent, i.e., strengthening the connectivity among the vehicles reduces chances of witnessing collision events in the platoon. 

\begin{figure}\center
\includegraphics[scale=0.6]{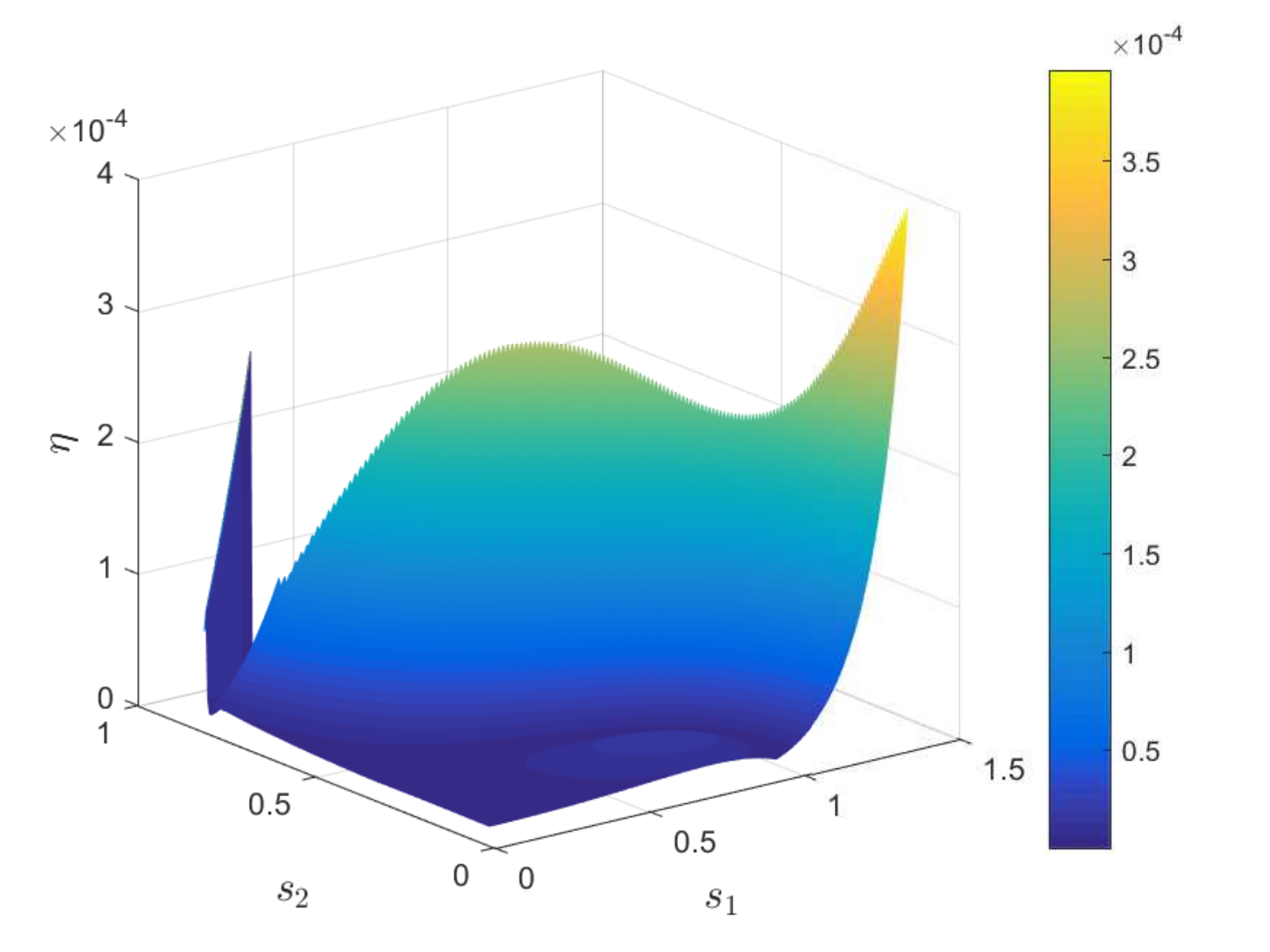}\caption{The relative error $\eta(s_1,s_2)$.}\label{fig: relativerror3}
\end{figure} 

\section{Approximation Formulas For Risk}\label{sect: approx} 

\begin{figure*}
\includegraphics[trim= 20 0 30 10, clip, width=6cm]{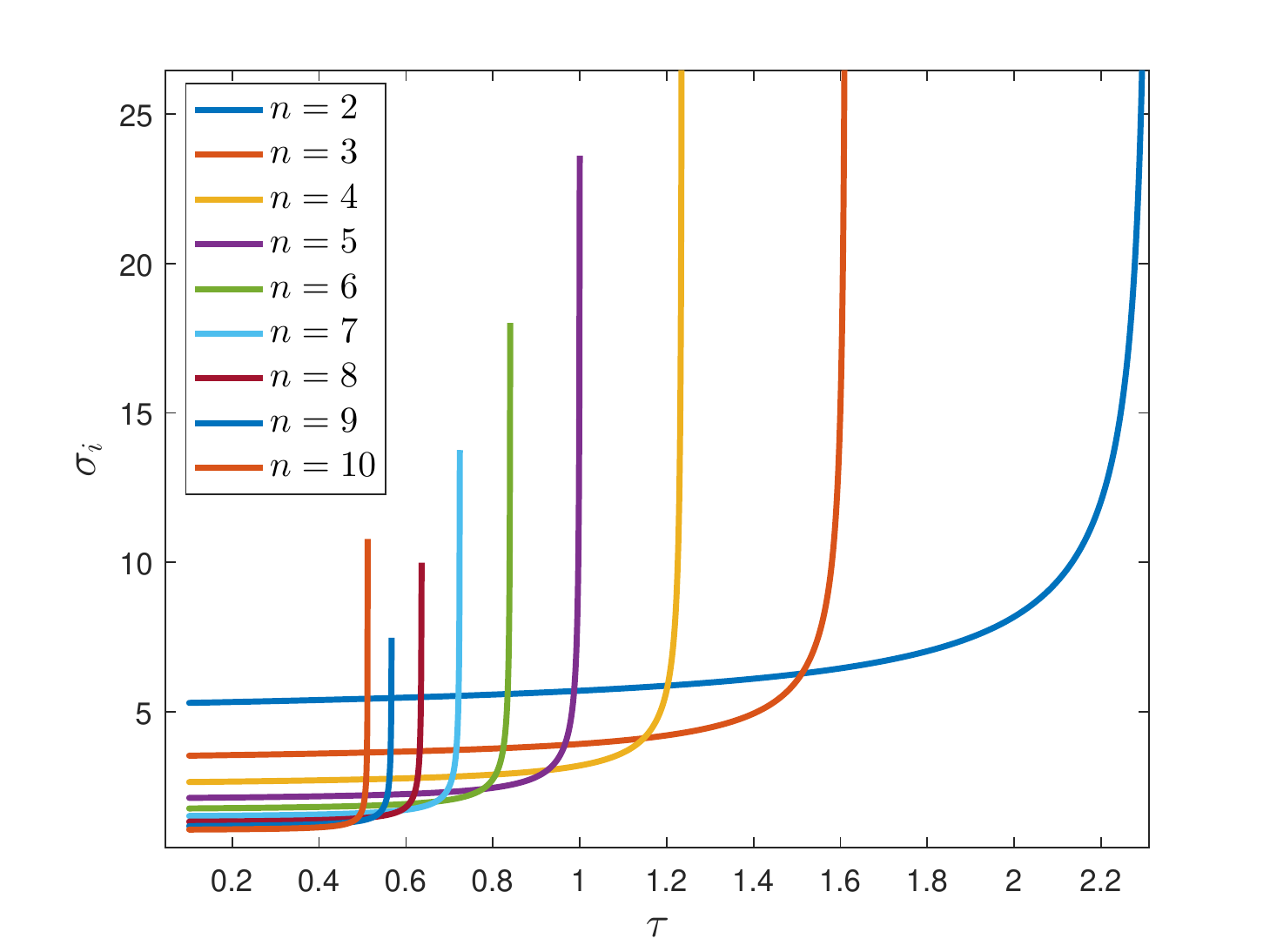}
\includegraphics[trim= 17 0 30 10, clip, width=6cm]{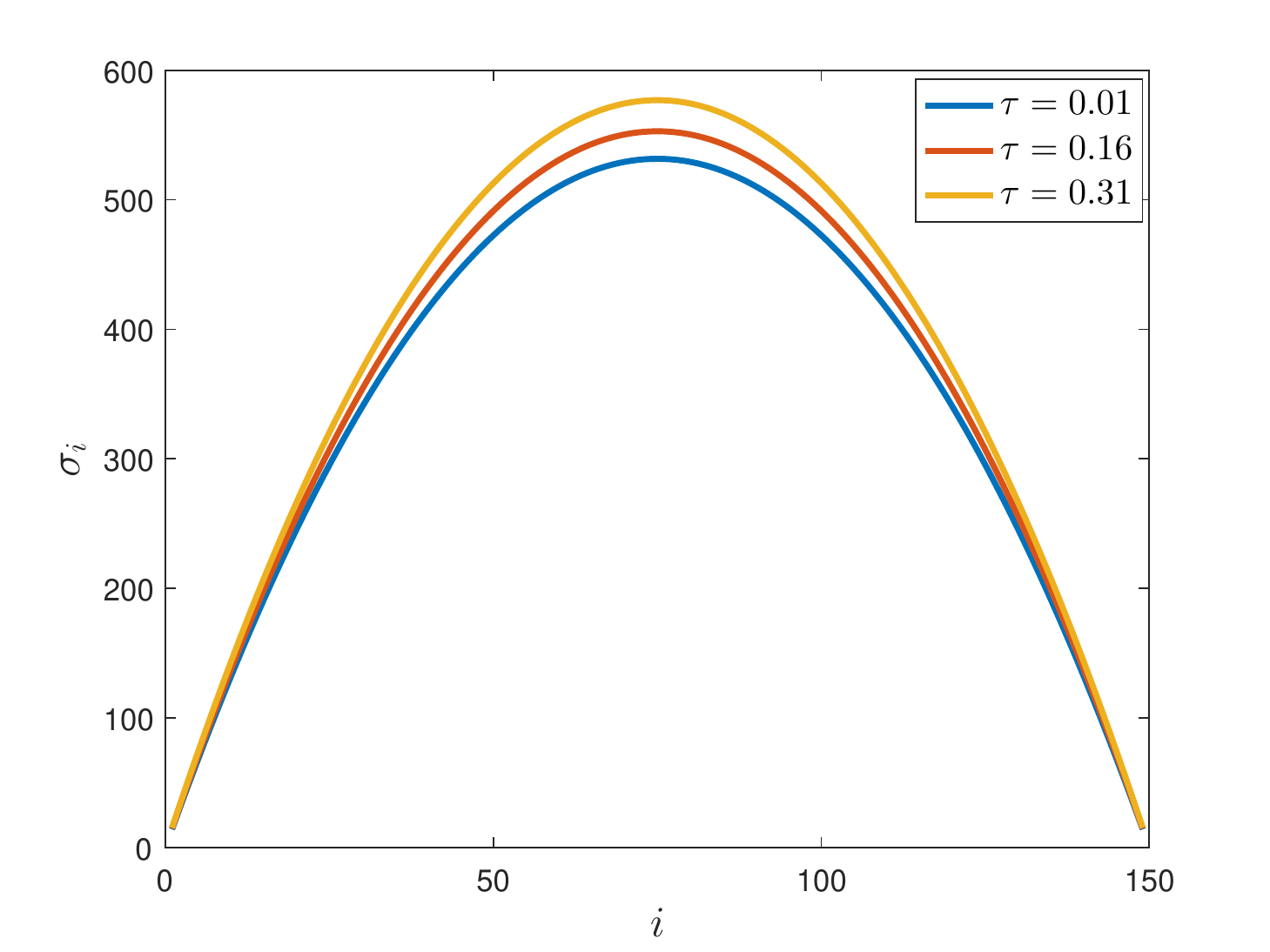}
\includegraphics[trim= 20 0 30 10, clip, width=6cm]{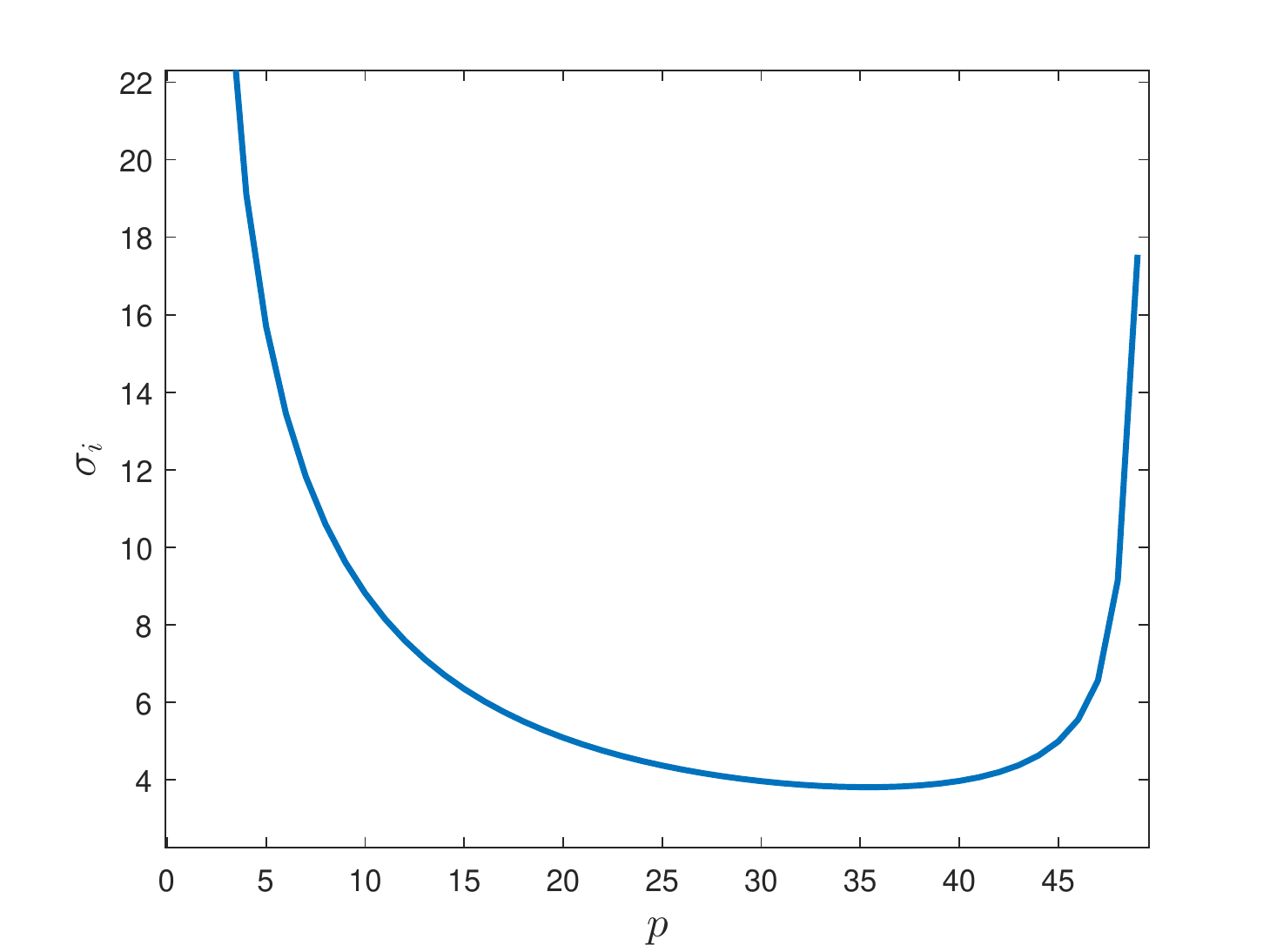}\caption{The plot of marginal variance for the complete, the path, and the $p$-cycle graphs. } \label{fig: topologicalgraph}
\vspace{-10pt}
\end{figure*}

Calculating explicit closed-form solutions for delay differential equations is almost impossible. This is also true for the delay model of platoon \eqref{eq: sys0} and its transition matrix $\Phi$. It was shown that the covariance matrix $\Sigma_{\infty}$ depends on $\Phi$ as in \eqref{eq: cov}. In fact,  obtaining an explicit expression for the solution of the unperturbed platoon heavily depends on the form of function $f: S\rightarrow \mathbb R_+$ as in \eqref{eq: functionf}. Although $f(s_1,s_2)$ has a closed-form, it does not admit an explicit form. Consequently, the formulas in \eqref{eq: riskcolformula} and \eqref{eq: riskampformula} are expressed in terms of improper integrals. This induces a computational burden that quickly becomes an issue as the number of vehicles in the platoon increases. Thus, it is desirable to find efficient  approximations of $f(s_1,s_2)$ over its domain $S$ in order to reduce computational complexity of our proposed methodology.  Our investigation reveals that the behavior of $f(s_1,s_2)$, depending on its parameters, is similar to some non-trivial rational function for which the existing conventional approximation techniques (e.g., Legendre polynomials) are proven to be inefficient. In the following, we first propose 
a rational approximation for $f(s_1,s_2)$ along with its relative error bound. The details of our derivations can be found in Appendix A. Then, this approximation is applied to obtain an efficient approximation of the risk measures.


We recall that $f(s_1,s_2)$ diverges on the boundary of $S$ for $s_2\neq 0$ where $f$ meets its poles. On the $x$-axis, over which $s_2=0$ and  $f$ is finite, the dynamics of the unperturbed platoon seize to satisfy  Definition \ref{def: platooning}. Thus, we will approximate $f$ in a compact subset of $S$ that excludes any pole or degeneracy. We adopt a  compact subset  of $S$ that is characterized by $\overline{S}=[0.1,s^*-0.05]\times [0.1, 0.9]$, where $s^*$ is the solution of $s^* \cot(s^*)=s_2$ for $s_2\in [0.1,0.9]$. This subset is depicted in Figure \ref{fig: stability2} along  with $S$. One can verify that $s^* \cot(s^*)=s_2$   is invertible for $s_2 \in [0.1,0.9]$ and it can be expressed as $s^*(s_2)$.   For $s_1, s_2\in \overline{S}$, we choose the following class of rational  functions \begin{equation}\label{eq: equationsimpler}\begin{split}
\tilde{f}(s_1,s_2)&= \frac{A_{s_2}(s_1)}{s_1^2}+\frac{B_{s_2}(s_1)}{s_1-s^*(s_2)}
\end{split}
\end{equation} in which the enumerators \begin{equation*}
\begin{split}
A_{s_2}(s_1)=\sum_{i=0}^5 a_i(s_2) s_1^i~~~\textrm{and}~~~
B_{s_2}(s_1)=\sum_{i=0}^4 b_i(s_2) s_1^i
\end{split}
\end{equation*} are polynomials in $s_1$  with real-valued coefficients that exclusively depend  on $s_2$. Functions $a_i$ and $b_i$ can not be expressed in closed-forms. Some typical graphs of these functions are shown in Figures  \ref{fig: coef1}. Figure \ref{fig: coef2} depicts the exact function $f(s_1,0.5)$ and the rational approximation $\tilde{f}(s_1,0.5)$ for $s_1 \in \overline{S}$ together with the associated relative error. The relative error function
 $$\eta(s_1,s_2)=\bigg| 1- \frac{\tilde{f}(s_1,s_2)}{f(s_1,s_2)} \bigg|$$ 
is plotted in Figure \ref{fig: relativerror3} for all $s_1,s_2 \in \overline{S}$. Our extensive numerical experiments verifies that $\max_{s_1,s_2\in \overline{S}} \eta(s_1,s_2) = \mathcal O (10 ^{-4}).$ Moreover, computational explorations suggest that $\overline{S}$ contains the minimum value of $\tilde{f}(s_1,s_2)$. This is a significant indication in favor of utilizing $\tilde{f}$ as a computationally efficient surrogate in order to develop efficient algorithms for design of low-risk platoons.

%

In the final step, we utilize approximation  \eqref{eq: equationsimpler} and arrive at a tight  approximation of the risk measure $\mathcal {\tilde{R}}_{\varepsilon}^{C/A} = \mathcal R_{\varepsilon}^{C/A}(\tilde{\sigma}_i)$, where 
\begin{equation*}
\tilde{\sigma}_{i}=\frac{|g|\tau^\frac{3}{2}}{\sqrt{2\pi}}\sqrt{\sum_{j=2}^n\big[(\mathbf e_{i+1}-\mathbf e_i]^T\mathbf q_j\big)^2 \bigg[\frac{A_{\beta \tau}(\lambda_j\tau)}{(\lambda_j\tau)^2}+\frac{B_{\beta \tau}(\lambda_j\tau)}{\lambda_j\tau-s^*(\beta \tau)}\bigg]}
\end{equation*} provided $(\lambda_j\tau,\beta\tau)\in \overline{S}$ for $j=2,\dots,n-1$.

\section{Examples Of Communication Topologies}\label{sect: topological}
Using the results of the previous section, we obtain approximate closed-form expressions for the risk of systemic events in platoons with 
certain symmetric communication topologies. 
The marginal variance is evaluated for platoons with complete, path, and $p$-cycle communication graphs with uniform feedback gains, i.e., $k_{ij} \equiv k$. The eigenvalues and eigenvectors of these graphs can be obtained explicitly \cite{Mieghem:2011:GSC:1983675,Gray:2005:TCM:1166383.1166384}. We restrict our attention to calculation of the marginal standard deviation of the relative distance between two successive vehicles. 
Through our analysis, it is possible to calculate the value-at-risk measures of the collision and  detachment events for a given confidence level  $\varepsilon$ and set of parameters $a, c, d$ and $h$.

\subsection{The Complete Graph}The eigenvalues  of a complete graph are: $\lambda_1=0$ and $\lambda_j=k n$ for all $j \in \{2,\dots,n\}$. For $(k n \tau,\beta \tau)\in S$, the marginal standard deviation for the complete graph topology is
\begin{equation*}
\sigma_i = |g|\Sp \tau^{\frac{3}{2}} \Sp \sqrt{\frac{f(k n \tau,\beta \tau)}{\pi}}
\end{equation*} 
for all $i \in \{1,\dots,n-1\}$. Using this quantity,  
one can easily calculate the value-at-risk measures. The first plot from left in Figure \ref{fig: topologicalgraph} illustrates $\sigma_i$ as a function of time-delay $\tau$,   number of vehicles, network parameters $g=1$, $k=0.3$, and $\beta=0.1$. We conclude that for small time-delay, larger ensembles of vehicles experience lower risk. For large value of $\tau$, it appears that the smaller the ensemble, the safer the platoon.    

\subsection{The Path Graph} The path graph over $n$ nodes has $n-1$ edges of $k$ weight. The $j^{th}$ eigenvalue is $\lambda_j=2k\big(1-\cos(\pi(j-1)/n)\big)$ with the corresponding eigenvectors $\mathbf q_1=\frac{1}{\sqrt{n}} \mathbbm 1$ and $\mathbf q_j=\big[q_j^{(1)},\dots,q_j^{(n)}\big]^T$ with $q_j^{(l)}=\sqrt{\frac{2}{n}}\cos\big[ \frac{\pi (n-j+1)}{2 n}(2l-1)\big]$ if $j \in \{2,\dots,n\}$. The marginal standard deviation is
\begin{equation*}
\sigma_i= \frac{2|g|\tau^{\frac{3}{2}}}{\sqrt{n\pi}}\sqrt{\sum_{j=2}^n w(j,i,n) f(\lambda_j \tau,\beta \tau )}
\end{equation*} with 
$w(j,i,n)=\sin^2\big(\frac{\pi (n-j+1) }{n}i\big) \sin^2\big(\frac{\pi (n-j+1)}{2n}\big)$. Then, one can easily calculate the risk measures. The second plot in Figure \ref{fig: topologicalgraph} illustrates $\sigma_i$ with respect to vehicle labels, where it is assumed that communication graph in all simulations is path with network parameters $g=1$, $k=0.3$, and $\beta=0.1$. We conclude that the safest regions with lowest risk of  collision or detachment are located in the two ends of the platoon. As we approach the center, the risk of collision or detachment will increase monotonically and reach its peak half way before it begins decreasing again. This implies that the vehicles in the middle of the platoon are more likely than the others to experience collision or detachment. 
It is  observed that time-delay uniformly increases the risk across the platoon.

\subsection{The $p$-Cycle Graph}
A platoon with a $p$-cycle communication graph is a network where each vehicle communicates with its $p$-immediate neighbors. The corresponding Laplacian matrix is a special type of circulant matrices whose eigen-structure is discussed in \cite{Gray:2005:TCM:1166383.1166384} from where it can be shown that
$\lambda_1=0$ and $\lambda_j=k\left(2p+1-\frac{\sin((2p+1)(j-1)\pi/n)}{\sin((j-1)\pi/n)}\right)$ with $([\mathbf e_{i+1}-\mathbf e_i]^T\mathbf q_j)^2=\frac{2}{n}\big(1-\cos\big(\frac{2\pi}{n}\big) \big)$ for all $j \in \{2,\dots, n\}$. 
The marginal standard deviation of the distance between vehicles $i+1$ and $i$ is 
\begin{equation*}
\sigma_i= \frac{|g|\tau^{\frac{3}{2}}}{\sqrt{n \pi}}\sqrt{\sum_{j=2}^n \bigg(1-\cos\bigg(\frac{2\pi (j-1)}{n} \bigg) \bigg) f(\lambda_j \tau,\beta \tau)}.
\end{equation*}
The third plot from left in Figure \ref{fig: topologicalgraph} is a graphic illustration of the marginal variance over a platoon with $n=101$ vehicles as a  function of parameter $p \in [2,50]$. The network parameters are $k=0.0211$, $\beta=1$, $\tau=0.5$, and $g=1$. For small value of $p$, i.e., for loosely connected platoons, the marginal variance, which is identical for all vehicles, is large. As connectivity enhances by increasing $p$, the platoon becomes less fragile to systemic events of collision and detachment. Finally, when $p$ approaches the limit value $50$,  i.e., the complete graph topology, the eigenvalues approach the boundary of $S$ and the platoon becomes unstable. The authors acknowledge that it is an unrealistic to have vehicles to communicate over a cycle graph as the first vehicle may not be able to communicate with the last. However, the $p$-cyclic graphs serves as a nice approximation for the $p$-path graph when $n$ is large enough.  For $n \gg 1$, the $p$-cycle graph resembles a graph where every vehicle communicates with its $p$ nearest neighbors from each side (front and behind). 

\begin{figure}
\includegraphics[trim= 10 0 30 10, clip, width=4.4cm]{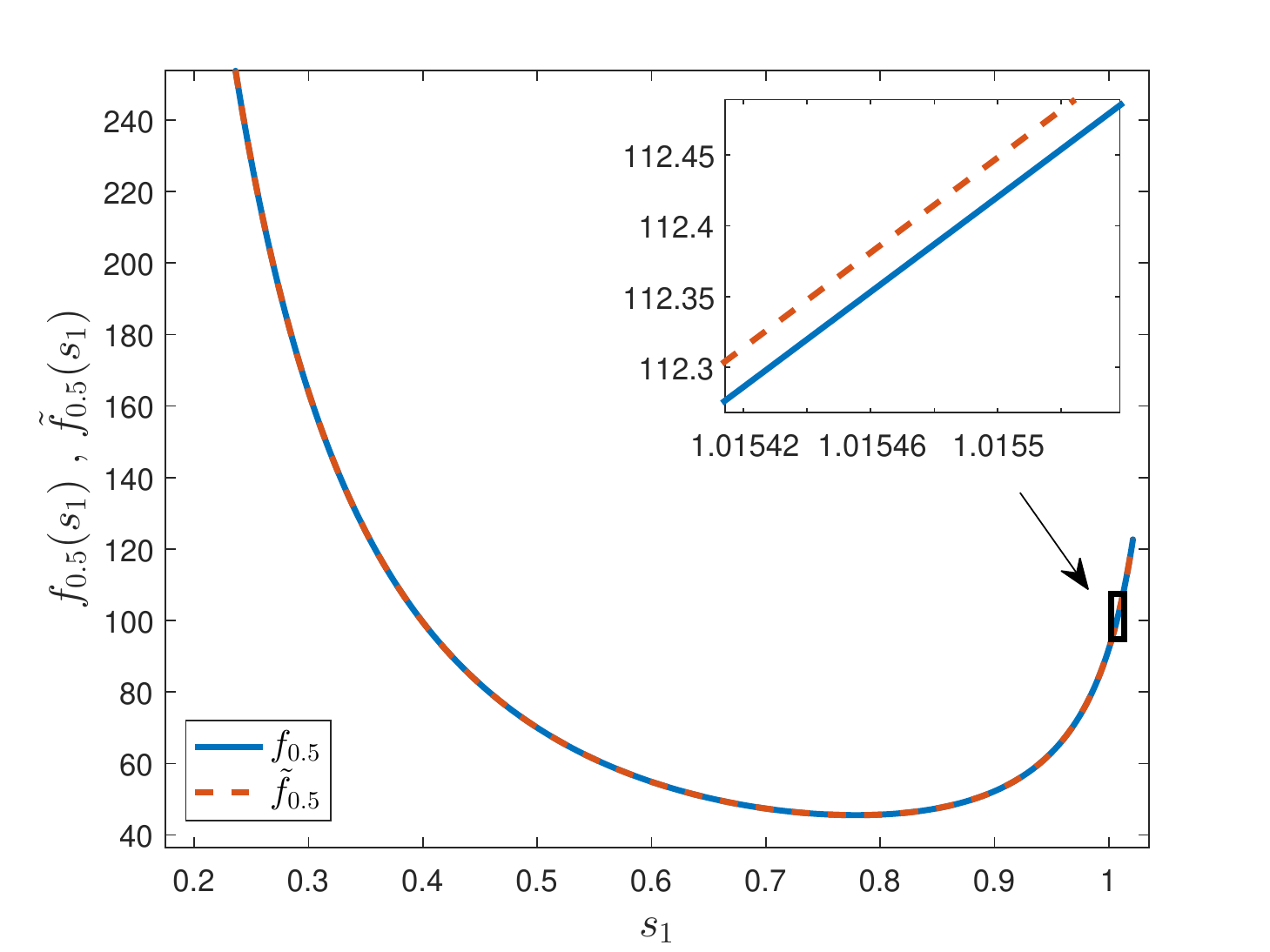}
\includegraphics[trim= 10 0 30 10, clip, width=4.4cm]{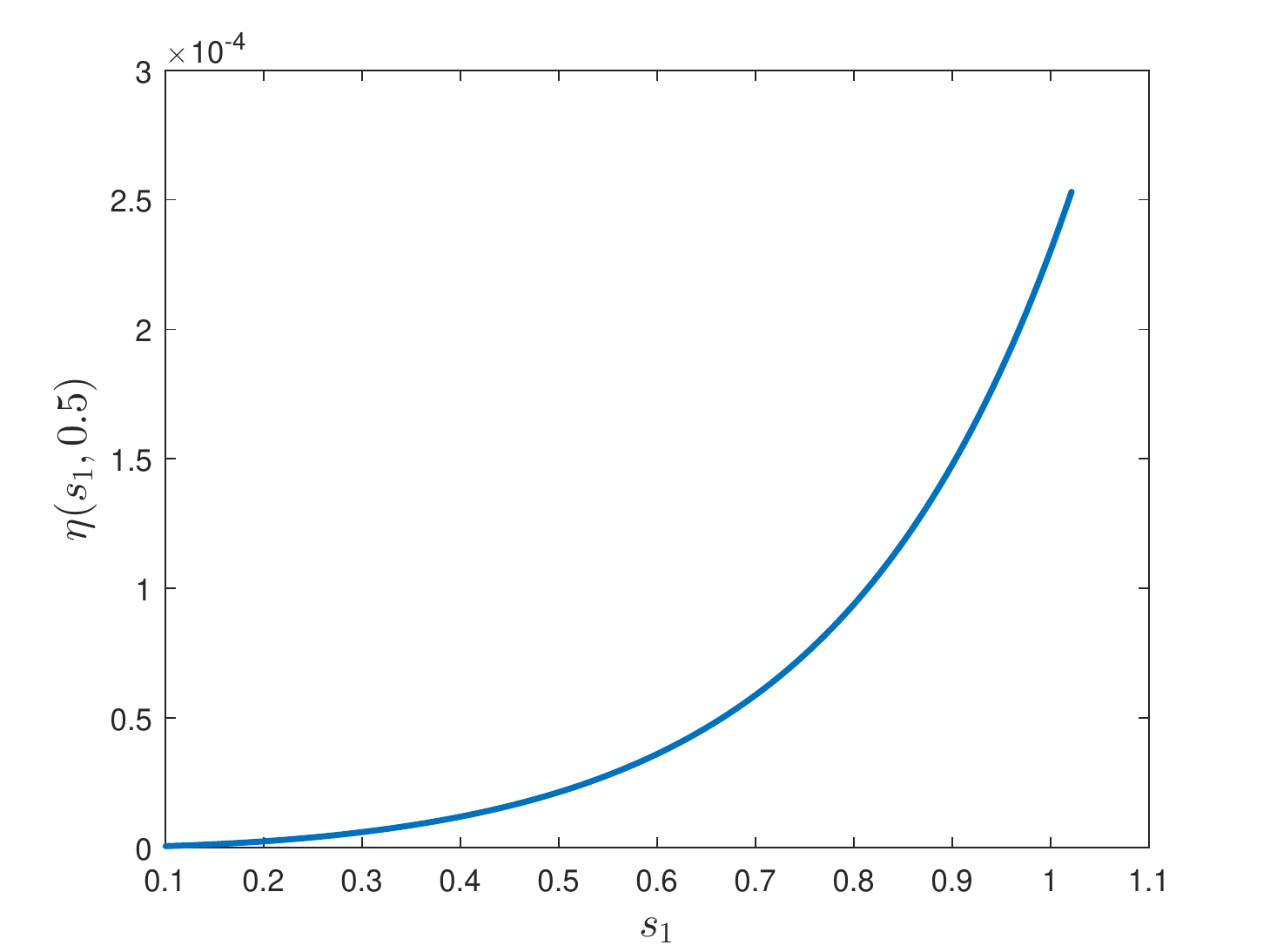}\caption{ (Left) The left plot depicts $f(s_1,0.5)$  and its approximation $\tilde{f}(s_1,0.5)$ over $[0.1,s^*(0.5)-0.05]$. The right plot is showing  the relative error.}\label{fig: coef2}
\end{figure}

\section{Simulations}\label{sect: simulation} We discuss three case studies for platoons with dynamics governed by \eqref{eq: sys0}. First, we show that time-delay can steer a platoon to become more prone to risk and eventually instability. The second case study illustrates that, in the presence of time-delay, deviating from the optimal graph topology increases the risk of systemic events. Finally, the third case examines how spatial localization of communication may affect the risk measures. 

\subsection{Risk Behavior vs. Connectivity}\label{example1} We consider  a platoon of $n=10$ vehicles. The desired distance is set to $d=1$, the scale between position  and velocity alignment to $\beta=1/3$, and the drift coefficient of noise to $g=2$. The other parameters are $c=1$ and $\varepsilon=0.01$ for the collision events, and $a=2,~h=1$,  $\varepsilon=0.05$ for the network detachment events. The desired relative distance between two successive vehicles $i+1$  and $i$ is $1$. The vehicles are in  collision when $x_t^{(i+1)}-x_t^{(i)}\leq 0$ and the vehicles have lost connectivity when $x_t^{(i+1)}-x_{t}^{(i)}\geq 2$. The objective is to guarantee that vehicles 6 and 5 will neither collide nor lose connectivity. The corresponding values of risk measures are 
\begin{equation*}\begin{split}
\mathcal R_{0.01}^{C,5}(\sigma_5)&=\begin{cases}
\frac{1}{1-0.4299\, \sigma_5}-1, & \sigma_5<0.4299\\
\infty,& \sigma_5\geq 0.4299,
\end{cases}\\  \mathcal R_{0.05}^{D,5}(\sigma_5)&=\begin{cases}
\frac{1}{1-0.6080, \sigma_5}-1, & \sigma_5<0.6080\\
\infty,& \sigma_5\geq 0.6080
\end{cases}
\end{split}
\end{equation*}

Topology of the communication graph is generated randomly by ensuring  Assumption 1 and  all feedback gains are equally chosen to be $r\geq 1$.  The first round of simulations assumes a platoon without time delay. For $r=1$, it turns out that $\sigma_5$ is greater than the critical values $\frac{d}{\kappa_{0.01} \sqrt{2}}=0.4299$ and $\frac{d}{\kappa_{0.05} \sqrt{2}}=0.6080$ that results in $\mathcal{R}_{0.01}^{C,5}=\mathcal{R}_{0.05}^{D,5}=\infty$. By tuning up the feedback gain to $r=1.9$, the risk measures become $\mathcal R_{0.01}^{C,5}=4.5285$ and $\mathcal R_{0.05}^{D,5}=3.8547$, which are rather large, but finite. Increasing $r$ to $5$ yields $\mathcal R_{0.01}^{C,5}=0.4519$ and $\mathcal R_{0.05}^{D,5}=0.2998$. This is to verify that in the absence of time-delay, the risk of systemic events becomes small as connectivity improves.  

Next, we consider a platoon with time-delay $\tau=0.07$ and repeat our simulations for different values of the feedback gain $r$. At $r=1$, it turns out that $\mathcal{R}_{0.01}^{C,5}=\mathcal{R}_{0.05}^{D,5}=\infty$. As we increase feedback gain up to $r=1.74$, the risk measure decreases to $\mathcal R_{0.01}^{C,5}=0.754$ and $\mathcal{R}_{0.05}^{D,5}=0.348$, before it starts to increase again. At $r=1.8401$, it goes back to $\mathcal{R}_{0.01}^{C,5}=\mathcal{R}_{0.05}^{D,5}=\infty$. Examples of the relative distance between vehicles 6 and 5 are illustrated in Figure \ref{fig: simulation1} over time interval $[0,100]$ seconds.

\begin{figure}
\includegraphics[trim=10 0 30 18, clip, width=4.5cm]{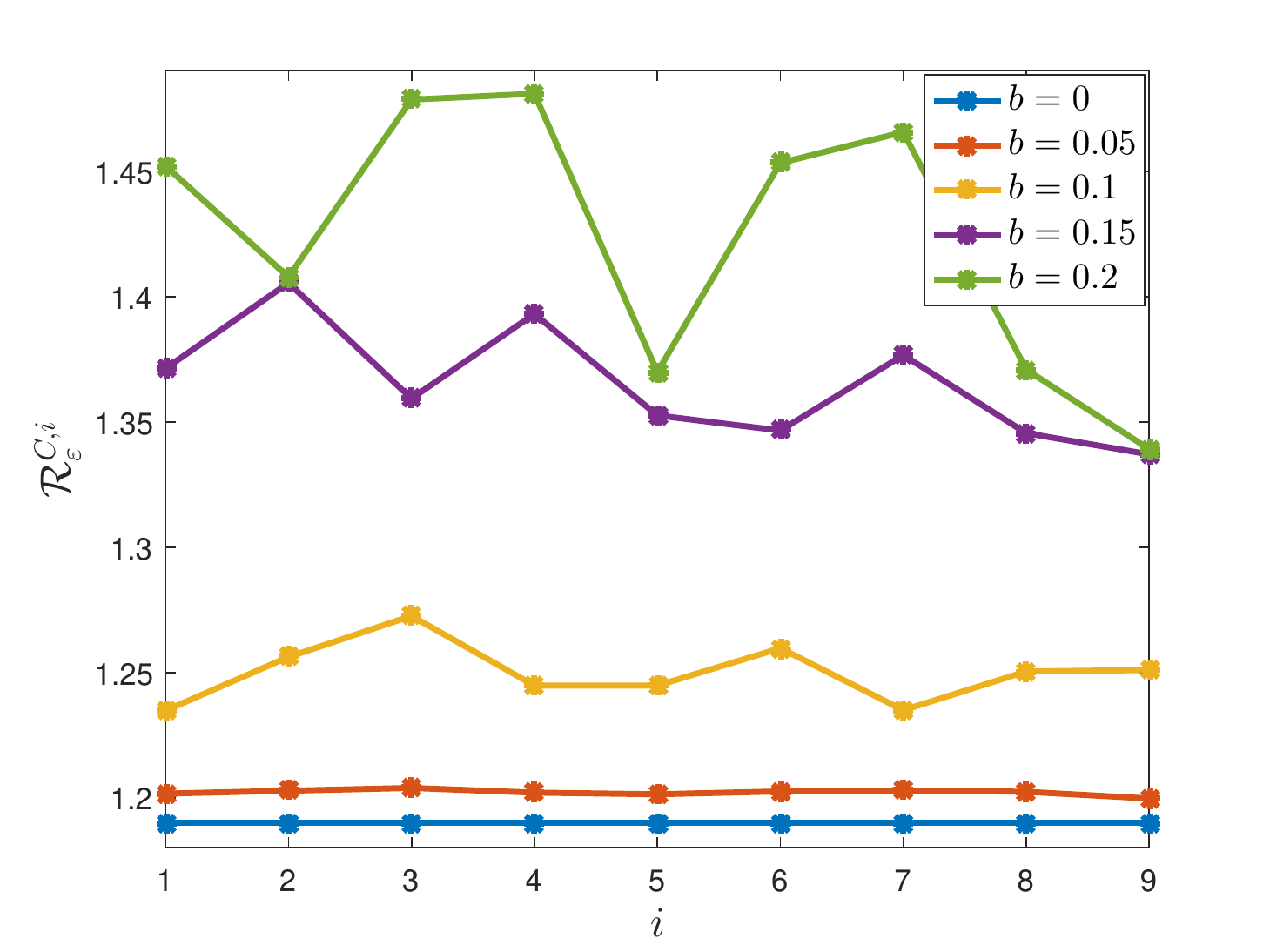}
\includegraphics[trim=10 0 30 18, clip, width=4.5cm]{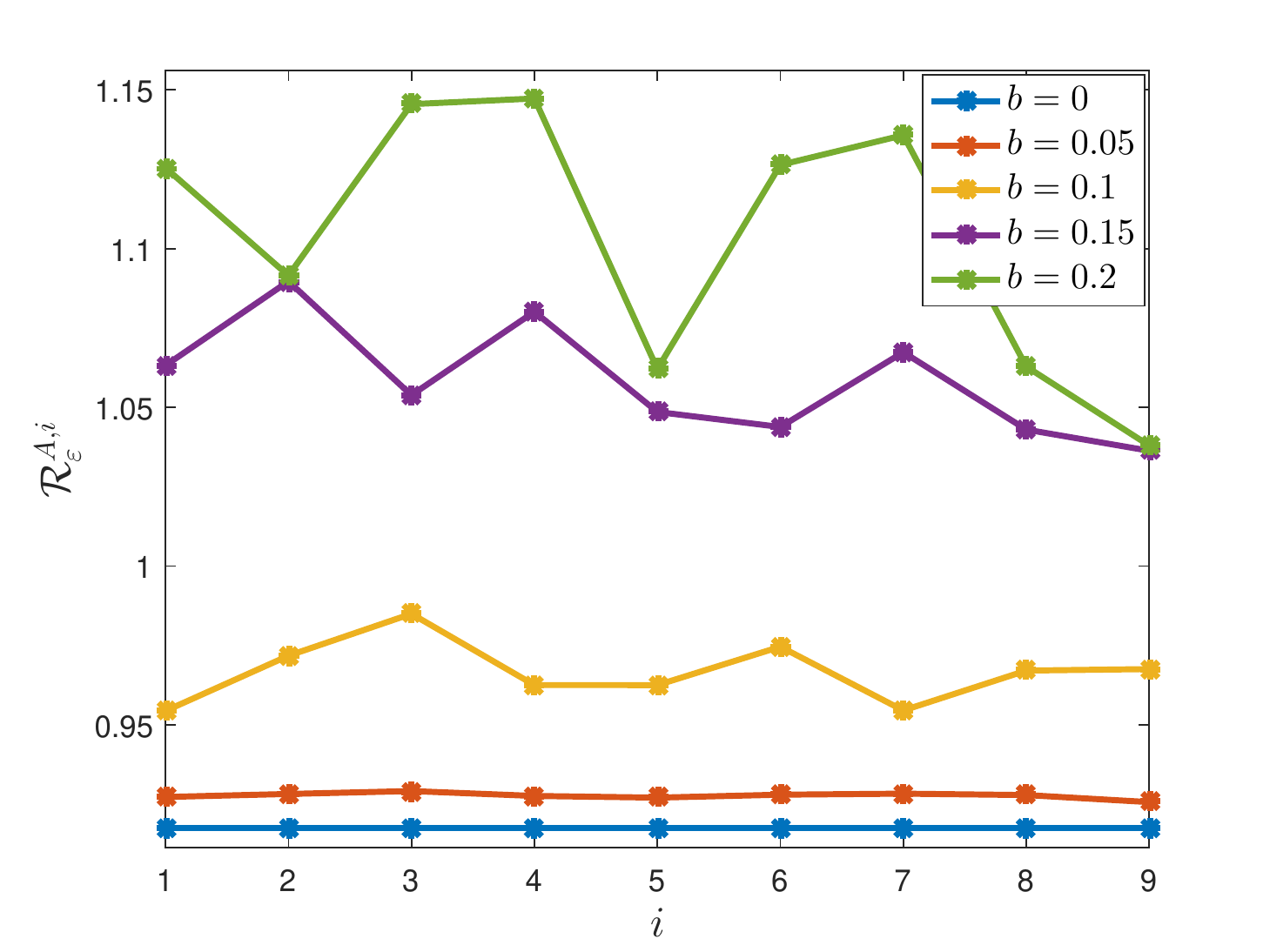}\caption{Example of randomly perturbed graphs. Randomness increases with parameter $b$. At $b=0$ we have the optimal graph topology. Risk curves are almost the same for systemic events of collision and detachment.}\label{fig: perturb}
\end{figure}

\subsection{Random Perturbations in Optimal Graphs}
In this case study, we investigate how deviation from an optimal graph topology affects the risk of collision and detachment by considering a platoon with $n=10$ vehicles. The network parameters are $g=1.5$, $\tau=0.1$, $d=0.5$, and $\beta=2.2$. For the collision events, we have set $c=1.5$ and a cut-off value $\varepsilon=0.05$. For the detachment events, we have set $a=2$, $h=3$, and a cut-off value $\varepsilon=0.1$. We recall from Section \ref{sect: tradeoff} that the optimal communication graph is the one with $\beta\tau\approx 0.220$ and $\lambda_j\tau\approx 1.111$ for $j=2,\dots,n$. An optimal graph for our example is the complete graph with identical link weights $k_{ij}^*=1.111$ for all $i,j \in \{1,\dots,n\}$. We generate random perturbation of the optimal topology by substituting feedback gain $k_{ij}^*$ with new feedback gain $k_{ij}^*+b\cdot \xi$, where $\xi$ is a random variable that is uniformly distributed in $(0,1)$. Figure \ref{fig: perturb} illustrates the value of the elements of vectors $\mathfrak{R}_{0.05}^{C}$ and $\mathfrak{R}_{0.1}^{D}$ for different values of parameter $b$. We remark that the risk of collision qualitatively behaves similar to the risk of detachment, although we numerically verified that they are not proportional. In addition, the risk of detachment is significantly smaller for all values of $b$ due to the difference in the cut-off values.

\begin{figure}
\includegraphics[trim=40 10 35 10, clip,width=9cm]{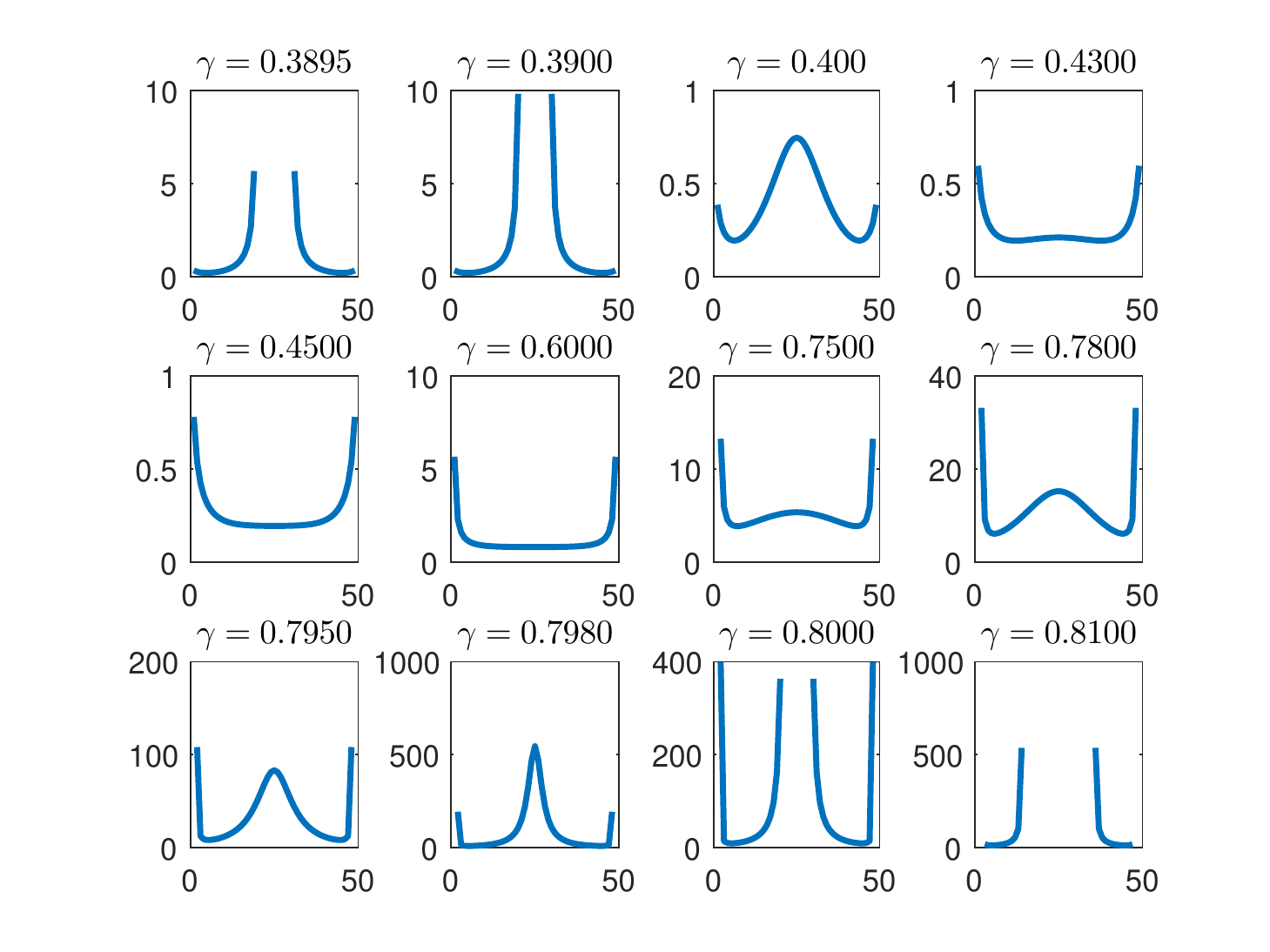} \caption{Phase transition of the risk measure for various exponents of $\gamma$. The $x$-axes of all figures represent the inter-distance between vehicles. The $y$-axes represent the risk of collision. The network parameters are $g=0.8$, $\beta=1$, $d=1$, $c=1.5$, $\tau=0.18$ and $\varepsilon=0.05$.}\label{fig: spatialdecay}
\end{figure}

\subsection{Spatially Decaying Topologies} In this case, we consider geometric graphs where all-to-all communication is allowed by enforcing  the range of connectivity to decrease with distance. This class of communication graphs is typical in  wireless communication networks. Each vehicle broadcasts its message, where signal-to-noise ratio of the received signal by another vehicle decreases with distance between the two vehicles.  We consider a platoon with $n=50$ vehicles with usual ascending labels  from $1$ to $50$. The receiver of the $i$'th vehicle collects the state information of the $j$'th vehicle with feedback gain  $k_{ij}=1.5\, e^{-\gamma |i-j|}$ for some $\gamma\geq 0$. The exponent $\gamma$ is the spatial decay index or localization parameter.  For small values of $\gamma$, the network topology approximates a complete graph. As $\gamma$ increases, the effective communication range of vehicles becomes more localized. For large enough  values of $\gamma$ that preserve connectivity, the communication topology resembles the path graph. In Figure \ref{fig: spatialdecay}, we illustrate $\mathfrak{R}_{0.05}^C$ for different exponents and observe various transition of the risk profile. For small $\gamma$, the graph is heavily connected, which combined with the effect of time-delay results in increased risk values. The high risk vehicles are the ones in the middle. As $\gamma$ increases, connectivity decreases and the platoon experiences lower risk of collision. It is remarkably interesting that the vehicles in the middle become very safe.  For larger values of $\gamma$, the communication network keeps losing connectivity and risk of systemic events  becomes more evident. The more susceptible vehicles for $\gamma\geq 0.85$ are again the ones in the middle. As $\gamma$ exceeds $1$, the communication gets very localized. Since $n$ is large, poor connectivity makes the platoon susceptible to noise, leading to infinite value of risk for many pairs of vehicles.

\section{Discussion} We focused on collision or detachment events between vehicles in a platoon, where each vehicle is modeled as a double integrator subject to exogenous  noise. The ensemble is controlled via a distributed consensus feedback control law that suffers from communication time-delay. We developed a risk oriented framework to assess the possibility of some relevant systemic events. Our model, albeit simple,  facilitates a rigorous analysis to the point of characterizing intrinsic interplay among the risk measures,  network connectivity,  time-delay, and statistics of the exogenous uncertainty. Our results  are particularly useful to design low-risk platoons by  optimizing topology of the underlying communication network. On the downside, we acknowledge a number of shortcomings that pose new challenges for future research.

We would like to point out that it is possible to study risk on more general models. Risk measures ask for distributional properties of stochastic processes (in our case, solutions of a stochastic differential equation). It is well-known that evolutionary densities of stochastic dynamical systems typically satisfy partial differential equations of Fokker-Planck type. This means that, at least in theory, it is feasible to evaluate risk of systemic events for a broad class of dynamical systems. There is, however, little hope for closed-form solutions such as the ones established in this work. Analysis may only be feasible on a case-by-case basis. Linearity helps us to leverage the Ornstein–Uhlenbeck type of processes with normal distributions. Risk analysis provides a solid framework to study the manner with which marginal variances that encapsulate features of the dynamical system (network topology, time-delay, noise) characterize systemic fragility as they project on the geometry of undesirable events. 

An interesting extension is to consider network with heterogeneous delays. Uniform time-delay is from many aspects simplistic. However, heterogeneity of time-delay destroys symmetry. A dire consequence is that it becomes quite challenging to identify stability region $S$ i.e., a crucial step for risk analysis. Numerical methods can be employed to overcome this hurdle \cite{kharitonov2012time}.

The two systemic events considered in this work are the event of vehicles collision and the event of vehicles detachment. Undesirable (unsafe) behavior is triggered when two successive vehicles are either too close or too far away.  The questions that fell within our scopes were on characterizing the types of interactions that promote robustness with respect to collision or detachment when real-world deficiencies, such as disturbance and time-delay, are present.   One future extension is to consider internal structure (e.g., engine dynamics) and non-trivial shapes of vehicles. It would be interesting to explore how vehicle dynamics integrate with network assimilated data and distributed control laws for motion coordination. To this end, we would like to mention one important feature of vehicle dynamics, which is omitted in this work. That is acceleration damping, i.e. a stabilizing state-feedback control term on the rate of change of velocity of vehicles that ensures finite speed of the ensemble. In addition, absolute feedback control methods are in many cases known to generally improve performance \cite{Bamieh12,4282756}. It would be very interesting to test these control policies from a risk theory perspective.

A final future research direction that is worth mentioning is to develop efficient and scalable algorithms to design communication topologies for platoons by striking a balance among connectivity, performance, and risk of collision or detachment events.



\begin{figure}\center
\includegraphics[trim=10 0 20 10, clip, width=8cm]{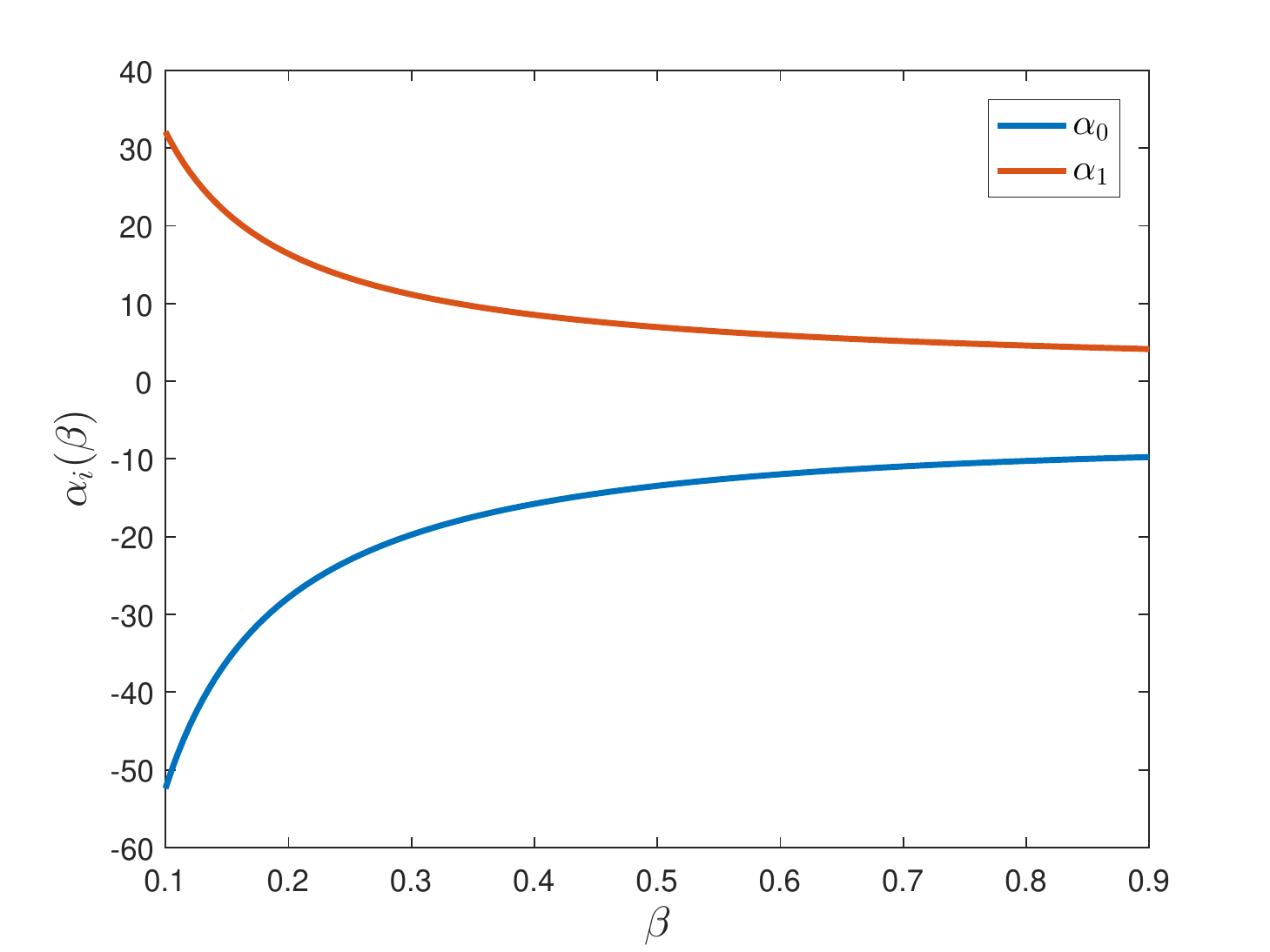}\caption{Plot of coefficients $\alpha_0$ and $\alpha_1$ as functions of or $\beta$.}\label{fig: coeffs1and2}
\end{figure}

\section*{Appendix A: Rational Approximation Of $f(s_1,s_2)$} 

The function $f(s_1,s_2)$ in \eqref{eq: functionf} plays an instrumental role in the actual calculation of risk. As it is not explicitly stated, we can approximate $f(s_1,s_2)$ for $s_1$ and $s_2$ in a compact subset $\overline{S}$ of $S$, defined as: 
$$\overline{S}=[0.1,s^*-0.05]\times [0.1, 0.9]$$ where $s^*~:~\frac{s^*}{\tan(s^*)}=s_2,~s_2\in [0.1,0.9],$ and it is presented in Figure \ref{fig: stability2}, together with $S$.  For $s_1,s_2\in \overline{S}\subset S$, we will construct a rational approximation of $f(s_1,s_2)$. Our approach relies on ideas developed in \cite{pranic2013}. 

Now, we see that for fixed $s_2\in (0,1)$ the function $f_{s_2}(s_1)=f(s_1,s_2)$ attains a pole at $s_1=0$ of order 4 and a pole at $s_1=s^*=s^*(s_2)\in (0,\pi/2)$ of order 1. In fact, the collection of all poles $s^*(s_2),~s_2\in(0,1)$ is the curved boundary of $S$. The zero poles lie along the vertical axis. For fixed $s_2\in [0.1,0.9]$ we recall $s^*=s^*(s_2)$ and consider the vector space $\mathcal T_{s_2}$ spanned by the functions
\begin{equation*}
\mathcal T_{s_2}=\bigg\{ 1, ~s_1, ~\frac{1}{s_1},~ s_1^2, ~\frac{1}{s_1^2}, ~s_1^3, ~\frac{1}{s_1-s^*},~s_1^4\bigg\}
\end{equation*} The inner product $$\langle g_1,g_2 \rangle =\int_{0.1}^{s^*-0.05}g_1(t)g_2(t)\,dt$$ for any $g_1,g_2 \in \mathcal T_{s_2}$, will be used to generate an orthonormal basis of $\mathcal T_{s_2}$ following the Gram-Schmidt process. We arrive at 
\begin{equation*}\begin{split}
\mathcal Q_{s_2}=\big\{\psi_{s_2}^{(0)}(s_1),\psi_{s_2}^{(1)}(s_1),\psi_{s_2}^{(2)}(s_1),& \psi_{s_2}^{(3)}(s_1),\psi_{s_2}^{(4)}(s_1),\\ 
& \psi_{s_2}^{(5)}(s_1),\psi_{s_2}^{(6)}(s_1),\psi_{s_2}^{(7)}(s_1)\big\}.
\end{split}
\end{equation*}  We introduce
\begin{equation*}\label{eq: firstequation}
P(s_1,s_2)= \sum_{k=0}^7 w_k \psi_{s_2}^{(k)}(s_1).
\end{equation*} with the weights $w_k=w_k(s_2)=\langle f_{s_2},\psi_{s_2}^{(k)}\rangle$. From the orthonormalization process,  $\psi_{s_2}^{(k)}$ are linear combinations of elements of $\mathcal T_{s_2}$. Thus we can write $P$ as
\begin{equation*}
\begin{split}
P(s_1,s_2)&=\frac{\sum_{k=0}^7\alpha_k(s_2)s_1^k}{s_1^2\,\big(s_1-s^*(s_2)\big)}
\end{split}
\end{equation*}
for $\alpha_k$ that generally depend on $s_2\in [0.1,0.9]$. The first two coefficients, $\alpha_1$ and $\alpha_2$ are illustrated in Figure \ref{fig: coeffs1and2}. Numerical explorations show that $\alpha_k,~k=2,\dots,7$ attain constant values  $\alpha_2 \approx -0.0742 $, $\alpha_3 \approx 0.0198 $, $\alpha_4 \approx -0.0036 $, $\alpha_5 \approx 0.0008 $, $\alpha_6 \approx -10^{-4} $, $\alpha_7 \approx 10^{-6}$.\footnote{The aforementioned coefficients appear to be approximately constant for a large range of values $s_2$ in $[0.1,0.9]$. In fact they all turn to be non-constant with non-smooth behavior for values near to $0.9$. We approximated $\{\alpha_i \}_{i=2}^7 $ with their average in $[0.1,0.9]$.} We can discard the terms $\alpha_5,\alpha_6,\alpha_7$, for being of negligible magnitude and this yields 
\begin{equation*}
\begin{split}
\tilde{f}(s_1,s_2)&=\frac{-0.0036 s_1^4+0.0198s_1^3-0.0742s_1^2+\alpha_1(s_2)s_1+\alpha_0(s_2)}{s_1^2(s_1-s^*(s_2))}\\
&=\frac{q_{s2}(s_1)}{s_1^2(s_1-s^*(s_2))}.
\end{split}
\end{equation*}

Elementary algebra yields
\begin{equation*}\label{eq: equationsimpler2}\begin{split}
\tilde{f}(s_1,s_2)&=- \frac{ q_{s_2}(s_1)}{(s^*(s_2))^2}\frac{s_1+ s^*(s_2)}{s_1^2}+\frac{q_{s_2}(s_1)}{(s^*(s_2))^2}\frac{1}{s_1-s^*(s_2)}\\
&=\frac{A_{s_2}(s_1)}{s_1^2}+\frac{B_{s_2}(s_1)}{s_1-s^*(s_2)}
\end{split}
\end{equation*} where 
\begin{equation*}
\begin{split}
A_{s_2}(s_1)&=-\frac{1}{s^*(s_2)}q_{s_2}(s_1)\big(s+s^*(s_2)\big)=\sum_{i=0}^5 a_i(s_2) s_1^i, \\
B_{s_2}(s_1)&=\frac{q_{s_2}(s_1)}{(s^*(s_2))^2}=\sum_{i=0}^4 b_i(s_2) s_1^i
\end{split}
\end{equation*}
The accuracy of $\tilde{f}(s_1,s_2)$ is validated numerically. The results are depicted in Figure \ref{fig: relativerror3}. The maximum relative  error is of order $10^{-4}$. 
\begin{figure*}[t]
\includegraphics[scale=0.39]{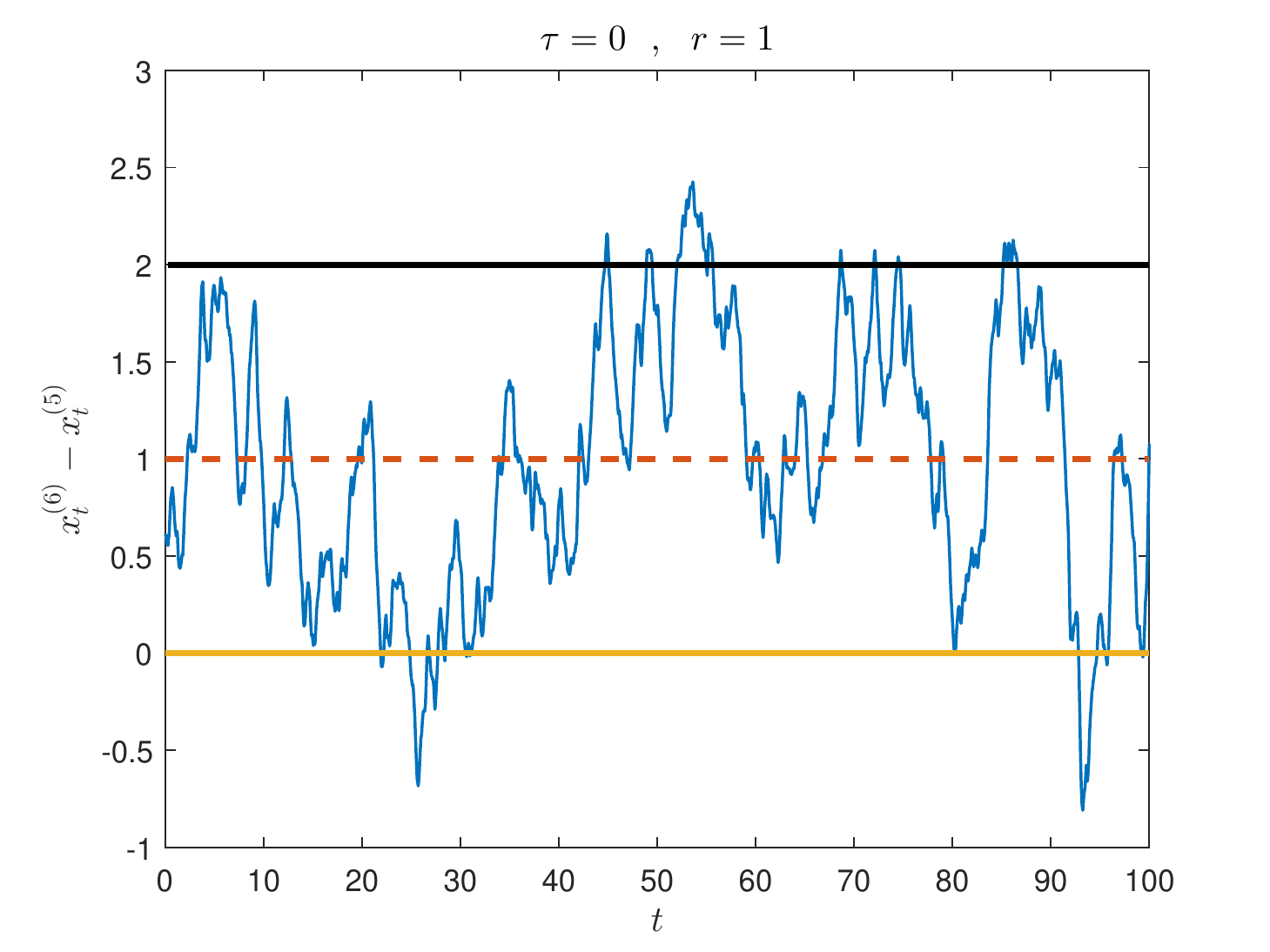}
\includegraphics[scale=0.39]{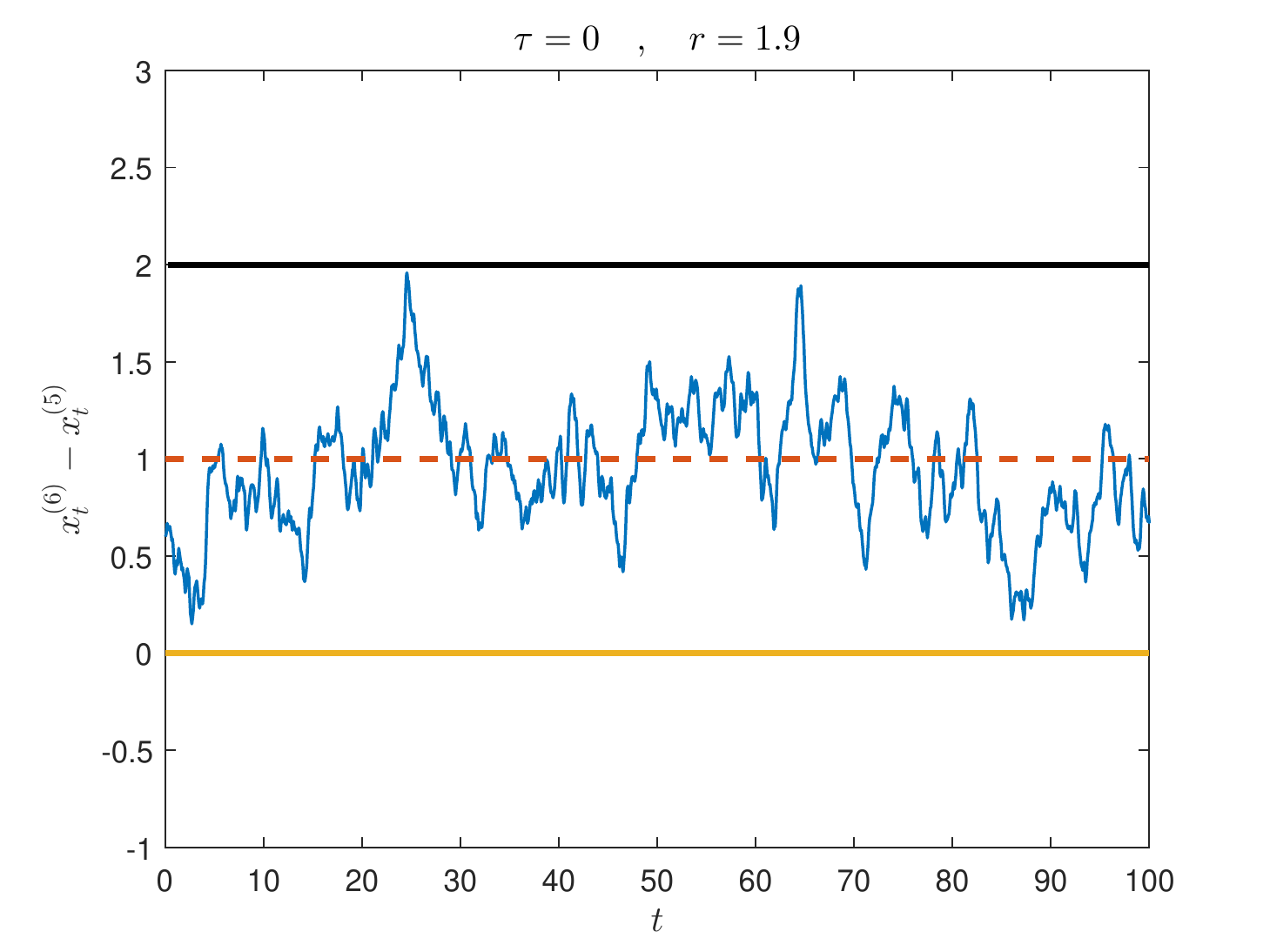} 
\includegraphics[scale=0.39]{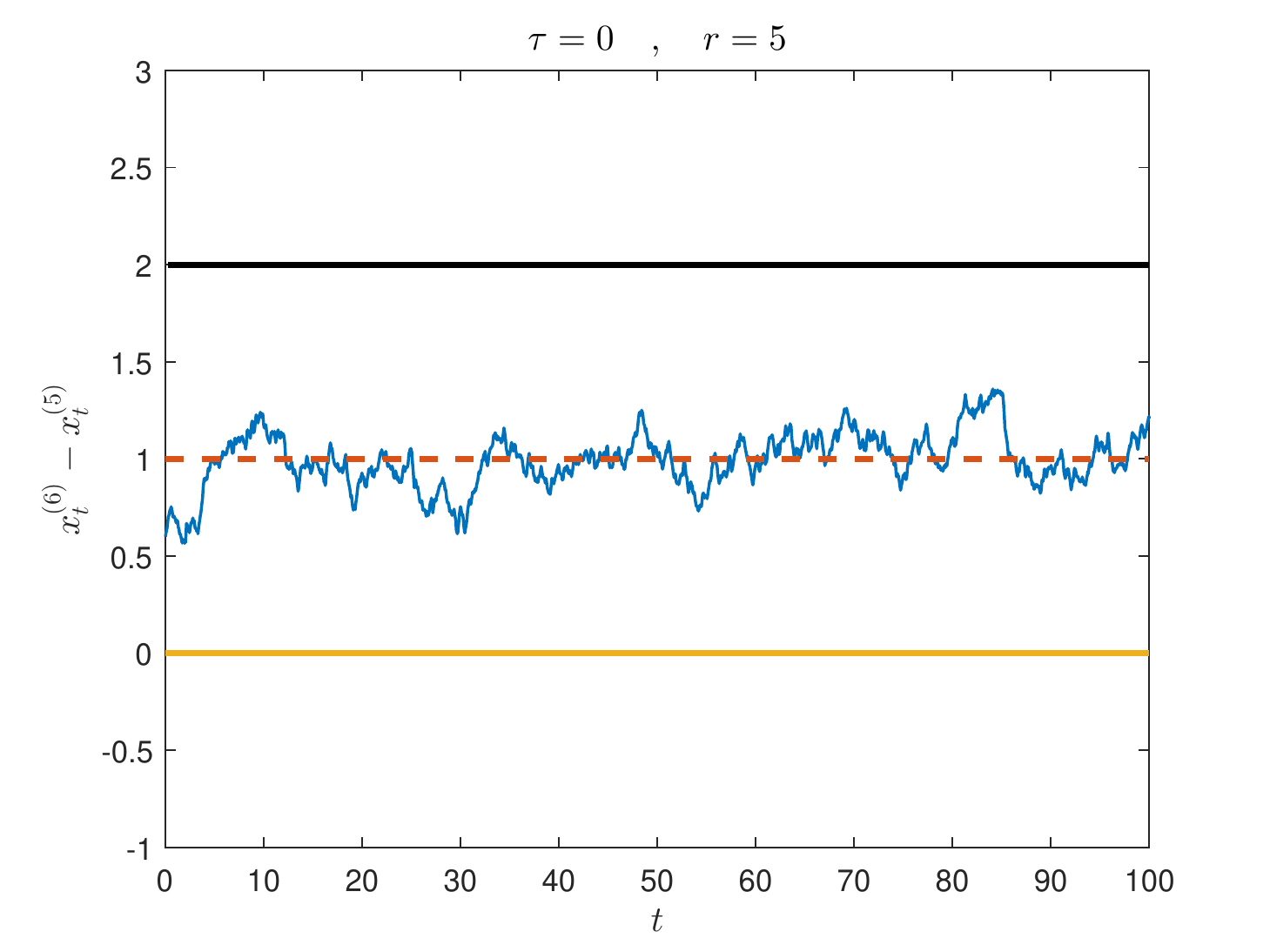}\\ 
\includegraphics[scale=0.39]{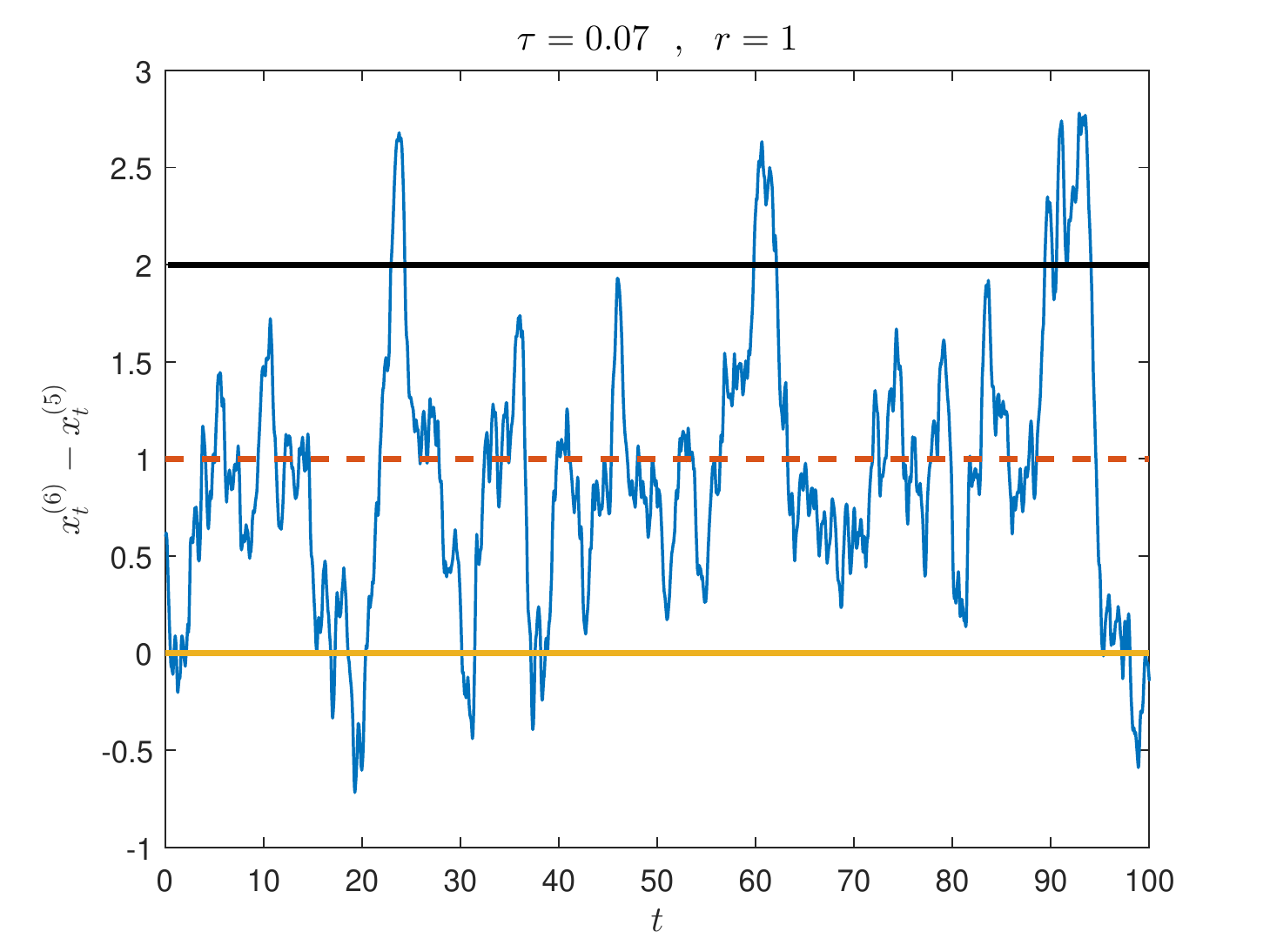}
\includegraphics[scale=0.39]{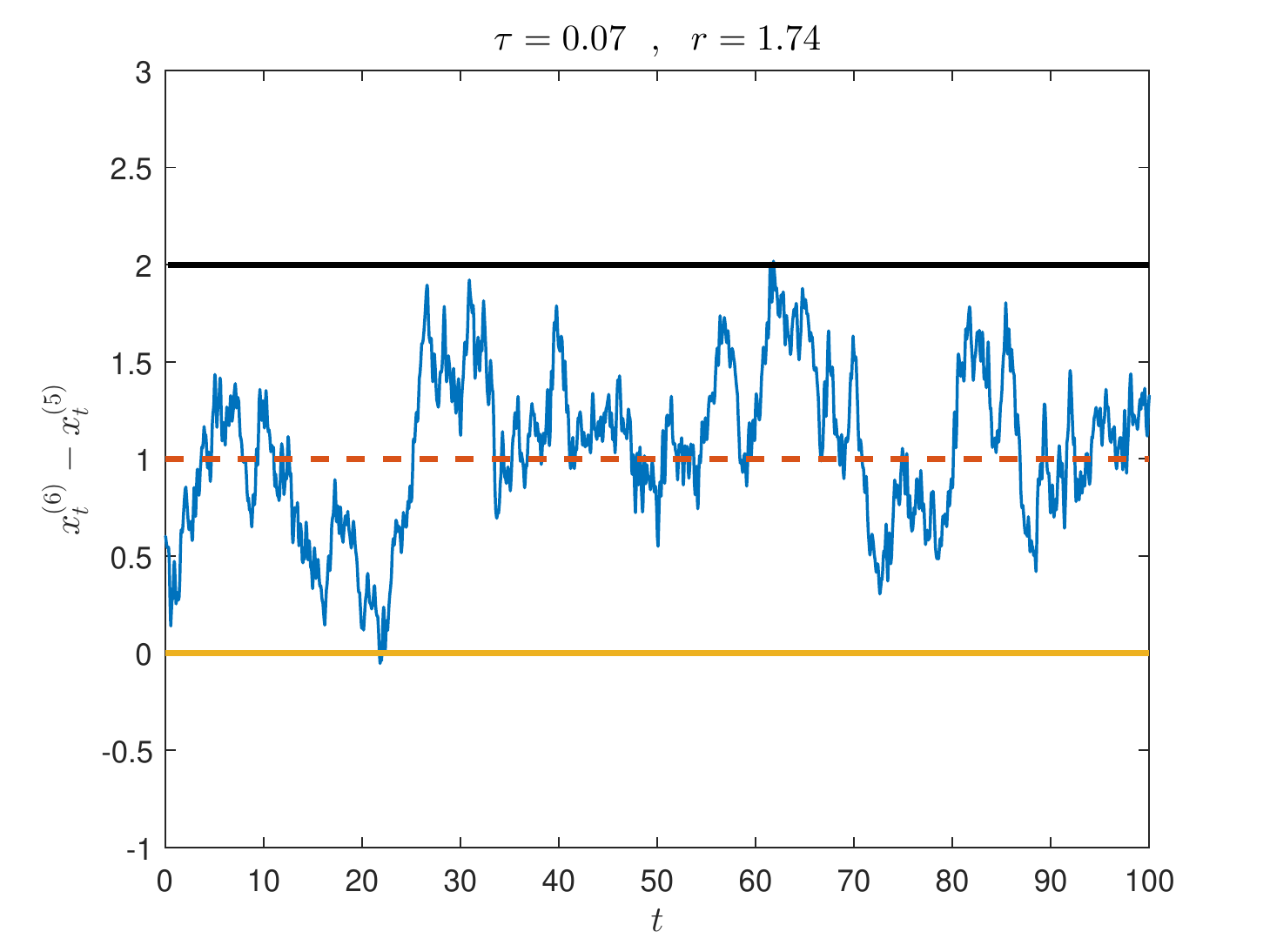}
\includegraphics[scale=0.39]{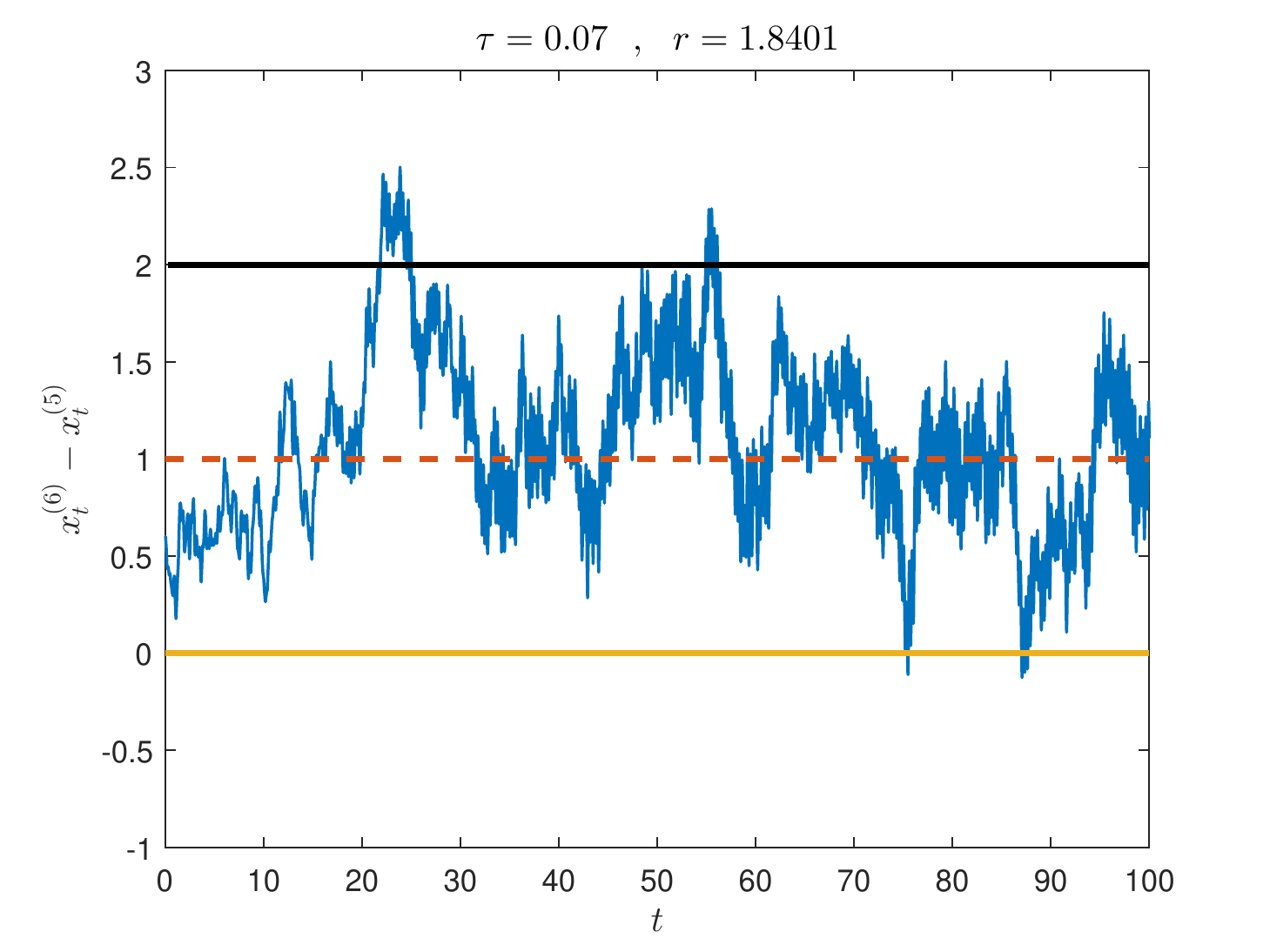}
\caption{ Simulation of Example \ref{example1}. The initial distance of 0.6 space units between vehicles 6 and 5 converges to a fluctuating distance around the desired value exponentially fast. In the absence of time delay, when $r$ increases (i.e., $\Xi_{\mathcal G}$ decreases), the risk measure improves monotonically with the network connectivity. The situation alters with non-negligible time delay, as the second row of figures illustrates.}\label{fig: simulation1}
\end{figure*}

\section*{Appendix B: Proofs}

\begin{proof}[Proof of Theorem \ref{thm: main0}]
The stability problem of \eqref{eq: sys0} is directly related to the stability of \eqref{eq: sys2}, that for $g=0$, decouples to
\begin{equation}\label{eq: sys3}
dz_t^{(i)}=\upsilon_t^{(i)} d_t,\hspace{0.2in}d\upsilon_t^{(i)}=-\lambda_i \upsilon^{(i)}_{t-\tau}\,dt-\beta \lambda_i z^{(i)}_{t-\tau}\,dt,
\end{equation} for $i=1,\dots,n$. We can now study the dynamics of \eqref{eq: sys0}, by looking at \eqref{eq: sys3} for every $i=1,\dots,n$, independently. The stability of the $i^{th}$ sub-system \eqref{eq: sys3} is characterized through the location of the roots of the characteristic function $$\Delta_i(s)=s^2+\lambda_i(\beta+s) e^{-\tau s}.$$
For $\tau>0$ we apply Theorem 13.12 in \cite{bellman1963differential}, to conclude that all roots of $\Delta_i(s)$ lie in the left complex half plane iff the conditions on $\lambda_i$ and $\beta$ hold. Standard results in the theory of delay differential equations \cite{lunel93} assure that as $t\rightarrow \infty$
$$ z_t^{(i)}\rightarrow 0 \hspace{0.1in} \text{and} \hspace{0.1in} \upsilon_t^{(i)}\rightarrow 0  $$
exponentially fast for all $i>1$. Inverting the transformation, we observe that in the limit 
$\mathbf v(t)=Q\boldsymbol \upsilon(t) \rightarrow Q [\upsilon_0,0,\dots,0]^T = Q [\frac{1}{\sqrt{n}}\sum_{i=1}^n v_0^{(i)},0,\dots,0]^T$
 which in turn implies $\lim_{t\rightarrow +\infty}v^{(i)}_t= \frac{1}{n}\sum_{i=0}^n v_0^{(i)}$. On the other hand, $x_t^{(i)}-x_t^{(j)}=d_i-d_j+\mathbf \sum_{l=1}^n(q_{il}-q_{jl}\big)z_t^{(l)}=d_i-d_j+\mathbf \sum_{l=2}^n(q_{il}-q_{jl}\big)z_t^{(l)}$ because $q_{i1}=q_{j1}$. So as $t\rightarrow \infty$, $z_t^{(l)}\rightarrow 0$ for $l>1$ and thus $x_t^{(i)}-x_t^{(j)}\rightarrow d_i-d_j=(i-j)d$.
\end{proof}

\begin{proof}[Proof of Theorem \ref{thm: main1}] The spatial-related steady-state statistics of the $i^{th}$ sub-system $\overline{z}^{(i)}$, are obtained from the marginal distribution of 
\begin{equation}\label{eq: statistics1}
\overline{z}^{(i)} \sim \mathcal N\bigg(0,g^2 \mathbf e_1^T\int_0^{\infty}\Phi_i(s)\begin{bmatrix}
0 &  0\\
0 & 1 
\end{bmatrix}\Phi_i^T(s)\,ds\,\mathbf e_1\bigg)
\end{equation}
where $\mathbf e_1=[1,0]^T$. Using Parseval's identity we obtain a more convenient integral representation: 
\begin{equation*}
\sigma_i^2=\frac{1}{2\pi}\int_{\mathbb R}\text{tr}\big[G^H(j\omega) G(j\omega)\big]\,d\omega
\end{equation*}
where $j^2=-1$, $G_i(j\omega)$ is the input-output transfer function: 
\begin{equation*}
G_i(j\omega)=g \,\mathbf e^T_1 \bigg(j\omega I_{2\times 2}- \begin{bmatrix}
0 & 1 \\
0 & 0
\end{bmatrix}- \begin{bmatrix}
0 & 0 \\
-\lambda_i\beta & -\lambda_i
\end{bmatrix}e^{-j\omega\tau}\bigg)^{-1}B_i
\end{equation*} and $G^H_i$ the complex conjugate of $G_i$. Elementary algebra yields
\begin{equation*}\begin{split}
\text{tr}\big[G^H(j\omega) G(j\omega)\big]&=\frac{g^2}{|\Delta_i(j\omega)|^2}\\
&=\frac{g^2}{(\lambda_i\beta-\omega^2 \cos(\omega\tau))^2+\omega^2(\lambda_i-\omega \sin(\omega\tau))^2}
\end{split}
\end{equation*}
and so\begin{equation*}
\sigma_i^2 = \frac{g^2}{2\pi}\int_{\mathbb R} \frac{d\omega}{(\lambda_i\beta-\omega^2 \cos(\omega\tau))^2+\omega^2(\lambda_i-\omega \sin(\omega\tau))^2} 
\end{equation*}
Changing variables to $\overline{\omega}= \omega \tau$, we obtain
\begin{equation*}\begin{split}
\sigma_i^2  = \frac{g^2\tau^3}{2\pi}\int_{\mathbb R} \frac{d\overline{\omega}}{\big((\lambda_i\tau)(\beta\tau)-\overline{\omega}^2 \cos(\overline{\omega})\big)^2+\overline{\omega}^2(\lambda_i\tau-\overline{\omega} \sin(\overline{\omega}))^2}.
\end{split} 
\end{equation*}
\end{proof}

\begin{proof}[Proof of Proposition \ref{prop: monotonicityofrisk}] Assume $\varepsilon_2<\varepsilon_1$. Then $\mathcal R_{\varepsilon_2}=\delta_2$, where $\delta_2$ is the solution of the optimization problem
\begin{eqnarray*}
& & \hspace{-0cm}\underset{\delta}{\textrm{minimize}}  ~~~~~ \delta \\
& & \hspace{-0cm} \mbox{subject to:} ~~~ \mathbb P\left\{\Sp y \in U_{\delta}\Sp \right\} < \varepsilon_2.
\end{eqnarray*} Since $\varepsilon_2<\varepsilon_1$, $\delta_2$ is a feasible solution of the optimization problem
\begin{eqnarray*}
& & \hspace{-0cm}\underset{\delta}{\textrm{minimize}}  ~~~~~ \delta \\
& & \hspace{-0cm} \mbox{subject to:} ~~~ \mathbb P\left\{\Sp y \in U_{\delta}\Sp \right\} < \varepsilon_1.
\end{eqnarray*} Therefore $\mathcal R_{\varepsilon_1}<\mathcal R_{\varepsilon_2}$. Now, let $\mathcal R_{\varepsilon_1}<\mathcal R_{\varepsilon_2}$ and label $\mathcal R_{\varepsilon_i}=\delta_i$ for $i=1,2$. By construction of $\{U_{\delta}\}$, $\mathbb P(\{y\in U_{\delta_i}\})=\varepsilon_i$ for $i=1,2$. Now,  $\delta_1<\delta_2$  implies  $U_{\delta_2}\subset U_{\delta_1}$. Then $\{y\in U_{\delta_2}\}\subset\{y\in U_{\delta_1}\}$ and thus  $\mathbb P\big(y\in U_{\delta_2}\big)< \mathbb P\big(y\in U_{\delta_1}\big)$, concluding the proof.
\end{proof}
\begin{proof}[Proof of Theorem \ref{cor: main1}]
From Theorem \ref{thm: main1} we have shown that the $2n\times n$ vector $\boldsymbol \psi:=\big[z^{(1)}_t,\upsilon_t^{(1)},\dots,z^{(n)}_t,\upsilon_t^{(n)}\big]^T$ is normally distributed with covariance matrix
\begin{equation*}\begin{split}
g^2\int_0^t \begin{bmatrix}
\Phi_1(t-s)B_1 \\ \vdots \\
 \Phi_n(t-s)B_n
\end{bmatrix}\begin{bmatrix}
B_1^T\Phi_1^T(t-s)~|~ \dots~|~ B_n^T \Phi_n^T(t-s)
\end{bmatrix}\,ds=\\
=g^2 \int_0^{t} \text{diag}\big\{\Phi_i(t-s)B_i B_i^T \Phi_i^T(t-s)\big\}\,ds
\end{split}
\end{equation*} by virtue of the form of $2\times n$ matrices $B_i$.  We introduce the affine  transformation $\mathbf z_t= E \boldsymbol \psi$ for  $E$ a $n\times 2n$ matrix with the $i^{th}$ row of to be the $1\times 2n$ canonical vector $\mathbf e_{2i-1}^T$. The covariance matrix of $\mathbf z_t$ is written as $g^2\int_0^t \tilde{\Phi}(t-s) \,ds $ for the $n\times n$ matrix 
\begin{equation*}
\tilde{\Phi}(t-s)=\text{diag} \bigg \{\begin{bmatrix}
1 & 0
\end{bmatrix}\Phi_i(t-s)\begin{bmatrix}
0 & 0 \\
0 & 1
\end{bmatrix}\Phi_i(t-s)^T \begin{bmatrix}
1 \\ 0
\end{bmatrix}\bigg\}
\end{equation*} From $\mathbf x=Q\mathbf z+\mathbf d$ and for fixed $i\geq 1$,  we can express the relative position of vehicles $i+1$ and $i$ as  $$x_t^{(i+1)}-x_t^{(i)}= [\mathbf e_{i+1}-\mathbf e_i]^T Q \mathbf{z}+d $$ In view of Assumption \ref{assum0} we have $q_{i1}\equiv q_{j1}$, we conclude that the latter is yet another affine tranformation of a multivariate normal distribution, from which we deduce that $x_t^{(i+1)}-x_t^{(i)}$ is normally distributed, with covariance matrix $$\sum_{j=2}^n g^2\big([\mathbf e_{i+1}-\mathbf e_i]^T\mathbf q_j\big)^2\int_0^t \tilde{\Phi}_{jj}(t-s) \,ds.$$
Take the limit $t\rightarrow \infty$ and apply Theorem \ref{thm: main1} to conclude.
\end{proof}
\begin{proof}[Proof of Theorem \ref{thm: main2}]
Observe that \eqref{eq: risk} is equivalent to
\begin{equation}\label{eq: auxrisk}
\inf\bigg\{\delta>0~:~\int_{-\infty}^{-\frac{d}{\sigma_i\sqrt{2}}\frac{\delta+c-1}{\delta+c}}e^{-t^2}\,dt<\varepsilon\sqrt{\pi} \bigg\}
\end{equation}
Therefore, in addition to the case $\sigma_i=0$, $\mathcal R_{\varepsilon}^{C,i}=0$ is equivalent to
\begin{equation*}\begin{split}\int_{-\infty}^{-\frac{d}{\sigma_i \sqrt{2}}\frac{c-1}{c}}e^{-t^2}\,dt<\varepsilon\sqrt{\pi}\Leftrightarrow~ \varepsilon>\frac{1-\text{erf}\big(\frac{d}{\sigma_i \sqrt{2}}\frac{c-1}{c}\big)}{2}.
\end{split}
\end{equation*} The last inequality is equivalent to the mutually exclusive cases:
\begin{equation*}
\bigg\{ \varepsilon\geq \frac{1}{2} \bigg \} ~~~\text{and}~~ \bigg\{ \varepsilon<\frac{1}{2}~~\text{and}~~\sigma_{i}\leq \frac{d}{\kappa_{\varepsilon}\sqrt{2}}\frac{c-1}{c} \bigg\}.
\end{equation*}
The union of these two cases, covers the first branch of \eqref{eq: riskcolformula}. On the other hand,
$$\mathcal R_{\varepsilon}^{C,i}=\infty~\Leftrightarrow~\int_{-\infty}^{-\frac{d}{\sigma_i\sqrt{2}}}e^{-t^2}\,dt\geq \varepsilon\sqrt{\pi}$$ which in turn is equivalent to
$\text{erf}\big(\frac{d}{\sigma_i\sqrt{2}}\big)\leq 1-2\varepsilon$, where $\text{erf}(\cdot)$ is the error function. For $1-2\varepsilon\in (0,1)$ we can invert the error function and conclude $$\mathcal R_{\varepsilon}^{C,i}=\infty~\Leftrightarrow~\sigma_i\geq \frac{d}{\sqrt{2}\text{erf}^{-1}(1-2\varepsilon)}.$$
For the third branch of $\mathcal R_{\varepsilon}^{C,i}$, we observe that the infimum of $\delta$ in \eqref{eq: auxrisk} is achieved at 
$$\int_{-\infty}^{-\frac{d}{\sigma_i\sqrt{2}}\frac{\delta+c-1}{\delta+c}}e^{-t^2}\,dt=\varepsilon\sqrt{\pi},$$ and a convenient representation of risk is $$\mathcal R_{\varepsilon}^{C,i}=\frac{\sqrt{2}\kappa_\varepsilon \sigma_i c -d (c-1)}{d-\kappa_\varepsilon \sigma_i \sqrt{2} }=\frac{d}{d-\kappa_\varepsilon \sigma_i \sqrt{2}}-c.$$
\end{proof} 
\begin{proof}[Proof of Theorem \ref{thm: riskestimates}] The proof relies on the Boole-Fr\'echet inequalities:\footnote{The proof of which can be found in \cite{halperin65}.}:  Let the collection of $\mathcal F$-measurable events $A_1,\dots,A_m$. Then,
\begin{equation}\label{eq: bf1}
\max_{i}\bigg\{\mathbb P (A_i) \bigg\} \leq \mathbb P \bigg(\bigcup_{i=1}^m A_i\bigg) \leq \min \bigg\{1, ~\sum_{i=1}^{m}\mathbb P(A_i) \bigg\},
\end{equation}
\begin{equation}\label{eq: bf2}
\begin{split}
\max\bigg\{0,~\sum_{i=1}^{m}\mathbb P(A_i)  -&(m-1)\bigg\}  \leq \mathbb P \bigg(\bigcap_{i=1}^m A_i\bigg) \leq \min_{i}\bigg\{\mathbb P (A_i) \bigg\}.
\end{split}
\end{equation}
These inequalities are the best possible probability estimates of the events $\bigcup_{i=1}^m A_i$ or $\bigcap_{i=1}^m A_i$ when nothing else is known, other than the individual probabilities $\mathbb P(A_i),~i=1,\dots,m$ . We will focus on global collision event risk as the steps on risk of global detachment are identical.
Observe that  \eqref{eq: jointcaprisk}, can be cast as the solution of the chance constraint optimization problem
\begin{eqnarray}
& &\underset{\boldsymbol \delta}{\textrm{minimize}}  ~~~~~\boldsymbol \delta\label{eq: problem1}\\
& & \mbox{subject to:} ~~~ \mathbb P\left(\Sp \bigcap_{i=1}^{n-1} \big\{ y^{(i)} \in C_{\delta_i}\big\} \Sp \right) < \varepsilon.\label{eq: problem1-2}
\end{eqnarray}
Consider the solutions $\{\delta_i^+\}_{i=1}^m$ of the scalar problems
\begin{eqnarray}
& & \hspace{0cm}\underset{\delta_i}{\textrm{minimize}}  ~~\delta_i \label{problem2}\\
& & \hspace{0cm} \mbox{subject to:} ~~ \mathbb P\big\{\Sp y^{(i)} \in C_{\delta_i}\Sp\big\} < \varepsilon,\label{problem2-1}
\end{eqnarray} 
In view of the right hand-side of \eqref{eq: bf1}, $\boldsymbol\delta^+=\big(\delta_1^+,\dots,\delta_{n-1}^+\big)$ is, in fact, a feasible solution of \eqref{eq: problem1}-\eqref{eq: problem1-2} :
\begin{equation*}\begin{split}
&\mathbb P\bigg(\bigcap_{i=1}^{n-1} \big\{ y^{(i)} \in U_{\delta_i^+}\big\} \bigg)\leq \min_{i}\bigg\{ \mathbb P\big\{  y^{(i)} \in U_{\delta_i^+}\big\}  \bigg\}\leq \varepsilon.
\end{split}
\end{equation*} This establishes the upper bound of $\mathbb V$. The lower bound of set $\mathbb V$ is trivial. 
%
%
%
We proceed with the second pair, for which we remark that is the solution of the following constraint optimization problem
\begin{eqnarray}
& & \hspace{0cm}\underset{\boldsymbol \delta}{\textrm{minimize}}  ~~\boldsymbol \delta\label{eq: problem3}\\
& & \hspace{0cm} \mbox{subject to:} ~~ \mathbb P\left(\Sp \bigcup_{i=1}^{n-1} \big\{ y^{(i)} \in C_{\delta_i}\big\} \Sp \right) < \varepsilon.\label{eq: problem3-1}
\end{eqnarray}
On the other hand, we can consider the solutions $\delta_1^*,\dots,\delta_{n-1}^*$ of the scalar problems
\begin{eqnarray*}
& & \hspace{0cm}\underset{\delta_i}{\textrm{minimize}}  ~~\delta_i \label{problem4}\\
& & \hspace{0cm} \mbox{subject to:} ~~ \mathbb P\big\{\Sp y^{(i)} \in U_{\delta_i} \Sp \big\} < \varepsilon_i ,\label{problem4-1}
\end{eqnarray*} for $i=1,\dots,n-1$, and $\varepsilon_i \in (0,1)$ such that $\sum_{i=1}^{n-1}\varepsilon_i =\varepsilon$.  Using \eqref{eq: bf1} we can show that $\boldsymbol \delta^*$ is a feasible solution of \eqref{eq: problem3}-\eqref{eq: problem3-1}. Indeed,
\begin{equation*}\begin{split}
 \mathbb P\left(\Sp \bigcup_{i=1}^{n-1} \big\{ y^{(i)} \in U_{\delta_i^*}\big\} \Sp \right)&\leq  \min\bigg\{1, \sum_{i=1}^{n-1}\mathbb P\big\{y^{(i)} \in U_{\delta_i^*} \big\}\bigg\}\\& < \min\bigg\{1,\sum_{i=1}^n \varepsilon_i\bigg\}=\min\{1,\varepsilon\}=\varepsilon.
 \end{split}
\end{equation*}. On the other hand, if $\boldsymbol{\delta}^\dagger$ is the optimal solution of \eqref{eq: problem3}-\eqref{eq: problem3-1}, then 
$$ \mathbb P\big\{\Sp  y^{(i)} \in U_{\delta_i^\dagger} \big\} <\varepsilon, ~ ~i=1,\dots,n-1,$$ by virtue of the left-hand side of \eqref{eq: bf1}. Consequently, $\delta_{k}^\dagger$ solve 
\begin{eqnarray*}
& & \hspace{0cm}\underset{\delta_k}{\textrm{minimize}}  ~~\delta_k \label{problem5}\\
& & \hspace{0cm} \mbox{subject to:} ~~ \mathbb P\big\{\Sp y^{(i)} \in U_{\delta_i} \Sp \big\} < \varepsilon\label{problem5-1}
\end{eqnarray*} concluding the proof of the theorem. 
\end{proof}

\begin{proof}[Proof of Lemma \ref{prop: limit}]
At first, observe that $f(s_1,s_2)\geq 0$ for $(s_1,s_2)\in S$ attains a minimum in the interior of $S$. This is because $S$ attains a compact closure, and $f$ diverges on its boundary. By virtue of continuity $f$ attains a minimum in its interior, let $\underline{f}$ be this minimum i.e., $\underline{f}=\inf_{(s_1, s_2)\in S}f(s_1,s_2)>0$. This is a value that can be numerically approximated $\underline{f}\approx 25.4603$. Then for $\sigma_i^2$ as in Theorem \ref{cor: main1}, we calculate:
\begin{equation*}\begin{split}
\sigma^2_i&=g^2\frac{\tau^3}{2\pi}\sum_{j=2}^n\big([\mathbf e_{i+1}-\mathbf e_i]^T\mathbf q_j\big)^2 f\big(\lambda_j\tau,\beta\tau\big)\\
&\geq g^2 \frac{\tau^3 \cdot \underline{f}}{2\pi}\sum_{j=2}^n\big([\mathbf e_{i+1}-\mathbf e_i]^T\mathbf q_j\big)^2\\&=\frac{g^2\,\tau^3\,\underline{f}}{2\pi}\,||\mathbf e_{i+1}-\mathbf e_i||^2= \frac{1}{\pi}\, \underline{f} \cdot g^2 \tau^3.
\end{split}
\end{equation*} where $\sum_{j=2}^n\big([\mathbf e_{i+1}-\mathbf e_i]^T\mathbf q_j\big)^2 =\sum_{j=1}^n\big([\mathbf e_{i+1}-\mathbf e_i]^T\mathbf q_j\big)^2=||\mathbf e_{i+1}-\mathbf e_i||^2$ due to the orthogonality of the vectors $\{\mathbf q_i\}_{i=1}^{n}$.
\end{proof}
\begin{proof}[Proof of Theorem \ref{cor: limit1}] The result follows directly after combining Theorems \ref{thm: main1}, \ref{thm: main2}, and Lemma \ref{prop: limit}.
\end{proof}
\begin{proof}[Proof of Theorem \ref{prop: effectiveresistancelimit}]
The proof directly follows from the eigenvalue restrictions dictated by the region $S$. Let $\lambda \tau\in (0,\pi/2)$. Then $\beta\tau < f(a)$ for $a$ to satisfy $g(a)=\lambda\tau$. In other words,

$$ \beta \tau < (\beta\tau)^*= f\circ g^{-1}(\lambda \tau ).$$ Inverting the last equality we obtain the limit $\lambda \tau = g \circ f^{-1}((\beta\tau)^*)$. For given $\beta,\tau$ such that $\beta\tau\in (0,1)$ the limits are reversed so as to $\lambda\tau < (\lambda\tau)^*=g \circ f^{-1}((\beta\tau))$. The last condition has to be satisfied for all non-zero eigenvalues of $L$. The result then follows directly by elementary algebra.
\end{proof}
\begin{proof}[Proof of Theorem \ref{thm: main3}] We will work the details for the systemic risk of vehicle collision only. Throughout  the proof $\mathcal R_{\varepsilon}=\mathcal R^{C,i}_{\varepsilon}$ and $\sigma_i=\sigma$ for notation simplicity.  The steps to derive of the second trade-off condition is identical.  From Theorem \ref{thm: main2} we have
\begin{equation*}
\begin{split}
\mathcal R_{\varepsilon}^2&=\bigg(\frac{d-c(d-\sqrt{2}\kappa_\varepsilon \sigma)}{d-\sqrt{2}\kappa_\varepsilon \sigma}\bigg)^2=\bigg(\frac{(1-c)+c\frac{\sqrt{2}\kappa_\varepsilon\,\sigma}{d}}{1-\frac{\sqrt{2}\kappa_\varepsilon\,\sigma}{d}}\bigg)^2\\
&=\frac{(1-c)^2+2(1-c)c\frac{\sqrt{2}\kappa_\varepsilon\,\sigma}{d}+c^2\big(\frac{\sqrt{2}\kappa_\varepsilon\,\sigma}{d}\big)^2}{\big(1-\frac{\sqrt{2}\kappa_\varepsilon\,\sigma}{d}\big)^2}=:\frac{I_{\text{e}}(\sigma)}{I_{\text{d}}(\sigma)}
\end{split}
\end{equation*} Both the enumerator $I_e$ and the denominator $I_d$ are positive. We will establish non-trivial lower bounds for the two expressions independently. We start with the enumerator $I_e(\sigma)$. It is easy to see that its minimum is achieved at $\underline{\sigma}=\frac{c-1}{c}\frac{d}{\sqrt{2}\kappa_{\varepsilon}}$, and it is achievable provided that $\underline{\sigma}\geq \sigma^*$ the hard limit from Theorem \ref{cor: limit1}. Consequently, $
I_{\text{e}}\geq \underline{E}^C $.  We proceed with the term $1/I_d$. Since  $|\frac{\kappa_{\varepsilon}\sigma\sqrt{2}}{d}|<1$ we can express $\frac{1}{I_d}$  as a geometric sum
\begin{equation*}\begin{split}
\frac{1}{I_d(\sigma)}&=\sum_{m=1}^\infty m \bigg(\frac{\sqrt{2}\kappa_{\varepsilon}\sigma}{d}\bigg)^{m-1}=1+\sum_{m=1}^\infty(m+1)\bigg(\frac{\sqrt{2}\kappa_{\varepsilon}\sigma}{d}\bigg)^{m}\\
&=1+\sum_{m=1}^\infty(m+1)\bigg(\frac{\sqrt{2}\kappa_{\varepsilon}}{d}\bigg)^{m}\big(\sigma^2\big)^{\frac{m}{2}}
\end{split}
\end{equation*}
Now,
\begin{equation*}
\frac{1}{I_d(\sigma)}\Xi_{\mathcal G}=\Xi_{\mathcal G}+\sum_{m=1}^\infty(m+1)\bigg(\frac{\sqrt{2}\kappa_{\varepsilon}}{d}\bigg)^{m}\big(\sigma^2\, \Xi_{\mathcal G}^{\frac{2}{m}} \big)^{\frac{m}{2}}.
\end{equation*} The first term of the sum above is bounded by $\Xi_{\mathcal G}>2n(n-1)\frac{\tau}{\pi}$. The $m^{th}$ term of the second infinite sum term is handled as follows: 
\begin{equation*}\begin{split}
\sigma^2 \cdot \Xi_{\mathcal G}^\frac{2}{m}&= g^2 \frac{\tau^3}{2\pi}\sum_{j=2}^n (w_j)^2f(\lambda_j\tau,\beta\tau)\cdot \Xi_{\mathcal G}^\frac{2}{m}\\&=g^2 \frac{\tau^3}{2\pi} (n\tau)^{\frac{2}{m}}\sum_{j=2}^n\bigg[(w_j)^2 f(\lambda_j\tau,\beta\tau) \bigg(\sum_{k=2}^{n}\frac{1}{\lambda_k\tau} \bigg)^{\frac{2}{m}}\bigg]
\end{split}
\end{equation*}
From the eigenvalue ordering \eqref{eq: laplacian} and the conditions of Theorem \ref{thm: main0}, we deduce for $j=2,\dots,n$, $
\sum_{j^{'}=2}^n \frac{1}{\lambda_{j^{'}}\tau}> \frac{(j-1)}{\lambda_j\tau}+\frac{(n-j)}{\vartheta(\beta\tau)}$. Then
\begin{equation*}
\bigg(\sigma^2 \cdot \Xi_{\mathcal G}^\frac{2}{m}\bigg)^{\frac{m}{2}}> n\tau  \bigg(g^2 \cdot \frac{\tau^3}{2\pi}\bigg)^{\frac{m}{2}}  \sqrt{\big(2 \underline{f}_m\big)^m}
\end{equation*}
Summing over all the terms, we arrive at the lower bound for  $\mathcal R_{\varepsilon}^2\cdot \Xi_{\mathcal G}$
\begin{equation*}
\mathcal R_{\varepsilon}^2\cdot \Xi_{\mathcal G} > n\tau \underline{E}^C \bigg( \frac{2(n-1)}{\pi}+ \sum_{m=1}^\infty 2^{\frac{m}{2}}(m+1)\bigg(\frac{|g|\tau^\frac{3}{2}\kappa_{\varepsilon}}{d\sqrt{\pi}}\bigg)^{m} \underline{f}_m^\frac{m}{2} \bigg).
\end{equation*} We take the square root on both sides to conclude.
\end{proof}
\bibliographystyle{ieeetran}    
\bibliography{bibliography}
\vspace{-0.9cm} 
\begin{IEEEbiography}{Christoforos Somarakis} received the B.S. degree in
Electrical Engineering from the National Technical
University of Athens, Athens, Greece, in 2007 and and the M.S. and Ph.D. degrees in applied mathematics from the University of Maryland
at College Park, in 2012 and 2015, respectively. He was a Post-Doctoral scholar and a Research Scientist with the Department of Mechanical Engineering and Mechanics at Lehigh University from 2016 to 2019. He is currently member of research stuff with the System Sciences Lab at Palo Alto Research Center. 
\end{IEEEbiography}
\vspace{-0.9cm} 
\begin{IEEEbiography}{Yaser Ghaedsharaf} received his B.Sc. degree in Mechanical Engineering from Sharif University of Technology in 2013. He is currently pursuing a Ph.D. in the Department of Mechanical Engineering and Mechanics at Lehigh University. He is the Runner-Up for NecSys 2016 Best Student Paper Award. His research interests include analysis and optimal design of networked control systems with applications in distributed control and cyber-physical systems.
\end{IEEEbiography}
\vspace{-0.9cm} 
\begin{IEEEbiography}{Nader Motee}(S’99–M’08–SM’13) received the B.Sc. degree in electrical engineering from the Sharif University of Technology, in 2000, and the M.Sc. and Ph.D. degrees in electrical and systems engineering from the University of Pennsylvania, Philadelphia, PA, USA, in 2006 and 2007, respectively. From 2008 to 2011, he was a Postdoctoral Scholar with the Control and Dynamical Sys tems Department, California Institute of Tech- nology, Pasadena, CA, USA. He is currently an Associate Professor with the Department of Mechanical Engineering and Mechanics, Lehigh University, Bethlehem, PA, USA. His current research area is distributed control systems and real-time robot perception. Dr. Motee is a past recipient of several awards including the 2019 Best SIAM Journal of Control and Optimization (SICON) Paper Prize, the 2008 AACC Hugo Schuck Best Paper Award, the 2007 ACC Best Student Paper Award, the 2008 Joseph and Rosaline Wolf Best Thesis Award, the 2013 Air Force Office of Scientific Research Young Investigator Program (AFOSR YIP) award,  2015 NSF Faculty Early Career Development (CAREER) award, and a 2016 Office of Naval Research Young Investigator Program (ONR YIP) award. 
\end{IEEEbiography}
\end{document}